\documentclass[10pt]{article}
\usepackage[a4paper,margin=25mm]{geometry}
\usepackage[T1]{fontenc}
\usepackage{lmodern,mathrsfs,xcolor}
\definecolor{circuitblue}{RGB}{43,70,92}
\definecolor{circuitpale}{RGB}{234,241,246}
\usepackage{tikz}\usetikzlibrary{positioning,calc,fit,backgrounds,decorations.pathreplacing}
\usepackage{xurl,needspace,float}
\usepackage{microtype,amsmath,amssymb,amsthm,mathtools,booktabs,array,longtable,enumitem}
\usepackage[hidelinks,pdftitle={Shallow Quantum Circuits for Geometrically Local Lindblad Dynamics}]{hyperref}
\usepackage{authblk}

\usepackage[backend=bibtex,style=alphabetic,natbib=true,sorting=nyt,
  maxbibnames=99,maxalphanames=4,giveninits=true]{biblatex}
\hypersetup{pdfauthor={Zheyu Shen, Yusen Wu, Xiao Yuan, Yukun Zhang}}
\allowdisplaybreaks[2]
\newtheorem{theorem}{Theorem}[section]
\newtheorem{lemma}[theorem]{Lemma}
\newtheorem{proposition}[theorem]{Proposition}
\newtheorem{corollary}[theorem]{Corollary}
\theoremstyle{definition}
\theoremstyle{remark}
\newcommand{\ket}[1]{\lvert#1\rangle}\newcommand{\bra}[1]{\langle#1\rvert}
\newcommand{\norm}[1]{\left\lVert#1\right\rVert}
\newcommand{\Tr}{\operatorname{Tr}}
\newcommand{\Enc}{\operatorname{Enc}}\newcommand{\ad}{\operatorname{ad}}
\newcommand{\polylog}{\operatorname{polylog}}\newcommand{\eps}{\epsilon}
\newcommand{\pub}{\mathsf{pub}}\newcommand{\priv}{\mathsf{priv}}
\title{Shallow Quantum Circuits for \\ Geometrically Local Lindblad Dynamics}
\newcommand{\pku}{Center on Frontiers of Computing Studies, School of Computer Science, Peking University, Beijing 100871, China}
\newcommand{\bnu}{School of Artificial Intelligence, Beijing Normal University, Beijing 100875, China}
\newcommand{\oxford}{Mathematical Institute, University of Oxford, Woodstock Road, Oxford OX2 6GG, United Kingdom}
\newcommand{\pkuphy}{School of Physics, Peking University, Beijing 100871, China}
\author[1,2]{Zheyu Shen\thanks{\texttt{zheyushen@stu.pku.edu.cn}}}
\author[3]{Yusen Wu\thanks{\texttt{yusen.wu@bnu.edu.cn}}}
\author[1]{Xiao Yuan\thanks{\texttt{xiaoyuan@pku.edu.cn}}}
\author[1,4]{Yukun Zhang\thanks{\texttt{yukun.zhang@maths.ox.ac.uk}}}
\affil[1]{\pku}
\affil[2]{\pkuphy}
\affil[3]{\bnu}
\affil[4]{\oxford}
\date{}
\begin{document}\maketitle

\begin{abstract}
Open-system dynamics govern decoherence, transport, and engineered dissipation in quantum many-body systems, making their faithful simulation a basic task for quantum computation. We propose a quantum algorithm for simulating time-independent, finite-range Lindblad dynamics with bounded local interaction strength on a constant-dimensional lattice of $n$ sites, allowing noncommuting jump operators. For evolution time $T\ge1$ and diamond-norm error at most $\epsilon$, the algorithm has circuit depth $O(T\operatorname{polylog}(nT/\epsilon))$ and gate count $O(nT\operatorname{polylog}(nT/\epsilon))$. These results extend the spatial decomposition of Ref.~\cite{HHKL} to open-system dynamics. The main obstacle is that the reversible overlap corrections used in spatial decompositions of unitary dynamics have no direct analogue for dissipative channels. We overcome this by retaining the environment in a shared unitary dilation, so that information lost from the system remains coherently available and the required spatial corrections become reversible. We prove exponential occupation bounds for the environmental states generated by the decomposition, allowing these states to be represented compactly and the local factors to be simulated coherently on the inputs that actually arise. The resulting construction reduces each short-time evolution to a constant number of parallel layers of operations on regions of logarithmic diameter, using nearest-neighbour gates and only polylogarithmic auxiliary density.

We also prove unconditional worst-case lower bounds using time-independent, purely dissipative processes that transport quantum information. For circuits with local quantum and classical communication, the lower bounds are $\Omega(nT)$ for the expected gate count and $\Omega(T)$ for the circuit depth on time scales up to the lattice diameter. They match the leading time and volume dependence of the upper bounds up to polylogarithmic factors. At fixed time and sufficiently small fixed error, we further prove a depth lower bound of $\Omega(\log n/\log\log n)$, showing that global accuracy can require increasing circuit depth as the lattice grows.
\end{abstract}

\section{Introduction}
Quantum simulation reproduces many-body dynamics on a quantum computer
\cite{Feynman,Lloyd,Georgescu,Altman}. We consider geometrically local
Markovian dynamics generated by a
Gorini--Kossakowski--Sudarshan--Lindblad operator
\cite{GKS,Lindblad,BreuerPetruccione}. Such dynamics also describe
protocols that use dissipation to prepare states or perform quantum
operations
\cite{Poyatos,KrausPreparation,VerstraeteDissipation,Diehl,Weimer,Barreiro,Schindler,Harrington}.
Our goal is to approximate the entire evolution channel with gate count
nearly proportional to the space--time volume and depth nearly
proportional to the evolution time. A guarantee for the complete channel
allows the output to be used in subsequent computation, including when
the input is entangled with a reference.

For lattice Hamiltonians, Haah, Hastings, Kothari and Low (HHKL) achieve
both scalings up to polylogarithmic factors \cite{HHKL}. They combine
Lieb--Robinson locality \cite{LR,Osborne} with an accurate simulator for
each spatial region. Open systems have analogous locality bounds
\cite{Poulin,BK,NVZ} and accurate Lindblad algorithms
\cite{Kliesch,ChildsLi,CleveWang,LiWang}. Combining these ingredients,
however, requires control of the environment. The spatial decomposition
uses inverse evolutions to correct overlaps, but an inverse of a
dissipative evolution need not be a physical channel on the system alone.

We establish the same leading time and volume dependence for arbitrary
static finite-range Lindblad generators with noncommuting jumps, given
by explicit local matrices. The construction uses a fixed finite gate set
on a same-dimensional nearest-neighbour lattice refinement and counts
its auxiliary positions. Complementary purely dissipative examples
require time, volume and precision costs before transport reaches the
lattice diameter.

We perform the spatial corrections in a common unitary dilation of
each short time segment. Permanent bath labels ensure that overlapping
restrictions and their adjoints act on the same environment registers.
Starting from the common vacuum, a locality estimate uniform in the
collision mesh gives a constant number of parallel layers on regions
of logarithmic diameter. Individual factors need a stronger input
guarantee: earlier corrections can populate later baths and correlate
them with the system. We control these inputs by an exponential
occupation bound, which permits a short coherent record of occupied
labels.\looseness=-1

The collision representation can contain too many periods to execute
directly. Building on quantum
transducers \cite{BJY,ChenHamiltonian,ChenDyadic,ChenLindblad,WangYe},
we represent each prescribed local unitary on both vacuum and occupied
baths. The repeating term order allows interval operations to be
evaluated from binary period addresses. The occupation bound controls
the correlated inputs created by preceding corrections. Together, these
properties make the local cost polylogarithmic in the period count and
yield the stated geometric gate and depth bounds.

\subsection{Main result}\label{sec:intro-result}
Let $\Lambda\subset\mathbb Z^D$ have $n$ sites. We fix the dimension,
the number of qubits per site, the support size and diameter, the term
incidence, and the number of jumps in a local term. The input consists
of local Hermitian matrices $H_e$ and jump matrices $L_{e\mu}$ with
rational real and imaginary parts, a rational time $T\ge0$, and a
rational error tolerance $0<\eps<1$. Its evolution is generated by
\begin{equation}\label{eq:introGKSL}
 \mathcal L(\rho):=-i[H,\rho]+
 \sum_{e,\mu}\left(L_{e\mu}\rho L_{e\mu}^{\dagger}
 -\tfrac12\{L_{e\mu}^{\dagger}L_{e\mu},\rho\}\right),
 \qquad H:=\sum_eH_e.
\end{equation}
Let $X_e$ denote the supplied support of term $e$.
Each supplied term has a computable rate $\lambda_e$ from
Section~\ref{sec:model}, within fixed constants of $\norm{H_e}$ or
$\sum_\mu\norm{L_{e\mu}}^2$ in fixed local dimension. We use
\begin{equation}\label{eq:intro-strength}
 \lambda:=\max_{x\in\Lambda}\sum_{e:x\in X_e}\lambda_e,
 \qquad a:=\lambda T,\qquad X:=\max\{2,n(1+a)/\eps\}.
\end{equation}
The support $X_e$ and hence $\lambda$ refer to the supplied decomposition;
Section~\ref{sec:model} gives the normalization and circuit model.

The hardware is a nearest-neighbour subgraph of a same-dimensional
lattice refinement. Gate count $G$ counts one- and two-qubit gates,
and $d_{\rm NN}$ is depth. Peak width $Q_{\rm peak}$ counts all
system, environment, work, padding and routing positions;
$N_{\rm fresh}$ counts individual auxiliary-qubit initializations,
including the initial preparation and each later segment refresh.

\begin{theorem}[Simulation of geometrically local Lindblad dynamics]\label{thm:main}\label{thm:informal}
Given the $D$-dimensional Lindbladian operator $\mathcal L$ of the form
given in Eq.~\eqref{eq:introGKSL} and an evolution time $T>0$, there exists an efficient classical algorithm that outputs a quantum channel $\widetilde\Phi_T$ such that 
\[
 \norm{\widetilde\Phi_T-e^{T\mathcal L}}_\diamond\le\eps.
\]
For a constant $C_D$ depending only on the fixed local parameters and
gate set, the constructed quantum channel requires computational resources
\begin{align*}
 G&\le C_Dn(1+\lambda T)(1+\log X)^{16D+4},&
 d_{\rm NN}&\le C_D(1+\lambda T)(1+\log X)^{16D+4},\\
 Q_{\rm peak}&\le C_Dn(1+\log X)^{6D},&
 N_{\rm fresh}&\le C_Dn(1+\lambda T)(1+\log X)^{6D}.
\end{align*}
\end{theorem}

Disjoint blocks contribute separately to the gate count but run in
parallel, which accounts for the factor of volume between gate count
and depth. The leading dependence on time and volume agrees with HHKL.
The logarithmic exponents are conservative and include arithmetic and
routing. Our construction uses polylogarithmic auxiliary density;
HHKL achieves constant auxiliary density for Hamiltonian simulation
\cite{HHKL}.

\paragraph{Consequences for finite-time tasks.}
For an input $\rho$ and observable $O$, the simulation bias is at most
$\eps\norm O$. If $\norm{e^{T\mathcal L}(\rho)-\sigma}_1\le\delta$,
the simulated output is within trace norm $\delta+\eps$ of $\sigma$.
These channel guarantees apply to subsequent quantum processing;
repetitions for estimating observables incur a separate sampling cost.

\paragraph{Complementary lower bounds.}
Static rational jumps with $H=0$ give complementary lower bounds.
For the lower bounds we use half-diamond error $\delta=\eps/2$
and write $E=\log(n/\delta)$. For normalized time above a fixed constant
and sufficiently small error, depth before diameter saturation is
\[
 \Omega_D\!\left(\lambda T+
       \frac{E}{\log(e+E/(\lambda T))}\right).
\]
In the same time range, the expected number of executed gates is
$\Omega_D(n\lambda T)$ even when gate locations are chosen adaptively.
For a fixed schedule, the corresponding gate lower bound is
\[
 \Omega_D\!\left(n\left[\lambda T+
 \left(\frac{E}{\log(e+E/(\lambda T)^{D+1})}\right)^{1/(D+1)}
 \right]\right)
\]
before its diameter cap. Section~\ref{sec:lower} gives the caps,
thresholds and circuit
assumptions. The depth bound requires classical records produced during
computation to propagate locally. The gate bounds allow unrestricted
classical communication. In both cases, auxiliaries are initially
separable across sites and outputs remain at their specified locations.
The three bounds may use different witness generators.

\subsection{Overview of the construction}\label{sec:overview}
We first decompose the joint system--environment evolution into local
factors, then implement each factor on the intermediate states produced
by the decomposition. The resulting circuit is summarized at the end of
this overview.

\paragraph{Spatial corrections in a common dilation.}
Each of the $O(1+\lambda T)$ segments starts with fresh collision baths
\cite{CollisionReview,Pocrnic,ChenLindblad} and retains them until its
end. Restrictions select gates from the same ordered product, and
adjoints use the same permanent bath labels. We assign each term to a
coarse lattice cube and activate the cubes one colour class at a time.
Writing $U_c$ for the evolution after colour $c$ has been activated,
the correction on the buffered region $\Omega_C$ is
\[
 W_C:=U_{c,\Omega_C}U_{c-1,\Omega_C}^\dagger.
\]
The locality comparison uses the common initial bath vacuum. In the
overlap operator, every contributing bath excitation must be both
created and absorbed. Two amplitudes of order $\sqrt h$ therefore
contribute at order $h$. Section~\ref{sec:patch} combines this
cancellation with the distance from the newly activated terms to the
restriction boundary. The resulting bound permits logarithmic buffers
and a constant number of parallel layers, uniformly in the mesh.

\paragraph{The inputs required by the local solver.}
Earlier factors can populate baths assigned to later periods or left
inactive in the current factor. Each permanent label is visited only a
bounded number of times. This gives an exponential occupation tail for
the intermediate states of the localized circuit and for the preimages
needed to justify its inverse operations. These states can therefore
be represented by short coherent lists of occupied labels and values.
We partition each of the two forward evolutions defining $W_C$ into
windows of consecutive periods. The factor uses the old windows in
reverse order with each window inverted, followed by the new forward
windows. A transducer represents each complete forward-window unitary
on both vacuum and occupied baths. Finite reuse and polynomial filtering
convert this representation into a circuit
\cite{BJY,ChenHamiltonian,ChenLindblad}, with its error on the correlated
inputs controlled by the occupation bound.

\paragraph{Accessing periods by their binary addresses.}
Sequential composition of the local transducers defines an auxiliary
unitary $S$ with one rotation for each collision occurrence. The dyadic
factorization of \citet{ChenDyadic} replaces this sequence with two changes
of basis and one root rotation. We apply the two-sided rotation merging
of Wang and Ye~\cite[Section~4.3.1]{WangYe} within this recursion,
allowing the rotations to depend on the occupied bath configuration.
Because the term order is fixed, a vacuum interval product is
determined by the occurrence counts of its terms. The occupation list
determines the changes for nonvacuum inputs. Both basis changes can
thus be evaluated from binary addresses at cost polynomial in block
volume, address width and numerical precision. These parameters are
polylogarithmic in $n(1+\lambda T)/\eps$. Disjoint processing regions
within each colour class contribute additively to gate count and run
in parallel.

\paragraph{Circuit generation and execution.}\label{par:execution-outline}
The classical compiler fixes the following sequence before the quantum
input is supplied.
\begin{enumerate}
\item It divides the requested time into equal segments, constructs the
coloured spatial regions, and fixes a common collision mesh. It then
allocates permanent bath lists and all address, signal and work registers
with common capacities and fixed widths.
\item For each forward window, it generates the address-tree circuit for
$S$ and the coefficient preparation and controlled reuse circuits. These
form a linear-combination-of-unitaries (LCU) circuit. It appends the
complete window circuit: the LCU circuit, its adjoint, and the LCU circuit
again, with two signal reflections and the prescribed overall phase.
A factor $W_C$ is assembled from the adjoints of the old windows in
reverse order, followed by the new forward windows.
\item It assigns physical locations to the named registers and places
the complete same-colour factor circuits in parallel. It expands controls,
routes two-qubit gates while restoring their logical locations, and
synthesizes the remaining rotations into the
fixed gate set. This produces one saved segment circuit.
\end{enumerate}
At execution time, the quantum computer initializes the segment
auxiliaries, applies the saved circuit and traces out the auxiliaries.
It repeats this circuit with fresh auxiliaries until time $T$. The bath lists and any signal or
work error remain coherent between windows and factors within a segment.
The microscopic periods specify the reference and the binary addresses;
there is no execution loop over all reference collisions.

\subsection{Relation to other simulation methods}\label{sec:related}
For Hamiltonian simulation, HHKL's spatial decomposition \cite{HHKL}
can be combined with several methods for approximating time evolution.
These include sparse-access algorithms \cite{BACS}, product formulas
\cite{Strang,Suzuki,ChildsTrotter}, compressed formulas and linear
combinations of unitaries \cite{BCS2014,BCCKS2014,BCK2015}, and signal
processing or qubitization \cite{QSP,Qubitization,QSVT}. Our spatial
analysis likewise specifies which local operations must be implemented
and the system--environment inputs on which they must be accurate.

Efficient local and sparse-access Lindblad algorithms
\cite{Kliesch,ChildsLi}, nearly linear-time simulation \cite{CleveWang}
and higher-order channel expansions \cite{LiWang} supply temporal methods.
Repeated interactions \cite{CollisionReview,Pocrnic} and enlarged-space
Hamiltonians \cite{DingLiLin} give unitary representations; commutator
estimates control their discretization \cite{ChildsTrotter,LindbladCommutator}.
The two-ancilla linear-combination method estimates observables through
sampling and rescaling \cite{TwoAncillaLCS}.

Our solver builds on the quantum transducer framework \cite{BJY}.
We use Cayley representations and finite reuse \cite{ChenHamiltonian},
dyadic changes of basis \cite{ChenDyadic}, and the query-optimal
Lindblad constructions of \citet{ChenLindblad} and \citet{WangYe}.
Both \citet{ChenLindblad} and \citet[Theorem~19 and Section~4.3]{WangYe}
give efficient additional one- and two-qubit gate counts in their
block-encoding oracle models. Their compressed environment histories
exploit the chronological structure of the simulated evolution. Spatial factors can instead populate later
labels or inactive baths before a window reaches them. Our occupation
bound tracks permanent labels and includes the required inverse
preimages. The weighted Cayley identity realizes the prescribed unitary
on the full bath space. The two-sided rotation merge of Wang and
Ye~\cite[Section~4.3.1]{WangYe} gives both changes of basis; we use it
within the dyadic recursion for term- and bath-dependent rotations.
These ingredients permit finite reuse and
filtering on the inputs produced by spatial localization. We retain
the common environment throughout the segment and trace it out after
all spatial corrections have been applied.
Section~\ref{sec:main-proof} analyzes the resulting channel.

\paragraph{Other local open-system simulation results.}
For commuting Hermitian jumps, Yu et al. approximate the channel by an ensemble of unitary circuits in a model with quantum random access memory
\cite{Yu}. In their normalized one-dimensional setting, sampled circuits
use $O(nT\polylog(nT/\eps))$ queries and additional gates.
Theorem~\ref{thm:main} instead counts refined nearest-neighbour gates and
all auxiliary positions for arbitrary static finite-range jumps.

\begin{table}[H]
\centering\small
\caption{Selected simulation results and their input restrictions. \emph{Additional gates} exclude the implementation cost of
the stipulated oracles. The displayed resources in the text retain
their original access and auxiliary-space assumptions.}
\label{tab:comparison}
\begin{tabular}{>{\raggedright\arraybackslash}p{.23\linewidth}
>{\raggedright\arraybackslash}p{.28\linewidth}
>{\raggedright\arraybackslash}p{.39\linewidth}}
\toprule
Work & Task and jump assumptions & Access and circuit resources\\\midrule
Chen et al.; Wang--Ye \cite{ChenLindblad,WangYe}
& Channel simulation in their Lindblad oracle models
& Query-optimal simulation with efficient additional gates; geometric
implementation of the oracles is not included\\
Yu et al. \cite[Theorem~5]{Yu}
& Channel simulation; commuting Hermitian jumps
& QRAM access; nearly linear queries and additional gates in the normalized one-dimensional setting\\
This work
& Full channel; arbitrary static finite-range jumps
& Explicit rational local matrices; fixed-alphabet nearest-neighbour gates on a same-dimensional refinement, with every auxiliary position counted\\
\bottomrule
\end{tabular}
\end{table}

Section~\ref{sec:model} fixes the model and collision representation.
Section~\ref{sec:patch} proves locality and bounds the occupation of
the inputs to local simulation. Sections~\ref{sec:compiler} and
\ref{sec:implementation} construct and implement the solver, then
combine its estimates to prove the main theorem.
Section~\ref{sec:lower} gives the transport lower bounds. The appendices
provide the collision error, finite numerical algorithms and gate
synthesis.

\section{Preliminaries and the collision representation}\label{sec:model}
We fix the supplied local decomposition and a common dilation on the
system and environment. The reduced channel specifies the simulation
task, but spatial corrections also require the dilation's action on
occupied baths. Section~\ref{sec:implementation} gives the finite
registers used to encode these baths.

\subsection{Local input and computational model}
Let $\Lambda\subset\mathbb Z^D$ have $n\ge1$ sites and $d_0$ system
qubits per site. The dimension $D\ge1$ and upper bounds on the local
support size $\kappa\ge1$, support diameter $r_*\ge1$ in the
$\ell_\infty$ metric, and jump multiplicity $\mu_*\ge1$ are fixed.
We also fix the incidence bound
$d_{\rm inc}\ge1$ after the terms have been split and zero terms
removed, as described below. Matrix entries belong to
$\mathbb Q+i\mathbb Q$. The time $T\ge0$ and tolerance $0<\eps<1$
are rational. Exact entrywise checks verify $H_e=H_e^\dagger$, matrix
dimensions and support constraints, and reject malformed input.
All subsequent classical error budgets are positive rationals or
dyadic numbers.

Hamiltonian terms and terms containing a stacked jump map have disjoint
index sets $\mathcal I_H$ and $\mathcal I_J$. We split mixed entries
in a prescribed order; Eq.~\eqref{eq:introGKSL} then sums over the
corresponding sets. Each term has an assigned support $X_e$. The
ordered term and support lists fix the tensor factors of its matrix,
with qubits ordered by increasing index within each site. The local
gates preserve this order as specified in Appendix~\ref{app:endpoint}.
Preprocessing first checks dimensions, then removes empty-support
terms, which are scalar identities with zero generator, and exact-zero
terms. It finally verifies the incidence bound $d_{\rm inc}$. Idle
sites remain included in $n$. More explicitly, the retained decomposition
must satisfy
\[
 |X_e|\le\kappa,\qquad \operatorname{diam}_\infty X_e\le r_*,
 \qquad
 \max_{x\in\Lambda}\bigl|\{e\in\mathcal I_H\cup\mathcal I_J:x\in X_e\}\bigr|
 \le d_{\rm inc}.
\]
For $e\in\mathcal I_J$, the supplied jump components are
$L_{e\mu}$ with $1\le\mu\le\mu_e\le\mu_*$.
Thus incidence counts the split, retained terms meeting a site, rather
than their matrix entries or their individual jump components.
These bounds imply that a region of volume $v$ contains only $O(v)$
term types; their strengths may vary and are accounted for separately.

The anchor $a(e)\in X_e$ is the first support site in the input order.
At scale $R$, we assign term $e$ to the coarse cube indexed by
\[
 C_R(e):=\bigl(\lfloor a(e)_1/R\rfloor,\ldots,\lfloor a(e)_D/R\rfloor\bigr).
\]
This cube owns the environment registers of the terms assigned to it.
Euclidean division gives the same assignment for negative coordinates.
Removing unused entries preserves the physical labels and anchors.

The Frobenius square $\norm A_F^2=\Tr A^\dagger A$ is rational, so exact
comparisons determine the least dyadic scales
$\alpha_e,\beta_e\in\{2^k:k\in\mathbb Z\}$ satisfying
\[
 \alpha_e^2\ge4\norm{H_e}_F^2,\qquad
 \beta_e^2\ge4\norm{J_e}_F^2,\qquad
 J_e\psi=\sum_\mu\ket\mu L_{e\mu}\psi.
\]
For surviving nonzero terms,
$2\norm{H_e}_F\le\alpha_e<4\norm{H_e}_F$ and
$2\norm{J_e}_F\le\beta_e<4\norm{J_e}_F$.
H and J denote Hamiltonian and stacked-jump occurrences, with rates
$\lambda_e=\alpha_e$ and $\lambda_e=\beta_e^2$, respectively. They define
$\lambda,a,X$ in Eq.~\eqref{eq:intro-strength}. For $a=0$, the target is
identity and the classical compiler returns the empty gate list.

All constants depend only on the fixed local model. Reading and
validating the input, and evaluating the normalization to the required
precision, take classical time polynomial in the input length. Quantum
resource bounds count the circuit produced by this computation.

Every logical qubit has a fixed position on a lattice refinement of
the same spatial dimension. Section~\ref{sec:implementation} specifies
the environment registers and their layout. Each segment initializes
its auxiliaries, applies a unitary and then discards them. The finite
gate set $\mathcal G$ in Section~\ref{app:terminal} contains one- and
two-qubit gates and is closed under adjoints. It includes exact Boolean,
Clifford and phase gates, together with algebraic one-qubit gates.
Deterministic synthesis preserves the specified phase and uses
$O(\log^4(1/\xi))$ gates with polynomial classical bit cost.
Continuous rotations appear only in intermediate descriptions; the
output circuit uses gates from $\mathcal G$.

The gate count $G$ includes every occurrence of an elementary gate.
Nearest-neighbour depth counts layers of gates on disjoint wires. Peak
width includes all allocated system, auxiliary, padding and routing
positions. The number of fresh qubits counts every auxiliary
preparation, including those in the first segment. Preparations and
discards occur at segment boundaries and are counted separately from
the unitary gates in $G$.

\subsection{Norms, circuits, and initialized subspaces}
Registers are finite-dimensional complex Hilbert spaces with orthonormal
computational bases. Trace distance is $\tfrac12\norm{\rho-\sigma}_1$.
Logarithms are natural unless the base is indicated; binary register
widths use $\log_2$. For operators,
$\norm A$ is the induced Euclidean norm; for matrix maps we use the diamond norm
\cite[Chapter~3]{Watrous}
\[
 \norm\Phi_\diamond:=\sup_{r\ge1}\sup_{Z\ne0}
 \frac{\norm{(\Phi\otimes I_r)(Z)}_1}{\norm Z_1}.
\]
Throughout, a channel is a completely positive, trace-preserving (CPTP)
map. All estimates are uniform in the dimension of an untouched reference.
Composition and unitary products act from right to left; verbal sequences
specify chronological order. We write $[A,B]=AB-BA$,
$\operatorname{ad}_A(B)=[A,B]$, and
\[
 \mathcal H_H(Z):=-i[H,Z],\qquad
 \mathcal D_J(Z):=\sum_\mu L_\mu ZL_\mu^\dagger
                   -\tfrac12\{J^\dagger J,Z\}.
\]
Unitary conjugation is $\operatorname{Ad}_U(Z)=UZU^\dagger$.

A stored circuit lists its one- and two-qubit gates, wire indices and
intermediate rotation parameters. Its adjoint reverses the list and
replaces each gate by its adjoint. Initialization specifies the input
subspace on which accuracy is required. The same gate sequence defines
a unitary on the full register space, which is needed to propagate
errors. The validated rational input fixes all classical branches,
rounding rules, loop bounds and segment operations. The resulting
quantum channel is unconditional.

We approximate each complete segment coherently before tracing out
its environment. When two isometries use the same environment, their
operator-norm difference controls the diamond-norm error of the reduced
channels. Errors then add under channel composition
\cite[Chapters~2--3]{Watrous}, \cite[Proposition~3]{KSW}.

\begin{lemma}[Channel transfer and composition]\label{lem:samechannel}
Let $U_C$ be a unitary circuit on system $S$ and auxiliary register $A$,
and let $E_C\psi:=U_C(\psi\otimes\ket0_A)$. Its induced channel is
\begin{equation}\label{eq:circuitkraus}
 \Phi_C(Z)=\Tr_A(E_CZE_C^\dagger)
          =\sum_aK_aZK_a^\dagger,
 \qquad K_a:=(I_S\otimes\bra a_A)E_C,
 \qquad\sum_aK_a^\dagger K_a=I.
\end{equation}
If $E_C,E_D$ use the same environment space, then
\begin{equation}\label{eq:column-channel-transfer}
 \norm{\Phi_C-\Phi_D}_\diamond\le2\norm{E_C-E_D}.
\end{equation}
For any two channels $\Phi,\Psi$ and integer $m\ge1$,
\begin{equation}\label{eq:freshcomposition}
 \norm{\Phi^m-\Psi^m}_\diamond\le m\norm{\Phi-\Psi}_\diamond.
\end{equation}
\end{lemma}
\begin{proof}
Expanding $E_CZE_C^\dagger-E_DZE_D^\dagger$ into two terms gives
the factor of two, including for inputs entangled with a reference
\cite[Proposition~3]{KSW}. Contractivity of channels and a telescoping
sum give Eq.~\eqref{eq:freshcomposition}.
\end{proof}
We also call an initialized map such as $E_C$ a \emph{column}; a column
error is measured in operator norm over the entire input space.
Thus a complete segment
needs to be accurate only when its baths start in vacuum. Its individual
spatial factors must also be accurate on the occupied inputs created by
earlier factors.

\subsection{Collision unitaries and permanent environment labels}
Repeated interactions \cite{CollisionReview,Pocrnic} give a reversible
circuit for each segment. We divide $[0,T]$ into $N_0$ segments, each
of duration $\tau=T/N_0$. A segment contains $M$ periods, each of
duration $h=\tau/M$, with the same prescribed term order.
Equation~\eqref{eq:globalparam} chooses the segment length needed for
locality; Eq.~\eqref{eq:mesh} chooses a dyadic $M$ to control the discretization error. Let $r$ denote the
number of windows per segment; Section~\ref{sec:prefix} chooses it as
a power of two dividing $M$. The local solver groups $M/r$ consecutive
periods into each window without changing their order or labels. The time step
$h$ describes simulated time, not hardware depth. Baths are retained
within a segment, discarded at its end and initialized afresh for the
next segment.

\begin{table}[H]
\centering\small
\caption{Time organization. Simulated time and circuit calls have
different roles. A window is a group of complete periods inside a
spatially restricted pass.}
\label{tab:time-scales}
\begin{tabular}{>{\raggedright\arraybackslash}p{.16\linewidth}
>{\raggedright\arraybackslash}p{.37\linewidth}
>{\raggedright\arraybackslash}p{.37\linewidth}}
\toprule
Object & Meaning & Size or duration\\\midrule
Segment & Evolution before a fresh environment is discarded & $N_0$ segments, each of duration $\tau=T/N_0$\\
Period & One prescribed ordered pass over the active terms & $M$ periods per segment, time step $h=\tau/M$\\
Window & Consecutive periods simulated by one local solver call & $r$ windows per pass, each with $M/r$ periods\\
Occurrence & One term $e$ at one period $m$ & Permanent label $(e,m)$; $L$ occurrences in a window\\
Reuse slot & One call to the window transducer & $N$ slots, chosen in Section~\ref{sec:reuse}; no extra simulated time\\
\bottomrule
\end{tabular}
\end{table}

An occurrence $b=(e,m)$ applies term $e$ in period $0\le m<M$ starting
at $mh$. Jump occurrences have distinct bath factors; H labels record
chronological positions. Restriction retains a gate exactly when its term
is active and $X_e$ lies wholly in the region; excluded gates become
identity without changing positions or labels. An inverse uses the adjoint
of that same gate on the same system and bath. We write
$\lambda_b=\lambda_e$, or $\lambda_j$ at chronological position $j$.

For maximum jump multiplicity $\mu_*\ge1$, we pad with zero jumps to bath
dimension $d_E=2^{\lceil\log_2(1+\mu_*)\rceil}$. Vacuum is $\ket0$ and
the remaining values form the occupied sector. We write
$P_{b,0}=\ket0\bra0_b$, $P_b=I-P_{b,0}$, or $P_0,P_1$ for one bath.
Nonzero padded values remain valid occupied inputs; padding preserves
the generator and reduced channel. For the segment's jump labels,
\[
 \mathcal H_E:=\bigotimes_{b\in\mathcal B_J}\mathbb C^{d_E}_b,
 \qquad E_0\psi:=\psi\otimes\ket0_E.
\]
All collisions, restrictions and adjoints act on this common space.
We use \emph{physical} for operators and states on this uncompressed
system--bath space. Section~\ref{sec:implementation} introduces their
finite-register encodings.

Tracing out a bath initialized to vacuum gives a Kraus map
\cite{Stinespring,Kraus,Choi}. Rational approximations to these maps
and their transducer representations
\cite[Lemma~2.1]{ChenLindblad}, \cite[Lemma~11 and Proposition~12]{WangYe}
motivate the following Cayley gates, defined on the full bath space:
\begin{align*}
 V_{H,h}&:=(I-ihH/2)(I+ihH/2)^{-1},\\
 V_{J,h}&:=(I-i\sqrt hG_J/2)(I+i\sqrt hG_J/2)^{-1},
 \qquad G_J:=\begin{pmatrix}0&J^\dagger\\J&0\end{pmatrix}.
\end{align*}
The operator $G_J$ couples the vacuum and occupied sectors through
$J$. It specifies both the channel induced by a vacuum input and the
unitary action on occupied inputs. Hamiltonian gates act only on the
system. Every restriction and inverse uses these same gates and bath
factors.

For $h\lambda\le1$ and integer $T/h$, the fresh-bath forward product
approximates the GKSL semigroup (Lemma~\ref{lem:collision}):
\begin{equation}\label{eq:referenceaccuracy-summary}
 \bigl\|\Phi_h^{T/h}-e^{T\mathcal L}\bigr\|_\diamond
 \le\tfrac12 n\lambda^2Th
 \le C_{\rm col}n\lambda^2Th,
 \qquad C_{\rm col}=8(\kappa+1).
\end{equation}
$\Phi_h$ is one period with independent fresh jump environments.
Appendix~\ref{app:collision} proves this by ordered resolvent identities
and positive-time splitting. We retain $C_{\rm col}$ for the common mesh.
The next section localizes before the environment trace, using adjoints
of these specified unitaries as inverse factors.

\section{Spatial decomposition and environment occupation}\label{sec:patch}
We now parallelize a short collision circuit while retaining its baths.
For adjacent regions $AB$ and $BC$, an inverse on the overlap $B$
removes the evolution counted twice. We first bound this correction
and use it to construct a constant number of local layers. The spatial
comparison acts on the isometries $U_cE_0$ obtained by activating
successive term sets. The local solver instead acts on exact prefixes
of the resulting localized product. We therefore establish a separate
occupation bound for those intermediate inputs.

\subsection{The three-region comparison and its activation form}
Following \citet[Lemma~6]{HHKL}, let disjoint regions $A,B,C$ cover
$\Lambda$, with $B$ separating $A$ and $C$. Each restriction $U_X$
retains terms whose complete supports lie in $X$. All restrictions
share the permanent bath registers. We compare
\begin{equation}\label{eq:abccomparison}
 V_{ABC}:=U_{ABC}E_0,\qquad
 \widetilde V_{ABC}:=U_{AB}U_B^\dagger U_{BC}E_0.
\end{equation}
Here $ABC=A\cup B\cup C$. The circuit first evolves $BC$, then
reverses the evolution on $B$ using the same baths, and finally
evolves $AB$.

\begin{figure}[t]
\centering
\begin{tikzpicture}[x=1cm,y=1cm]
 \fill[circuitpale] (0,0) rectangle (2.4,.65);
 \fill[circuitblue!25] (2.4,0) rectangle (4.8,.65);
 \fill[circuitpale] (4.8,0) rectangle (7.2,.65);
 \draw (0,0) rectangle (7.2,.65);
 \draw (2.4,0)--(2.4,.65); \draw (4.8,0)--(4.8,.65);
 \node at (1.2,.325) {$A$}; \node at (3.6,.325) {$B$};
 \node at (6,.325) {$C$};
 \draw[thick] (2.4,-.25)--node[below]{$U_{BC}$ first}(7.2,-.25);
 \draw[thick,circuitblue] (2.4,-.9)--node[below]{$U_B^\dagger$ second}(4.8,-.9);
 \draw[thick] (0,-1.55)--node[below]{$U_{AB}$ third}(4.8,-1.55);
\end{tikzpicture}
\caption{The HHKL overlap convention: $U_{BC}$, $U_B^\dagger$, then
$U_{AB}$ act on shared retained baths. Equation~\eqref{eq:abccomparison}
compares this initialized isometry with the full collision dilation.}

\label{fig:abc}
\end{figure}
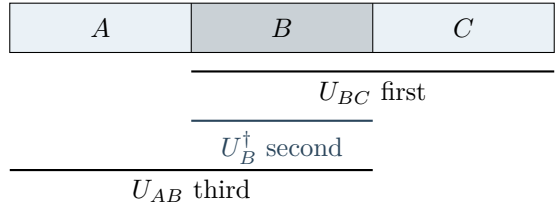

To obtain parallel layers, we will use a comparison that permits a
general set of newly activated terms. Its three-region specialization
takes the old circuit to be $U_{BC}$, the new circuit to be $U_{ABC}$,
and the restriction region to be $AB$. Whole-support restriction of
$U_{BC}$ to $AB$ is exactly $U_B$, while restriction of $U_{ABC}$ is
$U_{AB}$, so it gives Eq.~\eqref{eq:abccomparison}.

For the activation comparison, we fix one segment's chronological
collision sequence and nested active term sets
$\mathcal T_0\subseteq\mathcal T_1$. The circuits $U_0$ and $U_1$
retain precisely the occurrences of their respective active terms;
their whole-support restrictions $U_{0,\Omega}$ and $U_{1,\Omega}$
also require $X_e\subseteq\Omega$. Every excluded occurrence is
replaced by identity. All four products consequently use the same
gate, chronological position and bath label wherever an occurrence is
retained. With the vacuum initialization $E_0$ of the common segment,
the exact and corrected isometries are
\[
 V:=U_1E_0,\qquad
 \widetilde V:=U_{1,\Omega}U_{0,\Omega}^\dagger U_0E_0,
\]
and their overlap is
\begin{equation}\label{eq:gram}
 \mathcal M:=V^\dagger\widetilde V
 =E_0^\dagger U_1^\dagger U_{1,\Omega}U_{0,\Omega}^\dagger U_0E_0.
\end{equation}
The vacuum projections at both ends require every contributing bath
excitation to be absorbed again. To identify which spatial
contributions survive these projections, we use the complete occurrence set
$\mathcal B=\{(e,m):e\in\mathcal T_1,\ 0\le m<M\}$, including
formal Hamiltonian labels, and define its source and cut subsets by
\[
 \mathcal S:=\{(e,m)\in\mathcal B:e\in\mathcal T_1\setminus\mathcal T_0\},
 \qquad
 \mathcal C:=\{(e,m)\in\mathcal B:X_e\not\subseteq\Omega\}.
\]
Thus a source label belongs to a newly activated term, and a cut label
belongs to a term excluded by spatial restriction.
For Eq.~\eqref{eq:abccomparison}, a source label has support meeting
$A$, and a cut label has support meeting $C$. Thus a surviving
contribution must connect the two outer regions across $B$.
The following argument expresses this separation in the support-overlap
graph, so it also applies to the activation sets used for colour layers.

We apply inclusion--exclusion to $\mathcal B$. For
$K\subseteq\mathcal B$, let $\mathcal M(K)$ be the Gram operator obtained
by replacing every appearance of each label outside $K$ by identity.
The simultaneous replacement is essential: it preserves the
cancellations between a gate and its adjoint. Since
$\mathcal M(\varnothing)=I$, M\"obius inversion on the subset lattice
\cite{Rota} gives an expansion in the finite differences
\[
 \Delta_F\mathcal M:=\sum_{K\subseteq F}(-1)^{|F|-|K|}\mathcal M(K).
\]
The resulting expansion is
\begin{equation}\label{eq:mobius}
 \mathcal M(\mathcal B)=I+\sum_{\varnothing\ne F\subseteq\mathcal B}\Delta_F\mathcal M.
\end{equation}
Each coefficient depends only on the marked labels in $F$, since all
other labels have been removed.

The overlap graph has one vertex per permanent occurrence label,
including the formal H labels. For $b=(e,m)$, we write $X_b:=X_e$
for its system support. An edge joins $b,c$ when their system supports
overlap, irrespective of their chronological positions.

\begin{lemma}[Cancellation of noncrossing components]\label{lem:flat}
If a component of the graph induced by $F$ misses either $\mathcal S$ or
$\mathcal C$, then $\Delta_F\mathcal M=0$.
\end{lemma}
\begin{proof}
Let $K_1,\ldots,K_q$ be the connected components of the graph induced
by $F$. Their system supports are disjoint, and their bath factors have
distinct labels. After all labels outside $F$ have been removed, gates from different
components commute. The product of the four circuits therefore
factorizes over these components for every subset $K\subseteq F$.
Projection onto the product bath vacuum preserves this factorization.
Let $\mathcal M_i(K\cap K_i)$ denote the projected product on the
corresponding component. Distributing the alternating sum over subsets gives
\[
 \Delta_F\mathcal M
 =I_{\rm unused}\otimes\bigotimes_{i=1}^q
       \Delta_{K_i}\mathcal M_i.
\]
Here $I_{\rm unused}$ acts on the system factors outside the component
supports. This is an operator identity, so it also applies to entangled
inputs.

For a component missing the source, the old and new subcircuits
coincide on every submask, as do their restrictions. Writing these
circuits as $Q$ and $Q_\Omega$, their product of four circuits is
$Q^\dagger Q_\Omega Q_\Omega^\dagger Q=I$.
For a component missing the cut, restriction retains every gate; the
product is then $Q_1^\dagger Q_1Q_0^\dagger Q_0=I$.
In either case the alternating subset sum over the nonempty component
is zero. The tensor factorization therefore makes
$\Delta_F\mathcal M$ vanish whenever one component misses the source
or the cut.
\end{proof}

\begin{lemma}[Vacuum-compressed coefficient bound]\label{lem:activity}
Assume $h\lambda_b\le1$. For every marked set,
\[
 \norm{\Delta_F\mathcal M}\le\prod_{b\in F}\omega_b,
 \qquad \omega_b:=64h\lambda_b.
\]
The bound is uniform over system inputs and an untouched reference.
\end{lemma}
\begin{proof}
All appearances of a label must be bounded together. Estimating its
creation amplitudes separately would lose a factor of $\sqrt h$.
We expand each gate as $I+(V_b^{\pm1}-I)$ without changing the circuit
order. Inclusion--exclusion retains exactly those products with at
least one nonidentity difference for every marked label. Each label
appears at most four times, so its contribution is determined by a
fixed number of paths that return to vacuum.

For a jump bath, set $A_0=J^\dagger J$ and $A_1=JJ^\dagger$ on
the two sectors. The diagonal difference blocks of its Cayley gate are
$-(h/2)A_i(I+hA_i/4)^{-1}$, and the off-diagonal blocks are
$-i\sqrt h\,J(I+hA_0/4)^{-1}$ and
$-i\sqrt h\,(I+hA_0/4)^{-1}J^\dagger$.
Both resolvents are contractions. The normalization
$\norm J^2\le\lambda_b/4$ therefore gives, for either orientation,
\[
 \norm{P_i(V_b^{\pm1}-I)P_i}\le h\lambda_b,\qquad
 \norm{P_i(V_b^{\pm1}-I)P_j}\le\sqrt{h\lambda_b}\quad(i\ne j).
\]
Inserting the vacuum and occupied projections
between successive appearances therefore gives a sum over occupation
paths that start and end in vacuum. A nonconstant path must contain
either a diagonal difference or at least two changes of occupation.
With $t=\sqrt{h\lambda_b}\le1$, its total scalar weight is consequently
bounded by
\[
 \left[\begin{pmatrix}1+t^2&t\\t&1+t^2\end{pmatrix}^{4}\right]_{00}-1
 \le40t^2.
\]
The polynomial is exactly $10t^2+19t^4+10t^6+t^8$, with coefficient sum $40$. With fewer appearances the bound only decreases. The occupied block contains all bath values associated with the
stacked jump map $J$, so the estimate uses their combined operator norm.

For an H label, $\norm{V_H^{\pm1}-I}\le h\lambda_b$, and the analogous sum is at most $(1+h\lambda_b)^4-1\le15h\lambda_b$.

We must combine these estimates without assuming that system operators
commute. We move only bath projections past gates acting on other baths;
this is valid even when their system supports overlap. All system
operators remain in their original order. For the block factors
$A_1,\ldots,A_m$ of a fixed sector path, submultiplicativity gives
\[
 \norm{A_m\cdots A_1}\le\prod_{\ell=1}^m\norm{A_\ell}.
\]
We sum this scalar bound over the allowed paths. Their choices factor
over bath labels, so the scalar majorant is the product of the
one-label weights just estimated.
The constant $64$ covers both the H and J estimates. Tensoring with the
identity on a reference preserves all the operator norms above.
\end{proof}

We first sum connected sets joining the source to the restriction
boundary, then sum collections of these sets. This follows the approach
of polymer expansions \cite{KoteckyPreiss,FernandezProcacci}; here a
finite count of rooted trees suffices. We set
\[
 d_w:=\sup_b\sum_{c:\,X_c\cap X_b\ne\varnothing}\omega_c,
 \qquad W_{\mathcal S}:=\sum_{b\in\mathcal S}\omega_b.
\]
Every connected crossing set is assumed to have at least $m_*$ vertices.

\begin{lemma}[Weighted connected-set bound]\label{lem:animals}
If $6d_w<1$, the total activity of all connected sets joining the source to the restriction boundary is at most
\[
 \zeta=\frac{W_{\mathcal S}(6d_w)^{m_*-1}}{1-6d_w}.
\]
\end{lemma}
\begin{proof}
We first sum connected sets of a fixed size $k$ containing a fixed root
$b$. A depth-first search, with neighbours ordered by the fixed label
order, assigns it a rooted plane spanning tree. The labelled tree
recovers the vertex set, so this choice is an injection into labelled
plane trees. There are at most $4^{k-1}$ tree shapes, by the Catalan
count \cite[Chapter~6]{Stanley}. For a fixed shape, we relax the sum to
all labelings whose parent--child supports overlap, allowing repeated
labels. Summing a leaf label contributes at most $d_w$, uniformly in
its parent's label. Repeating this elimination leaves the root weight
$\omega_b$ and $k-1$ factors of $d_w$. Thus
\[
 \sum_{F\ni b:\,F\text{ connected},\,|F|=k}\prod_{c\in F}\omega_c
 \le\omega_b(4d_w)^{k-1}\le\omega_b(6d_w)^{k-1}.
\]
We use the weaker bound with $6d_w$ to simplify the subsequent choice
of radius. Every crossing set contains a source
label. Summing over source roots can count a set more than once, which
only enlarges the upper bound. We sum the geometric series over
$k\ge m_*$. The condition $6d_w<1$ gives the denominator in $\zeta$.
\end{proof}

Combining the connected-set bound with the cancellation of noncrossing
components gives the following estimate on the shared-environment evolution.
\begin{theorem}[Shared-environment locality]\label{thm:patch}
In the shared activation comparison of Eq.~\eqref{eq:gram}, let
$\omega_b=64h\lambda_b$, with $d_w,W_{\mathcal S}$ as above.
Suppose $h\lambda_b\le1$ for every occurrence label, $6d_w<1$, and
every connected vertex set meeting both $\mathcal S$ and $\mathcal C$
contains at least $m_*\ge1$ vertices. With
$\zeta=W_{\mathcal S}(6d_w)^{m_*-1}/(1-6d_w)$, the Gram operator
and initialized evolution satisfy
\begin{equation}\label{eq:patchbound}
 \norm{I-\mathcal M}\le e^\zeta-1,
 \qquad
 \norm{(U_1-U_{1,\Omega}U_{0,\Omega}^\dagger U_0)E_0}
 \le\sqrt{2(e^\zeta-1)}.
\end{equation}
The input domain is the image of $E_0$: an arbitrary system--reference state together with the shared initial bath vacuum.
\end{theorem}
\begin{proof}
By Lemmas~\ref{lem:flat} and \ref{lem:activity}, every surviving
component crosses source to cut and contributes at most its activity.
Dropping compatibility and disjointness only enlarges the nonnegative sum:
\[
 \prod_{C\text{ crossing connected}}(1+\prod_{b\in C}\omega_b)-1
 \le\exp\left(\sum_C\prod_{b\in C}\omega_b\right)-1\le e^\zeta-1.
\]
Lemma~\ref{lem:animals} bounds the exponent. Since $V,\widetilde V$
are isometries,
\[
 (V-\widetilde V)^\dagger(V-\widetilde V)=2I-\mathcal M-\mathcal M^\dagger,
\]
whose norm is at most $2\norm{I-\mathcal M}$. Taking square roots proves
the bound on initialized inputs in Eq.~\eqref{eq:patchbound}.
\end{proof}

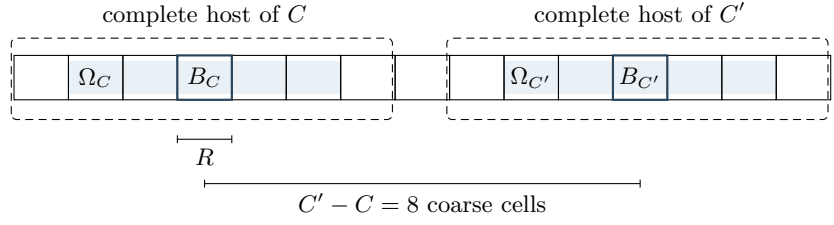
\begin{figure}[t]
\centering
\begin{tikzpicture}[x=0.72cm,y=0.65cm,font=\small]
  \draw[densely dashed,rounded corners=2pt] (-3.05,-.4) rectangle (3.95,1.25);
  \draw[densely dashed,rounded corners=2pt] (4.95,-.4) rectangle (11.95,1.25);
  \fill[circuitpale] (-2,.12) rectangle (3,.78);
  \fill[circuitpale] (6,.12) rectangle (11,.78);
  \foreach \x in {-3,...,11}{\draw (\x,0) rectangle +(1,.9);}
  \draw[thick,circuitblue] (0,0) rectangle +(1,.9);
  \draw[thick,circuitblue] (8,0) rectangle +(1,.9);
  \node at (.5,.45) {$B_C$}; \node at (8.5,.45) {$B_{C'}$};
  \node at (-1.5,.45) {$\Omega_C$};
  \node at (6.5,.45) {$\Omega_{C'}$};
  \node[anchor=south] at (.5,1.3) {complete host of $C$};
  \node[anchor=south] at (8.5,1.3) {complete host of $C'$};
  \draw (0,-.8)--node[below]{$R$}(1,-.8);
  \draw (0,-.72)--(0,-.88); \draw (1,-.72)--(1,-.88);
  \draw (.5,-1.7)--node[below]{$C'-C=8$ coarse cells}(8.5,-1.7);
  \draw (.5,-1.62)--(.5,-1.78); \draw (8.5,-1.62)--(8.5,-1.78);
\end{tikzpicture}
\caption{Schematic colour layer in one dimension (not to scale for
general $R,r_*$). Bold owner cubes $B_C,B_{C'}$ share a colour;
shaded support envelopes and dashed complete hosts are disjoint. Hosts
include environment, work and routing positions. Same-colour owners
differ by at least eight coarse cells in some coordinate in dimension $D$.
The comparison uses $\Omega_c=\bigcup_C\Omega_C$.}

\label{fig:global}
\end{figure}

\subsection{A constant number of local layers}
We index side-$R$ cubes by $C\in\mathbb Z^D$, with
$B_C=(RC+[0,R)^D)\cap\mathbb Z^D$, and write
$B^{+s}=\{x\in\mathbb Z^D:\operatorname{dist}_\infty(x,B)\le s\}$.
There are at most $n$ assigned coarse cubes:
\[
 \mathcal C_\Lambda:=\{C\in\mathbb Z^D:B_C\text{ contains the anchor of a retained term}\}.
\]
Colour them modulo eight, giving $\chi=8^D$. Circuit $U_c$ activates
colours up to $c$, with $U_0=I$. For owner $C$ of colour $c$, we set
\[
 \Omega_C:=B_C^{+(R+2r_*)}\cap\Lambda,
 \qquad W_C:=U_{c,\Omega_C}U_{c-1,\Omega_C}^\dagger.
\]
Its inverse old pass precedes its new forward pass on the permanent baths.
The complete host contains cubes $\norm{C'-C}_\infty\le3$.
For $R\ge2r_*$, same-colour envelopes and hosts are disjoint
(Figure~\ref{fig:global}): in a separating coordinate, nearest envelope
sites have distance at least
$7R+1-2(R+2r_*)=5R+1-4r_*\ge3R+1>r_*$.
Every local support therefore meets at most one envelope, so whole-support
restriction to their union assigns each gate to one commuting factor:
\[
 D_c^{\rm loc}:=\prod_{C\in\mathcal C_\Lambda:\,\mathrm{colour}(C)=c}W_C
 =U_{c,\Omega_c}U_{c-1,\Omega_c}^\dagger,
 \qquad\Omega_c=\bigcup_{C\in\mathcal C_\Lambda:\,\mathrm{colour}(C)=c}\Omega_C.
\]
The overlap graph obeys
\[
 d_w\le64\kappa\lambda\tau,\quad W_{\mathcal S}\le64n\lambda\tau,
 \quad m_*\ge\lfloor R/(2r_*)\rfloor.
\]
Indeed, summing overlapping labels through sites in $X_b$ gives
\[
 \sum_{c:\,X_c\cap X_b\ne\varnothing}\omega_c
 \le64\tau\sum_{x\in X_b}\sum_{e:\,x\in X_e}\lambda_e
 \le64\kappa\lambda\tau.
\]
Each term's period sum is $\sum_mh=\tau$, also proving the source bound
uniformly in the mesh. A support path of $k$ terms from source anchor in
$B_C$ to a cut site $z\notin\Omega_c$ reaches at most distance $kr_*$.
But $\Omega_C\subseteq\Omega_c$ gives
$\operatorname{dist}_\infty(z,B_C)>R+2r_*$, so $kr_*>R+2r_*$,
implying the weaker displayed $m_*$ bound. Ambient distance makes the
argument valid for missing sites and disconnected $\Lambda$.

We choose
\begin{equation}\label{eq:globalparam}
 N_0=\max\{1,\lceil4096\kappa a\rceil\},\quad\tau=T/N_0,
 \quad m_*=1+\left\lceil\log_8\frac{4096\chi^2nN_0^2}{\eps^2}\right\rceil,
 \quad R=2r_*m_*.
\end{equation}
For rational $A=4096\chi^2nN_0^2/\eps^2$, exact cross-multiplication
finds the least integer $k\ge0$ with $8^k\ge A$ and sets $m_*=1+k$.
Then $6d_w\le1/8$, $W_{\mathcal S}\le n$, and
$\zeta\le\eps^2/(3584\chi^2N_0^2)<1$.
Using $e^\zeta-1\le2\zeta$, the sum of $\chi N_0$ colour column bounds is
$2\chi N_0\sqrt\zeta\le\eps/(8\sqrt{14})<\eps/16$.
Channel transfer per segment and composition give spatial error below
$\eps/8$.

For the comparison on each segment, we use the telescoping identity
\[
 (D_\chi^{\rm loc}\cdots D_1^{\rm loc}-U_\chi)E_0
 =\sum_c D_\chi^{\rm loc}\cdots D_{c+1}^{\rm loc}
 (D_c^{\rm loc}-U_cU_{c-1}^\dagger)U_{c-1}E_0.
\]
Theorem~\ref{thm:patch} therefore uses the global activation column
$U_{c-1}E_0$. We put $W_{\rm loc}=D_\chi^{\rm loc}\cdots D_1^{\rm loc}$.
\begin{corollary}[Parallel local decomposition]\label{cor:geometric}
For Eq.~\eqref{eq:globalparam}, $W_{\rm loc}$ has at most $n$ regional
factors in $\chi=8^D$ parallel layers and
\[
 \norm{(W_{\rm loc}-U_\chi)E_0}\le\frac{\eps}{16N_0}.
\]
Factors have support diameter $O_D(1+\log X)$ and disjoint complete hosts
within a layer. The bound is uniform in the mesh and reference.
\end{corollary}
\begin{proof}
Colouring gives disjointness. The telescoping identity and
Eq.~\eqref{eq:patchbound} give $2\chi\sqrt\zeta<\eps/(16N_0)$;
Eq.~\eqref{eq:globalparam} gives $R=O_D(1+\log X)$.
\end{proof}
The local solver must be accurate on the intermediate states of the
localized circuit, as used in Section~\ref{sec:main-proof}. We next bound
the bath occupation of those states.

\subsection{Occupation bounds and the inputs of the local solver}\label{sec:prefix}
Earlier local factors can populate baths assigned to later periods or
left inactive in the current factor. The complete processing region of
$C$ must therefore contain all bath labels, over all periods, assigned
to cubes with $\norm{C'-C}_\infty\le3$. Physical occupation counts
only bath labels, although Hamiltonian labels also contribute to the
interaction weight. Every term used by $W_C$ has its anchor in
$\Omega_C$. Since $2r_*\le R$, its assigned cube is within distance
two of $C$, so the required bath space is included.
Section~\ref{sec:implementation} encodes this space using lists.
The anchors lie in a box of side $7R$, and their supports lie in
$RC+[-4R,5R)^D$. This box has side $9R$ and determines the volume
used in the interaction-weight bounds.

For a global prefix isometry $Y=UE_0$, each window and occupation
projection acts as identity outside its local registers. All bounds
hold for arbitrary system and reference inputs. We call a forward or
adjoint collision on one bath a visit. For durations $h_j$ and rates
$\lambda_j$, the total interaction weight is $\sum_jh_j\lambda_j$.
Interleaved gates on other baths affect the estimate only through their
unit operator norm.

The local solver analysis in Section~\ref{sec:private-decay} will sum
contributions from several physical prefixes of a call, including
terms with a specified occupied bath.
To bound the sum, we use one circuit whose visit counts dominate every
prefix involved. Let $\nu$ bound its visits per label and let $\Sigma$
be its total local interaction weight, including the target Hamiltonian
occurrences where required. This circuit supplies the same bounds
$a_b$ for every summand in that analysis. Theorem~\ref{thm:prefix} constructs
such a circuit for forward calls and for the preimages needed by inverse
calls.

\begin{lemma}[Occupation moment]\label{lem:moment}
Let $U_1,\ldots,U_m$ be elementary unitaries and
$Y=U_m\cdots U_1E_0$, where $E_0$ initializes every counted bath to
vacuum. Let $\mathcal B_{\rm cat}$ be the fixed
local set of physical bath labels. Its occupation operator is
$N_{\rm cat}=\sum_{b\in\mathcal B_{\rm cat}}P_b$.
Suppose each elementary gate touches at most one bath, and define
\[
 a_b=\sum_{j:\,j\text{ touches }b}\norm{[P_b,U_j]}.
\]
For any $F\subseteq\mathcal B_{\rm cat}$ and $z\ge1$, with all
products and bath sums taken over $\mathcal B_{\rm cat}$,
\begin{align}
 \norm{P_FY}&\le\prod_{b\in F}a_b,\qquad P_F=\prod_{b\in F}P_b,\label{eq:subset}\\
 Y^\dagger z^{N_{\rm cat}}Y
 &\preceq\prod_b[1+(z-1)a_b^2]I,\label{eq:physmoment}\\
 Y^\dagger P_bz^{N_{\rm cat}-P_b}Y
 &\preceq a_b^2\prod_{c\ne b}[1+(z-1)a_c^2]I.\label{eq:markedmoment}
\end{align}
For normalized physical collision gates, a total interaction-weight
bound $\Sigma$ and visit bound $\nu$ give
$\sum_ba_b^2\le\nu\Sigma/4$.
\end{lemma}
\begin{proof}
Since $P_bE_0=0$ and the occupation projections commute,
\[
 P_FUE_0=\left(\prod_{b\in F}\ad_{P_b}\right)(U)E_0.
\]
The Leibniz expansion preserves gate order. A derivative for two distinct
baths vanishes on a one-bath gate, so surviving terms assign distinct
marked baths to distinct gates. Unmarked factors have unit norm; allowing
all assignments bounds the sum by $\prod_{b\in F}a_b$, proving
Eq.~\eqref{eq:subset}.

We expand $z^{N_{\rm cat}}=\prod_b[I+(z-1)P_b]$ and apply
$Y^\dagger P_FY\preceq\norm{P_FY}^2I$ termwise, since $z\ge1$.
This proves Eq.~\eqref{eq:physmoment}; keeping $P_b$ marked and expanding
the other baths proves Eq.~\eqref{eq:markedmoment}.

For a jump Cayley gate or adjoint,
$\norm{[P_b,U_j]}\le\sqrt{h_j}\norm{J_j}
\le\sqrt{h_j\lambda_j}/2$. Cauchy--Schwarz over at most $\nu$ visits
gives $a_b^2\le(\nu/4)\sum_{j\text{ touches }b}h_j\lambda_j$.
Summing gives $\sum_ba_b^2\le\nu\Sigma/4$, unchanged by a reference.
\end{proof}

A window consists of consecutive complete periods, with total
interaction weight $s=\sum_{\rm occurrences}h\lambda_e$, which we
call its action. Dividing a segment into windows preserves the order
and physical labels. All relevant local supports fit in a box of volume
$v_*=(9R)^D$. Small window action controls the transducer residual
in Section~\ref{sec:private-decay}; the number of collision periods
will instead control discretization error. We choose
\[
 s_0=(256\cdot7)^{-2},\qquad
 r=2^{\lceil\log_2\max\{1,\lambda\tau v_*/s_0\}\rceil}.
\]
Let $M$ be the least power of two at least
\begin{equation}\label{eq:mesh}
 \max\left\{r,\frac{16C_{\rm col}na^2}{N_0\eps}\right\},\qquad h=\tau/M.
\end{equation}
Then $r\mid M$, $r=O_D(1+R^D)$, $\log M=O_D(1+\log X)$,
every forward window has action at most $s_0$, and
Lemma~\ref{lem:collision} gives collision error at most $\eps/16$.
Let $V^{(c)}_{C,j}$ be the forward circuit on group $j$ of $M/r$ periods
with colour mask $c$. Then
\begin{equation}\label{eq:regroup}
 W_C=V^{(c)}_{C,r}\cdots V^{(c)}_{C,1}
 V^{(c-1)\dagger}_{C,1}\cdots V^{(c-1)\dagger}_{C,r}.
\end{equation}
Old inverse windows execute $r$ to $1$, then new forward windows $1$ to $r$.

Within a nonempty forward window, write $V=V_L\cdots V_1$ in
chronological collision order and put $V_{<j}=V_{j-1}\cdots V_1$ for
$1\le j\le L$, with $V_{<1}=I$. Here $j$ indexes occurrences inside
one window, rather than windows in Eq.~\eqref{eq:regroup};
an empty window has only the identity prefix.

\begin{theorem}[Occupation bounds for local simulation inputs]\label{thm:prefix}
Suppose the exact local factors follow their prescribed colour and
owner order, with windows as in Eq.~\eqref{eq:regroup}. For a forward
window $V$ with ideal input $Y_{\rm in}$, we consider $Y_{\rm in}$ and
all columns $V_{<j}Y_{\rm in}$. For an inverse window $V^\dagger$,
consider $V^\dagger Y_{\rm in}$ and all columns
$V_{<j}V^\dagger Y_{\rm in}$ required by the forward analysis.
Each call family satisfies Lemma~\ref{lem:moment} with one common set
of nonnegative bounds $a_b$ and the uniform estimates
\[
 \nu:=2\cdot7^D+4,\qquad \Sigma_*:=\nu\lambda\tau v_*,
 \qquad \sum_ba_b^2\le\nu\Sigma_*/4.
\]
The bounds include occupied spectators in the complete local bath set
and are independent of the numbers of windows and periods.
\end{theorem}
\begin{proof}
A bath can be touched only by factors whose envelopes contain its term
support. Its anchor belongs to at most $7^D$ such enlarged cubes, and
each factor contains at most one old inverse pass and one new forward
pass. Consequently, a permanent label occurs at most $2\cdot7^D$ times
in the complete localized product. Each pass uses a label in exactly
one window, which preserves this count under subdivision.

For each call, the bound must also cover the prefixes used in its
catalyst estimate. In the inverse case these extend
$V^\dagger Y_{\rm in}$ by a forward prefix of $V$, adding at most two
visits per current label. Thus $\nu=2\cdot7^D+4$ covers both call families.

To obtain simultaneous bounds, we use the circuit for $Y_{\rm in}$ followed by
the complete $V$ in the forward case, and the circuit for $Y_{\rm in}$ followed
by $V^\dagger$ and then the complete $V$ in the inverse case. We retain
this as an unsimplified visit sequence, so each required prefix is
present even though the appended inverse and forward windows cancel
at the endpoint. For each label $b$, we define $a_b$ by summing its
commutator norms over this call's complete auxiliary physical circuit. Every prefix in the
corresponding sum uses a subset of those visits. Consequently,
one common set of $a_b$ bounds every term in that estimate.

Finally, each local term supplies $M$ permanent labels, each contributing
at most $\nu h\lambda_e$ to the common interaction-weight bound.
All associated supports lie in the fixed box of volume $v_*$, so the
sitewise definition of $\lambda$ gives a total at most
$\nu\lambda\tau v_*$. Bath occupation charges the jump labels;
we also include the target window's H occurrences in this bound for
the unmarked contribution in Section~\ref{sec:private-decay}.
Gates on other baths commute with the marked
occupation projections, and their unit norm preserves this estimate
even when they share system qubits. This proves the common value of $\Sigma_*$. Cauchy--Schwarz over
visits to each fixed label, as in Lemma~\ref{lem:moment}, supplies
the additional factor $\nu$ in $\sum_ba_b^2\le\nu\Sigma_*/4$.
\end{proof}

In particular, taking $z=2$ in Eq.~\eqref{eq:physmoment} gives, for
each of these physical input isometries $Y$ and every integer $u\ge0$,
\[
 \norm{\mathbf1_{(2u,\infty)}(N_{\rm cat})Y}
 \le e^{\nu\Sigma_*/8}2^{-u}.
\]
The number of stored records is therefore controlled by local
interaction weight and accuracy. The total number of baths enters
only through the label width. Section~\ref{sec:implementation} uses
this tail estimate to compare finite lists with the full physical bath
space. Table~\ref{tab:contracts} summarizes the input domains of the
different estimates. Occupation is bounded on ideal intermediate
states; approximation errors are propagated by the actual unitary
circuits that follow them.

\begin{table}[H]
\centering\small
\begin{tabular}{>{\raggedright\arraybackslash}p{.24\linewidth}>{\raggedright\arraybackslash}p{.42\linewidth}>{\raggedright\arraybackslash}p{.23\linewidth}}
\toprule
Comparison & Input on which the bound holds & Error norm and location\\\midrule
Collision channel & Arbitrary system--reference input with independent bath vacuum per jump & Diamond norm; Lemma~\ref{lem:collision}\\
Spatial decomposition & Global activation unitary applied to the shared initial vacuum & Initialized operator norm; Theorem~\ref{thm:patch}\\
Ideal amplified window & Physical localized prefixes and the inverse preimages of Theorem~\ref{thm:prefix} & Initialized operator norm; Eq.~\eqref{eq:outererror}\\
Encoded auxiliary circuit & Both public and private components, bounded occupation and initialized work & Restricted operator norm; Lemma~\ref{lem:directfinite}\\
Implemented local factor & The original physical prefixes, including their high-occupation tails & Initialized operator norm; Theorem~\ref{thm:block}\\
Terminal synthesis & Entire allocated register Hilbert space & Full operator norm; Eq.~\eqref{eq:sk-error}\\\bottomrule
\end{tabular}
\caption{Input domains in the approximation argument. Each initialized
operator estimate is uniform over the input and an untouched reference.
The assembled segment estimate is converted to channel error and
composed across segments using Lemma~\ref{lem:samechannel}.}
\label{tab:contracts}
\end{table}

\section{Coherent simulation within a spatial block}
\label{sec:compiler}
Fix a nonempty forward window $V=V_L\cdots V_1$, with occurrences
in chronological order. It acts on the complete system--bath data
space, including occupied baths that are inactive in the window. An
empty window is the identity. The input isometry $Y$ is a physical
prefix or an inverse preimage covered by Theorem~\ref{thm:prefix}.
Executing the window one collision at a time would give a gate count
linear in $L$, which can be polynomial in $1/\eps$.
We will instead use the fixed term order to act on an occurrence
address in $\lceil\log_2L\rceil$ qubits. The resulting circuit implements
a larger unitary $S$ from which a close approximation to $V$ can be
obtained by a short sequence of calls.

\paragraph{How the larger unitary represents the target.}
The transducer framework \cite[Section~3.1, Theorem~3.1]{BJY}
divides a Hilbert space into a \emph{public} data subspace and an
orthogonal \emph{private} subspace. These are two coherent branches,
distinguished by a flag; the word private does not mean a discarded
environment. Write
\[
 S=\begin{pmatrix}S_{00}&S_{01}\\S_{10}&S_{11}\end{pmatrix}.
\]
For public amplitude $a$ and private amplitude $z$, one call performs
\[
 a\oplus z\ \longmapsto\
 (S_{00}a+S_{01}z)\oplus(S_{10}a+S_{11}z).
\]
Thus $S_{10}$ moves amplitude into the private branch and $S_{01}$
can return it to the public branch. A linear map $\Gamma$ with
\begin{equation}\label{eq:transducer-intro}
\begin{aligned}
 S(x\oplus\Gamma x)&=Vx\oplus\Gamma x,\\
 (I-S_{11})\Gamma&=S_{10},\qquad S_{00}+S_{01}\Gamma=V
\end{aligned}
\end{equation}
describes a private amplitude that is unchanged by a call while the
public amplitude undergoes $V$; the second line is the block form of
the first. This is the \emph{catalyst map}.
It is an analysis object: $\Gamma x$ need not be normalized, and the
circuit starts with private amplitude zero.

Finite reuse feeds one shared private branch through several public
address slots and coherently averages their outputs
\cite[Lemma~5.1 and Sections~6--7]{ChenLindblad}. For an integer $N\ge1$,
its selected operator is
\[
 P_N=V-S_{01}g_N(S_{11})\Gamma,
 \qquad g_N(z):=\frac1N\sum_{j=0}^{N-1}z^j.
\]
Section~\ref{sec:reuse} combines the operators $P_N$ into an
approximation $M_q$ of $V$ and constructs a unitary with selected block
$M_q/2$. For an integer $q\ge1$, the residual is
\begin{equation}\label{eq:residual-roadmap}
 V-M_q=S_{01}p(S_{11})^qg_{6q}(S_{11})\Gamma,
 \qquad p(z):=\frac{z^2+z^4}{2}.
\end{equation}
Thus the quantity to control on the physical input $Y$ is
$p(S_{11})^q\Gamma Y$: the averaging polynomial is a contraction and
commutes with $p(S_{11})^q$, and $S_{01}$ is also a contraction. The
adjoint estimate is needed for the complete amplification circuit.
Section~\ref{sec:private-decay} obtains both estimates by cancelling
same-label motion in a retained reference. The remaining reference
transitions advance through distinct occupied labels, while deviations
from the reference are small when the window action is small. Low
occupation therefore forces a long surviving product to contain many
small deviations. The polynomial describes the residual of a unitary
reuse circuit; it is not applied as a separate nonunitary gate.

\paragraph{From the transducer to the local circuit.}
We first build one local transducer, compose the local identities,
and give an address-tree circuit for their product $S$.
We then bound its private residual and the catalyst's occupation.
The bath configurations present in $Y$ enter these estimates;
the catalyst itself is never prepared. Finally we turn the exact
identities into finite circuits on short lists of occupied baths.
Section~\ref{sec:window-error} combines their low-occupation guarantee
with the physical tail bound and the residual estimate on the original
$Y$. For an inverse window, the same forward estimate is applied to
$V^\dagger Y_{\rm in}$ before taking the adjoint of the saved circuit.

\subsection{A Cayley transducer for occupied baths}
The following identity extends the Cayley transducer for Hamiltonians
\cite[Lemma~2]{ChenHamiltonian} and the isometries for fresh baths
\cite[Lemma~3.1]{ChenLindblad}, \cite[Proposition~12]{WangYe} to the
prescribed unitary on both vacuum and occupied inputs. The two bath
sectors may have different weights, provided their geometric mean
gives the required coupling. We choose weight one on the occupied
sector and a weight proportional to $h$ on vacuum. This preserves the
physical gate and makes the transducer's public-to-public block vanish
on occupied input. Consequently, an interval product of these blocks
is exactly zero if the interval contains an occupied active jump bath.
Otherwise it is a scalar computable from the number of occurrences of
each term. This property will make the address-tree coefficients
efficiently evaluable in Section~\ref{sec:implementation}.

Let $\mathcal H_{\rm cell}$ contain the local system and selected jump
bath, with public space $\mathcal H_{\rm cell}$ and private space
$\mathbb C^2_{\mathsf Q}\otimes\mathcal H_{\rm cell}$.
With a public/private flag $\mathsf P$, the three active sectors are
$\ket{\pub,0}_{\mathsf P\mathsf Q}\otimes x$,
$\ket{\priv,0}_{\mathsf P\mathsf Q}\otimes x$ and
$\ket{\priv,1}_{\mathsf P\mathsf Q}\otimes x$.
They use the same cell register; the unused public ancilla-one sector
can be fixed by identity.
For $K=K^\dagger$, $\norm K\le1/2$, and $0\preceq\Theta\preceq I$, we define
\[
 O:=\begin{pmatrix}K&\sqrt{I-K^2}\\\sqrt{I-K^2}&-K\end{pmatrix},\qquad
 q_0:=\binom I0,\qquad P:=q_0q_0^\dagger.
\]
Here $q_0x=\ket0_{\mathsf Q}\otimes x$ embeds the cell into the private
ancilla-zero subspace, and $q_0^\dagger Oq_0=K$.
In particular,
\[
 O(\ket0_{\mathsf Q}\otimes x)
 =\ket0_{\mathsf Q}\otimes Kx+
   \ket1_{\mathsf Q}\otimes\sqrt{I-K^2}\,x.
\]
The Hermitian unitary $O$ is a reflection, since $O^2=I$.
We call its address-controlled application the \emph{selected
block-encoding gate}; it is distinct from the signal reflection used
for amplitude amplification in Section~\ref{sec:reuse}.
The data-space operators
\[
 t:=(I-\Theta)(I+\Theta)^{-1},\qquad
 b:=2\sqrt\Theta(I+\Theta)^{-1}
\]
are commuting functions of $\Theta$. They form the larger block operator
\[
 R_\Theta:=\begin{pmatrix}t&-ibq_0^\dagger\\q_0b&i(I-P+q_0tq_0^\dagger)\end{pmatrix},
\]
which couples public to the selected private subspace. Applying $O$ first gives
\[
 S_\Theta:=R_\Theta(I\oplus O),\quad
 H_\Theta:=\sqrt\Theta K\sqrt\Theta,\quad
 D_\Theta:=I+iH_\Theta,\quad
 \gamma:=(I+iO)q_0\sqrt\Theta D_\Theta^{-1}.
\]
$K$ and $\Theta$ need not commute; these factor orders are fixed.
On input $x\oplus0$, a first call to $S_\Theta$ produces
$tx\oplus q_0bx$. The block encoding acts on the private amplitude
when it is present in a later call. The following identity gives the
stationary private amplitude and identifies the target transformation.
\begin{lemma}[Weighted Cayley identity]\label{lem:cell}
Both $S_\Theta$ and $V_\Theta=(I-iH_\Theta)D_\Theta^{-1}$ are unitary, and
\[
 S_\Theta\binom{I}{\gamma}=\binom{V_\Theta}{\gamma},
 \qquad \gamma^\dagger \gamma=2D_\Theta^{-\dagger}\Theta D_\Theta^{-1}.
\]
\end{lemma}
\begin{proof}
$O$ is a Hermitian involution, and $t^2+b^2=I$ makes $R_\Theta$
unitary on the two coupled copies; the remaining copy has phase $i$.
Since $H_\Theta$ is self-adjoint, $D_\Theta$ is invertible and
$V_\Theta$ unitary. We put $r_\Theta=\sqrt\Theta$ and
$B_K=\sqrt{I-K^2}$. We use
\[
 t+br_\Theta=I,\qquad tr_\Theta-b=-r_\Theta,\qquad
 O\gamma=(O+iI)q_0r_\Theta D_\Theta^{-1}.
\]
Right-multiplying by $D_\Theta$ clears the common denominator. The public
component and the private ancilla-zero component become
\[
\begin{aligned}
 tD_\Theta+br_\Theta-ibKr_\Theta
 &=I+i(tr_\Theta-b)Kr_\Theta=I-ir_\Theta Kr_\Theta,\\
 bD_\Theta-tr_\Theta+itKr_\Theta
 &=(b-tr_\Theta)+i(br_\Theta+t)Kr_\Theta=(I+iK)r_\Theta.
\end{aligned}
\]
Here each second equality uses the preceding commuting functions of
$\Theta$, leaving $K$ in its prescribed position. The other private
numerator is $iB_Kr_\Theta$, so the private output is
$\gamma D_\Theta$. Multiplication by $D_\Theta^{-1}$ proves the
identity; $(I+iO)^\dagger(I+iO)=2I$ gives its metric.
\end{proof}
The Frobenius normalization gives $\norm K\le1/2$ for both local
specializations. The mesh choice in Eq.~\eqref{eq:mesh} also gives
$h\lambda_e\le1$, so their weights are positive contractions.
For a Hamiltonian occurrence, the choice $K=H/\alpha$ and
$\Theta=h\alpha I/2$ gives $H_\Theta=hH/2$, as required by the
physical Cayley step. For a jump occurrence, we take
$K=G_J/\beta$ and
\[
 \Theta=(h\beta^2/4)P_0+P_1.
\]
For the jump choice, $G_J$ couples the two bath sectors, so the geometric
mean of their weights gives $H_\Theta=\sqrt h G_J/2$. The weighted
identity therefore realizes the prescribed physical unitary on both
sectors. The transducer $S_\Theta$ has public-to-public block $t_J$,
which is supported on vacuum:

\begin{equation}\label{eq:through}
 t_J=t_J^{(0)}P_0,\qquad
 t_J^{(0)}=\frac{1-h\beta^2/4}{1+h\beta^2/4},
\end{equation}
whereas $t_H$ is a scalar.
For H, $\gamma_H^\dagger \gamma_H\preceq h\alpha I$. For J,
\begin{equation}\label{eq:metricbound}
 \gamma_J^\dagger \gamma_J\preceq\tfrac54h\beta^2P_0+4P_1.
\end{equation}
To verify Eq.~\eqref{eq:metricbound}, we set $r=h\beta^2/4$ and let
$K_{10}:=P_1KP_0$ be the block of $K$ from $\operatorname{ran}P_0$ to
$\operatorname{ran}P_1$. On vacuum input, the resolvent components are
$z_0=(I+rK_{10}^\dagger K_{10})^{-1}x$ and
$z_1=-i\sqrt rK_{10}z_0$. The exact metric and
$\norm{K_{10}}\le1/2$ give, for $x\in\operatorname{ran}P_0$,
\[
 \norm{\gamma P_0x}^2
 =2r\bigl(\norm{z_0}^2+\norm{K_{10}z_0}^2\bigr)
 \le(5/2)r\norm x^2.
\]
On the whole space, the same metric and the resolvent contraction give
$\norm{\gamma}^2\le2$. Splitting the input into these two sectors and
using $\norm{u+v}^2\le2\norm u^2+2\norm v^2$ proves
Eq.~\eqref{eq:metricbound}. The vacuum term is proportional to $h$. The occupied term is of constant
order, so summing it over a window requires the occupation bounds for
the physical prefixes.

\subsection{Sequential composition of local transducers}
Lemma~\ref{lem:cell} holds on the full data space, so sequential
composition also permits occupied baths
\cite[Definition~9.8 and Proposition~9.9]{BJY}. For
$V=V_L\cdots V_1$, one common public space and orthogonal private
occurrence summands form
\begin{equation}\label{eq:semanticspace}
 \mathcal K:=\mathcal H_{\rm dat}^{(\pub)}\oplus
 \bigoplus_{j=1}^L(\mathbb C^2_{\mathsf Q}\otimes\mathcal H_{\rm dat})_j.
\end{equation}
This direct sum is represented by one data register and an address,
not by $L$ copies of the data. Public amplitude $x$ occupies
$\ket{\pub,0,0}\otimes x$, and amplitude $\ket z\otimes x$ in private
summand $j$ occupies $\ket{\priv,j-1,z}\otimes x$, where $z=0,1$
is the block-encoding ancilla value. Superpositions of these branches
may entangle the address with the shared system--bath data.
Section~\ref{sec:implementation} adds the finite bath encoding and work
registers. A static term can recur in different periods, but its
permanent labels $(e,m)$ remain distinct. A forward window visits each label at most once,
without assuming that its bath is initially empty.
The following form of sequential composition is adapted from
\citet[Proposition~9.9]{BJY}.
\begin{lemma}[Sequential composition]\label{lem:imported-composition}
Suppose $S_j$ is a unitary on $\mathcal H\oplus\mathcal L_j$ and
$S_j(I,\gamma_j)^{\mathsf T}=(V_j,\gamma_j)^{\mathsf T}$, where every $V_j$
acts on the common public space $\mathcal H$. When each $S_j$ acts as
identity on the other private summands, their product has target
$V_L\cdots V_1$ in the graph identity, with catalyst map
$\Gamma x:=\bigoplus_j \gamma_jV_{j-1}\cdots V_1x$.
\end{lemma}
\begin{proof}
Apply the local identities in order: step $j$ retains
$\gamma_jV_{j-1}\cdots V_1x$ in its private summand, while later
steps act as identity there.
\end{proof}
We apply this identity with $\mathcal H=\mathcal H_{\rm dat}$ and
$\mathcal L_j=(\mathbb C^2_{\mathsf Q}\otimes\mathcal H_{\rm dat})_j$.
The $j$th catalyst is evaluated on the physical prefix
$V_{j-1}\cdots V_1$, controlled by Theorem~\ref{thm:prefix}.
Using block indices for the public and private components $0,1$ gives
\begin{equation}\label{eq:cascadegraph}
 S:=S_L\cdots S_1=\begin{pmatrix}S_{00}&S_{01}\\S_{10}&S_{11}\end{pmatrix},
 \qquad S\binom{I}{\Gamma}=\binom{V}{\Gamma},
 \qquad S_{10}=(I-S_{11})\Gamma,\quad S_{00}+S_{01}\Gamma=V.
\end{equation}
The address trees implement $S$ with a logarithmic occurrence register.
To track its occupation changes and interval products, let
$D_{a,b}=t_b\cdots t_a$ (identity for an empty interval) and let $q_j$
embed data into the $\mathsf Q=0$ coordinate of summand $j$.
For a transition from private summand $j$ to a later summand $i$,
there is one chronological route: apply $O_j$, leave summand $j$
for public, remain public through cells $j+1,\ldots,i-1$, and enter
summand $i$. Its amplitude does not encounter $O_i$, because $O_i$
acts on the private input before that cell's public-to-private transfer.
Once a cell has been passed, its private summand is left untouched.
Multiplying along these routes gives
\begin{equation}\label{eq:allcorners}
\begin{aligned}
 S_{00}&=D_{1,L},\\
 (S_{10})_i&=q_i b_iD_{1,i-1},\\
 (S_{01})_j&=-iD_{j+1,L}b_jq_j^\dagger O_j,\\
 (S_{11})_{jj}&=i(I-q_jq_j^\dagger+q_jt_jq_j^\dagger)O_j,\\
 (S_{11})_{ij}&=-iq_ib_iD_{j+1,i-1}b_jq_j^\dagger O_j\quad(i>j),\\
 (S_{11})_{ij}&=0\quad(i<j).
\end{aligned}
\end{equation}
In particular, $(S_{11})_{ij}$ for $i>j$ applies $O_j$ before the
interval product. Thus the latter is evaluated on the bath state after
that local operation. This ordering is essential on occupied inputs and
will be preserved by the compressed implementation.

Let $\mathcal O(y)$ be the occupied physical bath labels in basis
configuration $y$, and let $b(j)$ be the permanent label of occurrence
$j$, with a distinct formal label for a Hamiltonian occurrence. For the
private residual estimate, we also count the selected occurrence label.
The counted set is $\mathcal O(y)$ in the public space and
$\mathcal O(y)\cup\{b(j)\}$ in private summand $j$. Thus the selected
label remains counted when its bath changes. Selecting a vacuum bath
adds one label, selecting an occupied bath counts it once, and a
Hamiltonian occurrence adds its formal label. Physical occupation
continues to count only bath labels.

Let $N_{\rm aug}$ count this set and act as identity on the system,
block-encoding ancilla and reference. We put
$Q_{\le k}=\mathbf1_{[0,k]}(N_{\rm aug})$. These operators act on the
direct sum in Eq.~\eqref{eq:semanticspace}; unused address coordinates
are handled by a separate identity extension. The physical count
$N_{\rm cat}$ controls the input tail. The count $N_{\rm aug}$ also
includes the selected label and controls private residuals and the
capacity needed for reuse.
\begin{lemma}[Change in augmented occupation number]\label{lem:bandwidth}
For the physical H/J specializations above and every $k\ge0$,
\[
 S Q_{\le k}=Q_{\le k+1}S Q_{\le k},\qquad
 S^\dagger Q_{\le k}=Q_{\le k+1}S^\dagger Q_{\le k}.
\]
\end{lemma}
\begin{proof}
We inspect the blocks because a complete transducer call contains
many occurrences. For $i>j$, only $O_j$ changes a bath value in
$(S_{11})_{ij}$. Its label $b(j)$ is already included in the input
augmented set $F$; the interval factors preserve all bath occupations,
and the output set is contained in $F\cup\{b(i)\}$. The adjoint block is
\[
 (S_{11}^\dagger)_{ji}
 =iO_jq_jb_jD_{j+1,i-1}^\dagger b_iq_i^\dagger.
\]
Here $O_j$ acts last, and its bath is charged by the output label
$b(j)$. Thus the output set is contained in $F\cup\{b(j)\}$.
The blocks coupling public and private components, and the diagonal blocks, have the same one-label support
property. Every block therefore sends occupation at most $k$ into
occupation at most $k+1$, which proves both projection identities.
\end{proof}

\subsection{An exact binary-tree factorization}\label{sec:tree-factorization}
The composed transducer contains $L$ rotations. Executing them
directly would make the gate count proportional to the number of
collision occurrences. We instead evaluate their product through a
binary tree on the occurrence address. The dyadic factorization of
\citet[Lemma~10]{ChenDyadic}, also used in
\citet[Section~9.3]{ChenLindblad}, assumes a common rotation cosine.
Here the cosine depends on the term and bath configuration. We use
the unequal-rotation merge of Wang and Ye~\cite[Section~4.3.1,
Eqs.~(128)--(130)]{WangYe} in this recursive structure. We give the
identity explicitly to track its bath-dependent coefficients and its
action on both public and private inputs, as required for finite reuse.

Let $Q_j$ apply $O_j$ only in private summand $j$. For $k\ne j$,
$R_k$ couples public to private summand $k$ and is identity on summand
$j$, so $Q_j$ commutes with $R_k$. This argument uses the direct-sum
supports, and hence holds for arbitrary system matrices.
Define the full selected operation by
\[
 Q:=Q_L\cdots Q_1=I_{\pub}\oplus\bigoplus_{j=1}^L O_j.
\]
The private summands are orthogonal, so $Q$ applies only the block
encoding selected by the occurrence address. Its implementation scans
the fixed list of term types, rather than executing all $L$ occurrences.
The operator $Q$ is distinct from the one-qubit ancilla $\mathsf Q$.
On a fixed bath configuration $y$, let $G_j$ be the real rotation
with cosine $t_j(y)$ and sine $b_j(y)$ between the public component and the coordinate of private summand $j$ with ancilla state zero; it fixes the coordinates with ancilla state one.
The cell rotation factors as $R_j=G_jD_j$, where $D_j$ applies phase
$i$ to both ancilla sectors of summand $j$ and fixes the other summands.
These phases commute with other cells' real rotations. Moving the local
block encodings to the input and collecting the phases therefore gives
\[
 S=(G_L\cdots G_1)D_{\rm ph}Q,
 \qquad D_{\rm ph}=I_{\pub}\oplus iI_{\priv}.
\]
The recursive step merges two adjacent intervals. Each interval has one
private representative, so the step acts on three coordinates: public,
the left representative, and the right representative. For interval
cosines and sines $(a,b)$ and $(c,d)$, the exact product on these
coordinates is
\[
 G_BG_A=\begin{pmatrix}ac&-cb&-d\\b&a&0\\da&-db&c\end{pmatrix}.
\]
Its combined public cosine is $ac$. With
$s=\sqrt{1-a^2c^2}>0$, the required private basis changes are
\[
 I_{AB}=s^{-1}\begin{pmatrix}cb&d\\-d&cb\end{pmatrix},\qquad
 O_{AB}=s^{-1}\begin{pmatrix}b&da\\-da&b\end{pmatrix}.
\]

For example, if the first selected jump bath is occupied, its cosine
and sine are $a=0$ and $b=1$. With $c^2+d^2=1$ for the second
leaf, the formulas give
\[
 I_{AB}=\begin{pmatrix}c&d\\-d&c\end{pmatrix},\qquad O_{AB}=I.
\]
Thus the input basis change is generally nontrivial even when the
output basis change is identity. This is the asymmetry that an
occupied bath introduces into the two-leaf merge.

The two changes of basis align different coupling directions. Writing
$\ket p,\ket A,\ket B$ for the three coordinates, the public output is
\[
 G_BG_A\ket p=ac\ket p+b\ket A+ad\ket B.
\]
The normalized private input direction coupled to public is
$(cb\ket A+d\ket B)/s$, whereas the private output direction is
$(b\ket A+ad\ket B)/s$. The input rotation sends the former to $\ket A$;
the inverse output rotation sends $\ket A$ to the latter. We verify the
remaining private directions in the proof below.

In the ordered coordinates $(p,A,B)$, write
\[
 G(t,s):=\begin{pmatrix}t&-s&0\\s&t&0\\0&0&1\end{pmatrix},
 \qquad t^2+s^2=1.
\]
This rotation couples the public coordinate to the left representative and fixes the other private coordinate.

\begin{lemma}[Merging two rotation sequences \cite{WangYe}]\label{lem:merge}
Assume $0\le a,c\le1$, $b=\sqrt{1-a^2}$, and $d=\sqrt{1-c^2}$,
as holds for our public-to-public factors. These matrices satisfy
\[
 G_BG_A=(1\oplus O_{AB}^{\mathsf T})G(ac,s)(1\oplus I_{AB}).
\]
When $s=0$, define all three factors as identity.
\end{lemma}
\begin{proof}
For $s>0$, $c^2b^2+d^2=b^2+d^2a^2=s^2$ makes both private
basis changes orthogonal. The first row and column agree; the remaining
private block follows from the same unit-circle identities:
\[
 s^{-2}\begin{pmatrix}b&-da\\da&b\end{pmatrix}
 \begin{pmatrix}ac&0\\0&1\end{pmatrix}
 \begin{pmatrix}cb&d\\-d&cb\end{pmatrix}
 =\begin{pmatrix}a&0\\-db&c\end{pmatrix}.
\]
At $s=0$, nonnegativity forces $a=c=1$ and $b=d=0$, so all factors
are identity.
\end{proof}

On each bath configuration, the private basis changes are defined
recursively. For adjacent child intervals $A,B$, we extend the merge
matrices by identity on their other coordinates and set
\[
 F_{\rm in}^{AB}=I_{AB}(F_{\rm in}^{A}\oplus F_{\rm in}^{B}),\qquad
 F_{\rm out}^{AB}=O_{AB}(F_{\rm out}^{A}\oplus F_{\rm out}^{B}).
\]
Each leaf has identity basis changes. The child changes have disjoint
private supports and commute with the other child's public rotation,
so Lemma~\ref{lem:merge} proves the full interval factorization by
induction. We take
$w=\lceil\log_2L\rceil$ and pad the right end with identity rotations
to obtain $2^w$ leaves; when $L=1$, there are no merge levels.
Choosing the leftmost leaf as each interval's representative yields
\begin{equation}\label{eq:direct}
 S=(1\oplus F_{\rm out}^{\mathsf T})G_{\rm root}
 (1\oplus F_{\rm in})D_{\rm ph}Q.
\end{equation}
The root rotation has cosine $\prod_{j=1}^Lt_j(y)$ and sine
$\sqrt{1-(\prod_{j=1}^Lt_j(y))^2}$; the trees fix the ancilla-one
coordinates. The two basis changes generally differ because each merge has distinct
private input and output coupling directions. Right padding is preserved
at every level: an empty right child has $c=1,d=0$, so both private
basis changes are identity, and a node with two empty children has all
three factors equal to identity. Induction over the tree therefore
shows that the valid occurrence subspace is invariant and the unused
right-padding coordinates are fixed. In the
register implementation, a flag $\mathsf P$ distinguishes the public
and private spaces, and an occurrence register $\mathsf T$ stores $j-1$
for private summand $j$. Thus the valid private projector
$\Pi_{\rm valid}$ selects $\mathsf P=\priv$ and
$0\le\mathsf T<L$. We implement the phase as
$I+(i-1)\Pi_{\rm valid}$ and the selected block encoding as identity
outside this projector, so padded labels retain phase one.
Equation~\eqref{eq:direct} consequently holds on every actual public
and private input and extends by identity to address padding.
Figure~\ref{fig:direct-circuit} shows the resulting order of operations.

\begin{figure}[t]
\centering
\begin{tikzpicture}[x=1cm,y=1cm,font=\small,
 bx/.style={draw=circuitblue,minimum height=8mm,align=center,inner sep=3pt}]
 \draw[double,double distance=1.2pt] (0,0)--(12.3,0);
 \node[anchor=east,align=right] at (0,0) {encoded\\registers};
 \node[bx,fill=circuitpale,minimum width=10mm] at (1,0) {$Q$};
 \node[bx,fill=circuitpale,minimum width=14mm] at (2.7,0) {$D_{\rm ph}$};
 \node[bx,fill=circuitpale,minimum width=22mm] at (5,0) {$F_{\rm in}$};
 \node[bx,fill=circuitpale,minimum width=21mm] at (7.7,0) {$G_{\rm root}$};
 \node[bx,fill=circuitpale,minimum width=26mm] at (10.5,0) {$F_{\rm out}^{\dagger}$};
 \draw (3.65,.72)--node[above]{execution order}(11.8,.72);
\end{tikzpicture}
\caption{Execution order in Eq.~\eqref{eq:direct}: $Q$, private phase, input tree from leaves to root, root rotation, and inverse output tree to leaves. The different tree rotations depend on a preserved address prefix and the bath configuration after $Q$, the only operation changing a bath value. Permanent labels are fixed. Section~\ref{sec:implementation} supplies register operations and costs.}
\label{fig:direct-circuit}
\end{figure}
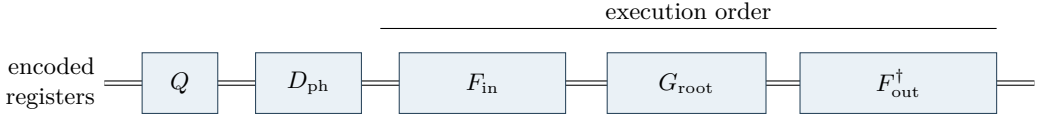

\begin{figure}[H]
\centering
\begin{tikzpicture}[x=1cm,y=1cm,font=\small,
nd/.style={draw=circuitblue,fill=circuitpale,rounded corners=2pt,
minimum height=7mm,align=center}]
\node (a) at (0,2.1) {$00$};
\node (b) at (2.1,2.1) {$01$};
\node (c) at (5.1,2.1) {$10$};
\node (d) at (7.2,2.1) {$11$};
\node[nd] (ab) at (1.05,.9) {$I_{[1,2]}$\\representative $00$};
\node[nd] (cd) at (6.15,.9) {$I_{[3,4]}$\\representative $10$};
\node[nd] (rt) at (3.6,-.55) {$I_{[1,4]}$\\representative $00$};
\draw (a)--(ab); \draw (b)--(ab);
\draw (c)--(cd); \draw (d)--(cd);
\draw (ab)--(rt); \draw (cd)--(rt);
\node[anchor=west] at (8.0,2.1) {address $\mathsf T_1\mathsf T_0$};
\node[anchor=west,align=left] at (8.0,.9) {rotate $\mathsf T_0$\\preserve $\mathsf T_1$};
\node[anchor=west,align=left] at (8.0,-.55) {rotate $\mathsf T_1$\\under $\mathsf T_0=0$};
\end{tikzpicture}
\caption{Four-leaf input tree, with occurrence indices $1,\ldots,4$
encoded as addresses $00,01,10,11$. Each node rotates its two child
representatives; the orthogonal coordinates remain present.
The first level acts on $(00,01)$ and $(10,11)$ with potentially
different angles. The second acts on $(00,10)$.
All rotations are in the private $\mathsf Q=0$ sector and are controlled
by the bath configuration after $Q$.
The inverse output tree has the same pairs in reverse level order,
with the transposed $O_{AB}$ rotations instead of $I_{AB}$.
The edges describe the recursion, not physical data movement.}
\label{fig:four-leaf}
\end{figure}
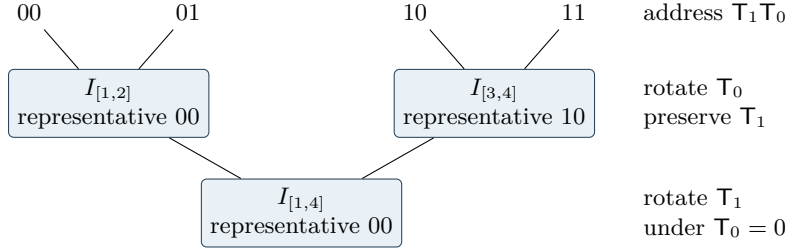

\subsection{Suppression of the private residual}\label{sec:private-decay}
We bound both $p(S_{11})^q\Gamma Y$ and
$p(S_{11}^\dagger)^q\Gamma Y$, where $p(z):=(z^2+z^4)/2$.
We separate the private block into a reference operator and a small
perturbation. Reference transitions preserve the counted label set and
advance to distinct occurrences. Perturbations introduce only a few
labels. A long nonzero product acting on a state with low occupation
must therefore contain many small perturbations. The occupation
moment bounds the complementary tail on the original physical input
isometry $Y$.

The reference $S^{(0)}$ keeps the block encodings and permanent labels.
It sets Hamiltonian weights to zero and jump weights to $P_1$, retaining
the exact quarter turns associated with occupied jump baths. Only
Hamiltonian rotations and jump rotations on vacuum are removed.
After cancelling the common block encodings and phases, we compare
the two rotation sequences through their total interaction weight.
The estimate is uniform in the collision mesh.
\begin{lemma}[Small-action perturbation]\label{lem:perturb}
For a forward window of action $s=\sum_Hh\alpha+\sum_Jh\beta^2$,
\[
 \norm{S_{11}-S_{11}^{(0)}}\le\norm{S-S^{(0)}}\le2\sqrt s.
\]
\end{lemma}
\begin{proof}
After cancelling the common block encodings and phases, we fix a bath
configuration and compare two real rotation chains. The occupied
reference leaves $r_1<\cdots<r_t$ make exact quarter-turn swaps, while
all other reference rotations are identity. With $r_0=0$ denoting public
and $r_{t+1}=L+1$, we define
$T_j=G_{r_j}^{(0)}\cdots G_{r_1}^{(0)}$, with $T_0=I$.
Factoring out the reference swaps gives the exact identity
\[
 T_t^\dagger(G_L\cdots G_1)=\widetilde G_L\cdots\widetilde G_1,
 \qquad \widetilde G_k=\begin{cases}
 T_j^\dagger G_kT_j,&r_j<k<r_{j+1},\\
 I,&k\in\{r_1,\ldots,r_t\}.
 \end{cases}
\]
The first case couples leaf $k$ only to representative $r_j$,
possibly with its sine sign reversed by the coordinate permutation.

For $j\ge1$, the reference swaps satisfy
$T_j^\dagger\ket p=-\ket{r_j}$, and $T_j$ fixes every later weak leaf.
Thus the conjugated rotations in $(r_j,r_{j+1})$ act only on
$\{r_j\}\cup(r_j,r_{j+1})$, with $r_0$ denoting public. These
coordinate sets are pairwise disjoint. The remaining operator is therefore a direct sum of sequences in which one central
coordinate couples successively to distinct leaves. Its norm is the
maximum over these sequences. It is enough to bound
one such sequence in terms of the total weight of its weak rotations.

Let $U_\star$ be the finite coordinate matrix of one such product,
with nonnegative cosines $a_1,\ldots,a_m$. Its trace is
$\prod_ja_j+\sum_ja_j$, including the possible sine signs. The
Frobenius norm on this coordinate matrix satisfies
\[
 \norm{U_\star-I}_F^2
 =2\left(1-\prod_ja_j+\sum_j(1-a_j)\right)
 \le4\sum_j(1-a_j).
\]
A weak H leaf satisfies $1-a_j\le h\alpha_j$, and a vacuum J leaf
satisfies $1-a_j\le h\beta_j^2/2$. Thus the coordinate operator norm
is at most $2\sqrt s$. Tensoring with the untouched system and reference
preserves that norm. Taking the maximum over the disjoint blocks and
bath configurations proves the bound on the full transducer.
\end{proof}

The following projection identity gives the annihilating polynomial
\cite[Lemma~16, Eq.~(87)]{WangYe}.
\begin{lemma}[Annihilating polynomial]\label{lem:imported-annihilator}
If $O=O^\dagger=O^{-1}$ and $P=P^\dagger=P^2$ satisfy $POP=0$,
then $A_0:=i(I-P)O$ satisfies $A_0^4+A_0^2=0$.
\end{lemma}
\begin{proof}
$POP=0$ makes $P,OPO$ orthogonal projections. For $Q=P+OPO$,
ordered multiplication gives $A_0^2=-I+Q$ and
$A_0^4=(I-Q)^2=I-Q$, proving the claim.
\end{proof}
For J, $P=q_0P_1q_0^\dagger$ selects the occupied bath sector with ancilla state zero,
and $POP=q_0P_1KP_1q_0^\dagger=0$ because the coupling has no
occupied--occupied block. For H, $P=0$. Thus $p$ annihilates every
reference diagonal block. Following the filtering construction of
\citet[Section~6.2]{ChenLindblad}, put
\[
 F:=p(S_{11}^{(0)}),\qquad E:=p(S_{11})-p(S_{11}^{(0)}).
\]
These act on the actual private direct sum of
Eq.~\eqref{eq:semanticspace}, including valid occupied spectators.
On hardware padding the extension is identity, so $p(I)=I$ there.
Contractivity gives $\norm F\le1$ and
$\norm E\le3\norm{S_{11}-S_{11}^{(0)}}\le6\sqrt s$.
As proved below, $F$ strictly advances the occurrence index, has zero
Hamiltonian blocks and preserves the counted set, whereas $E$ can
introduce at most four labels. We use these properties to identify the
nonzero products before summing their norm bounds.
\begin{lemma}[Single-pass polynomial decay]\label{lem:packet}
Let $S$ be the transducer of a single-pass forward window with the
specified block encodings from $H/\alpha$ and $G_J/\beta$. Its J
occurrences have distinct permanent bath labels, and its H occurrences
have distinct formal labels disjoint from the bath labels. If the
action satisfies $s\le s_0=(256\cdot7)^{-2}$, then for every integer
$u\ge1$, on the actual private direct sum of Eq.~\eqref{eq:semanticspace},
\[
 \max\{\norm{p(S_{11})^{7u}Q_{\le2u}},
       \norm{p(S_{11}^\dagger)^{7u}Q_{\le2u}}\}\le2^{-2u}.
\]
\end{lemma}
\begin{proof}
We fix $u\ge1$ and put $K=2u$ and $q=7u$. A word in $F,E$, or
in $F^\dagger,E^\dagger$, with $f$ reference factors and $b$
perturbation factors can be nonzero on $Q_{\le K}$ only if
$f\le K+4b$. To prove this count before summing the norm bounds,
we retain the entire augmented label set:
$\Pi_{\mathcal R}$ projects onto private basis states whose set is
exactly $\mathcal R$. For J occurrences $i>j$, the reference
transition is
\[
 (S_{11}^{(0)})_{ij}=-i q_iP_{b(i)}
 \Bigl(\prod_{j<k<i:\,k\text{ a J occurrence}}P_{b(k),0}\Bigr)
 P_{b(j)}q_j^\dagger O_j.
\]
The source projection leaves bath $b(j)$ occupied after $O_j$, so
its permanent label remains in the augmented set when the private
occurrence label changes. The target projection requires bath $b(i)$ to be occupied
already, and every other bath is unchanged. A diagonal block can change
only its labelled bath, whose label is included on both sides by
definition; an H coordinate likewise retains its formal label.
Consequently the reference private block preserves each augmented set:
\[
 \Pi_{\mathcal R}S_{11}^{(0)}\Pi_{\mathcal R'}=0\quad(\mathcal R\ne\mathcal R').
\]
In a fixed sector of the counted label set $\Pi_{\mathcal R}$, the operator
$F=p(S_{11}^{(0)})$ is strictly lower triangular in the occurrence
index. Its diagonal blocks vanish by the annihilating polynomial, and
its H rows and columns vanish because a reference H cell has zero
coupling between public and private components. A nonzero $F$ transition therefore reaches a
different J occurrence whose permanent label lies in $\mathcal R$.
The single-pass condition makes these occurrence labels distinct.

We consider a word with $f$ factors $F$ and $b$ factors $E$, acting on
$Q_{\le K}$. To decide when this word vanishes, we resolve each $E$
into its polynomial monomials and the intermediate augmented-label
sectors. We retain the union of the input set and all labels introduced
along a resolved path. A reference transition preserves the current
set. An actual private transition can introduce only its target label,
so a monomial of degree at most four enlarges the retained union by at
most four labels. After $b$ perturbation factors, that union has size
at most $K+4b$, even if some labels are later removed from the current
set.

Both $F$ and $E$ are lower triangular. Each $F$ factor strictly
advances the index, and intervening $E$ factors cannot reverse it.
The $f$ reference targets consequently have distinct permanent labels
in the retained union, giving $f\le K+4b$. For adjoint words the indices
decrease, and the adjoint blocks in Eq.~\eqref{eq:allcorners} have the
same one-label introduction property. The identical bound therefore
holds in that direction. If it fails, every resolved path vanishes,
and hence the complete operator word vanishes. Set $e:=\norm E$.
For a surviving word, only $\norm F\le1$ and $\norm E=e$ enter the norm
estimate. The following count is over the positions of its $E$ factors;
the resolved occurrence paths and bath configurations establish
vanishing and contribute no further factor.

With $f+b=q=7u$ and $K=2u$, every surviving product has $b\ge u$.
The small-action assumption gives $e\le1/(16\cdot7)$, so bounding each
complete product by $e^b$ yields
\[
 \norm{(F+E)^qQ_{\le2u}}
 \le\sum_{b\ge u}\binom qb e^b
 \le8^{-u}(1+8e)^{7u}\le4^{-u}.
\]
The last inequality uses $(1+1/14)^7\le\sum_{j=0}^72^{-j}\le2$.
The adjoint words satisfy the same bound on $b$ on their own input
domain, as shown above; the identical binomial estimate proves the
second inequality.
\end{proof}

For a physical bath label $b$, write $N_{\ne b}=N_{\rm cat}-P_b$ and let $Q_{N>2u}=\mathbf1_{(2u,\infty)}(N_{\rm cat})$ denote the physical occupation tail projection. For a formal Hamiltonian label, $N_{\ne b}=N_{\rm cat}$.

\begin{lemma}[Occupation moment of the catalyst]\label{lem:catalyst}
Under the common visit and interaction-weight bounds of
Section~\ref{sec:prefix},
\[
 (\Gamma Y)^\dagger2^{N_{\rm aug}}\Gamma Y\preceq M_\Gamma^2I,
 \qquad M_\Gamma^2=(5/2+2\nu^2)\Sigma e^{\nu\Sigma/4}.
\]
Also $\norm{Q_{N>2u}Y}\le e^{\nu\Sigma/8}2^{-u}$.
\end{lemma}
\begin{proof}
The private occurrence summands are orthogonal, so their weighted
squared norms add. In the summand labelled by $b$, the private label
already contributes one to $N_{\rm aug}$, whether or not that bath is
physically occupied. Its weight is therefore exactly
$2\,2^{N_{\ne b}}$. For a jump occurrence $j$, the local catalyst
$\gamma_j$ commutes with all other-bath occupation factors; applying
Eq.~\eqref{eq:metricbound} in the selected bath gives
\[
 \gamma_j^\dagger(2\,2^{N_{\ne b}})\gamma_j
 \preceq(\tfrac52h\beta^2P_0+8P_b)2^{N_{\ne b}}.
\]
Theorem~\ref{thm:prefix} supplies one set of $a_b$ for all the physical
prefixes in this sum. The unmarked moment bound consequently makes the
vacuum term at most $(5/2)h\beta^2e^{\nu\Sigma/4}$, while the marked
bound in Eq.~\eqref{eq:markedmoment} makes the occupied term at most
$8a_b^2e^{\nu\Sigma/4}$. If $\mathcal J$ is the set of jump
occurrences in the window, the labels $b(j)=(e,m)$ are distinct.
Thus their marked contributions sum to at most $\sum_ba_b^2$.
We retain the weaker $\nu$-weighted bound, with $\nu\ge1$, to keep
the uniform constant below:
\[
 \sum_{j\in\mathcal J}a_{b(j)}^2
 \le\nu\sum_ba_b^2\le\nu^2\Sigma/4.
\]
The H terms contribute at most $2h\alpha e^{\nu\Sigma/4}$ each.
The sums of the window weights $h\alpha$ and $h\beta^2$ are bounded
by the common $\Sigma$, so the H and vacuum terms together contribute
at most $(5/2)\Sigma e^{\nu\Sigma/4}$. Adding the marked terms gives
$M_\Gamma^2$, uniformly in the number of periods.

The physical tail estimate is separate: taking $z=2$ in
Eq.~\eqref{eq:physmoment} and applying the spectral cutoff at $2u$
gives the stated bound on $Q_{N>2u}Y$.
\end{proof}
For an integer $u\ge1$ and $q=7u$, we apply the private-block decay
to the catalyst by splitting $\Gamma Y$ at augmented occupation $2u$. With $Q^{\rm aug}_{>2u}:=I-Q_{\le2u}$
on the private space, Lemma~\ref{lem:catalyst} gives
\[
 \norm{\Gamma Y}\le M_\Gamma,\qquad
 \norm{Q^{\rm aug}_{>2u}\Gamma Y}\le M_\Gamma2^{-u}.
\]
Lemma~\ref{lem:packet} bounds the low part by $M_\Gamma2^{-2u}$;
contractivity bounds the high part by $M_\Gamma2^{-u}$. Their sum gives
\begin{equation}\label{eq:cyclicpacket}
 \norm{p(S_{11})^q\Gamma Y},\qquad
 \norm{p(S_{11}^\dagger)^q\Gamma Y}\le2M_\Gamma2^{-u}.
\end{equation}
The second estimate controls the adjoint residual on the target
isometry $VY$. The inverse transducer has catalyst $\Gamma V^\dagger$,
and $\Gamma V^\dagger VY=\Gamma Y$. The same bound on physical
prefixes therefore controls both residual directions.

\subsection{Finite reuse and oblivious amplitude amplification}\label{sec:reuse}\label{sec:finite-reuse}
Finite reuse couples $N$ public summands to one shared private summand
\cite[Lemma~4]{ChenHamiltonian}, \cite[Lemma~5.1]{ChenLindblad}.
The initial state has zero amplitude in the private summand. Preparing
and unpreparing a uniform public index makes the selected public block
the average of the $N$ public outputs; the complete output retains
the private amplitude.
We combine different reuse lengths to obtain the error polynomial
$p(S_{11})^qg_{6q}(S_{11})$. Its decay suppresses the residual, and
the coefficient normalization determines the amplification step.
\begin{lemma}[Finite reuse]\label{lem:imported-reuse}
Let $S=\left(\begin{smallmatrix}S_{00}&S_{01}\\S_{10}&S_{11}\end{smallmatrix}\right)$
be unitary and satisfy $S(I,\Gamma)^{\mathsf T}=(V,\Gamma)^{\mathsf T}$,
where $V$ is unitary on the public space. For an integer $N\ge1$, consider the circuit that applies $S$
successively between each of $N$ public summands and one common private
summand, with preparation and unpreparation of the uniform public
superposition. Its selected operator is
\begin{equation}\label{eq:imported-reuse}
 P_N=V-S_{01}g_N(S_{11})\Gamma,\qquad
 g_N(z):=\frac1N\sum_{j=0}^{N-1}z^j.
\end{equation}
\end{lemma}
\begin{proof}
In the present application the public space is the complete system and
bath, represented by one data register and a reuse index. This
specializes \citet[Eq.~(30)]{WangYe} to the identity input embedding.
For input $x$, let $z_k$ and $b_k$ be the private amplitude after slot
$k$ and its public output, both rescaled by $\sqrt N$. Uniform public
input and zero private amplitude give
\[
 z_0=0,\qquad z_k=S_{10}x+S_{11}z_{k-1},\qquad
 b_k=S_{00}x+S_{01}z_{k-1}.
\]
Using $(I-S_{11})\Gamma=S_{10}$ solves this recurrence:
\[
 z_{k-1}=(I-S_{11}^{k-1})\Gamma x,\qquad
 b_k=Vx-S_{01}S_{11}^{k-1}\Gamma x.
\]
Each actual slot has the factor $1/\sqrt N$.
Unpreparing the uniform index selects the public amplitude
$N^{-1}\sum_{k=1}^N b_k$, proving Eq.~\eqref{eq:imported-reuse}.
The private amplitude remains in the output and is included in the
amplification error below. For the adjoint, $S^\dagger$ has target
$V^\dagger$ and catalyst $\Gamma V^\dagger$. To relate this to the adjoint of the saved circuit,
write $S^{[k]}$ for the operation on public slot $k$ and the shared
private summand, and
$C_N(S):=S^{[N]}\cdots S^{[1]}$. Let $J_N$ reverse the public slots,
$k\mapsto N+1-k$, and fix the private summand. Then
\[
 J_NC_N(S)^\dagger J_N^\dagger=C_N(S^\dagger).
\]
The uniform public state is fixed by $J_N$, so these circuits have the
same selected block. This is an analytical identity; the implemented
adjoint simply reverses the saved gates and need not implement $J_N$.
Consequently,
\begin{equation}\label{eq:reuseadj}
 P_N^\dagger=V^\dagger-S_{10}^\dagger g_N(S_{11}^\dagger)\Gamma V^\dagger.
\end{equation}
\end{proof}
Equation~\eqref{eq:reuseadj} is the reverse selected block on the target
input $VY$ associated with source $Y$. Section~\ref{sec:window-error}
establishes the finite-precision inverse bound on the corresponding
input domain.

We write $\mathscr U_N$ for the complete reuse unitary. The public and private
flag $\mathsf P$ and register $\mathsf M$ of width
$w_M=\lceil\log_2(10q+1)\rceil$ encode its direct sum.
Appendix~\ref{app:reuse-registers} specifies the controlled calls and
public-conditioned preparations for superposed lengths, unused branches
and the saved adjoint. The following coefficients are the smoothed
weights of \citet[Eqs.~(93)--(94)]{WangYe}.
\begin{lemma}[Smoothed coefficients]\label{lem:imported-coefficients}
Let $q\ge1$ be an integer, $p(z):=(z^2+z^4)/2$, and
$p(z)^q=\sum_k c_kz^k$, where absent coefficients are zero. Set
\[
 \lambda_N:=\frac{N}{6q}(c_{N-6q}-c_N),\qquad 1\le N\le10q.
\]
These coefficients satisfy
\begin{equation}\label{eq:lambda}
 \sum_N\lambda_Ng_N(z)=p(z)^qg_{6q}(z),\qquad
 \sum_N\lambda_N=1,\qquad \sum_N|\lambda_N|=2.
\end{equation}
\end{lemma}

\begin{proof}
The degree-$j$ coefficient of $\sum_N\lambda_Ng_N$ is
$\sum_{N>j}\lambda_N/N$. Substituting the stated consecutive
differences gives the degree-$j$ coefficient of $p^qg_{6q}$.
The two disjoint coefficient supports are
\[
 \lambda_{2q+2j}=-\frac{2q+2j}{6q}\,2^{-q}\binom qj,\qquad
 \lambda_{8q+2j}=\frac{8q+2j}{6q}\,2^{-q}\binom qj,
 \quad0\le j\le q.
\]
The binomial mean $q/2$ gives signed masses $-1/2$ and $3/2$,
proving both sums in Eq.~\eqref{eq:lambda}.
\end{proof}
An exact binomial recurrence computes the $O(q)$ entries. Their absolute
sum is two, which permits one amplification step, and no reuse length
exceeds $10q$. The recursive preparation of
\citet[Theorems~8--9]{SBM} realizes the corresponding amplitude vector.
\begin{lemma}[State preparation]\label{lem:imported-preparation}
For a specified normalized nonnegative real vector on $b$ qubits,
recursive state preparation gives a unitary mapping $\ket{0^b}$ exactly
to that vector using $O(2^b)$ one-qubit rotations and CNOTs.
\end{lemma}
\paragraph{Coefficient preparation and selection.}
Lemma~\ref{lem:imported-preparation} on $\mathsf N$, of width
$w_N=\lceil\log_2(10q+1)\rceil$, gives
\begin{equation}\label{eq:coefficient-prepare}
 \mathsf{PREPARE}\ket{0^{w_N}}
 :=\sum_{N=1}^{10q}\sqrt{|\lambda_N|/2}\,\ket N.
\end{equation}
The $O(q)$ rational probabilities determine all angles and an $O(q)$
gate cost; subtrees with zero probability have angle zero.
Appendix~\ref{app:reuse-registers} specifies finite angle evaluation
and the convention for the error on the full unitary.

Let $\sigma_N:=-1$ for $2q\le N\le4q$ and $\sigma_N:=1$ otherwise;
on the nonzero coefficient support this is $\operatorname{sgn}\lambda_N$.
We define the complete selection unitary and LCU circuit by
\begin{equation}\label{eq:prepareselect-full}
 \mathsf{SELECT}:=\sum_{N=0}^{2^{w_N}-1}\ket N\bra N
                    \otimes\sigma_N\mathscr U_N,
 \qquad U_o:=(\mathsf{PREPARE}^\dagger\otimes I)
              \mathsf{SELECT}(\mathsf{PREPARE}\otimes I),
\end{equation}
Here $\mathscr U_N:=I$ outside $1\le N\le10q$, and branches preserve $N$.
Appendix~\ref{app:reuse-registers} specifies the
sign predicate with helper registers initialized to zero; saved gate lists define a complete unitary
on every helper state, with Eq.~\eqref{eq:prepareselect-full} describing
the action on the data for initialized helper registers.

Each reuse unitary has selected block $P_N$ with normalization one.
The coefficient absolute sum is $\sum_N|\lambda_N|=2$, so the LCU rule
\cite{ChildsWiebe}, \cite[Section~4.3, Lemma~52]{QSVT}
gives the block $M_q/2$ with exact ideal preparation
\cite[Lemma~6.1]{ChenLindblad}, \cite[Lemma~8]{WangYe}.
Let $E_{\rm sig}$ embed data into public with all outer labels and work
zero, and $P_o=E_{\rm sig}E_{\rm sig}^\dagger$. Then
\[
 E_{\rm sig}^\dagger U_oE_{\rm sig}=M_q/2,\qquad
 M_q:=\sum_N\lambda_NP_N.
\]
The sign is implemented in $\mathsf{SELECT}$ as in
Eq.~\eqref{eq:prepareselect-full}. Work and reuse labels stay coherent until
uncomputed. The signal projection is identity on physical data;
Eq.~\eqref{eq:signal-registers} gives its registers, and finite
preparation errors are bounded on the complete implemented unitary.

For the source--target pair $x=E_{\rm sig}Y\psi$,
$y=E_{\rm sig}VY\psi$, define
\[
 e_f=P_oU_ox-y/2,\qquad e_r=P_oU_o^\dagger y-x/2.
\]
Both residuals enter the three-call amplification circuit
\cite{BCCKS}, \cite[Lemma~7.1]{ChenLindblad}:
\[
 R_o=I-2P_o,\qquad \mathcal F=-U_oR_oU_o^\dagger R_oU_o.
\]
Its complete output, including the component outside signal, obeys
\begin{equation}\label{eq:oaaresidual}
 \mathcal Fx-y=-2U_oe_r+2(I-2U_oP_oU_o^\dagger)e_f.
\end{equation}
Indeed, expanding both reflections with $P_ox=x$ gives
$\mathcal Fx=U_ox+2(I-2U_oP_oU_o^\dagger)P_oU_ox$; substituting
$P_oU_ox=y/2+e_f$ and $P_oU_o^\dagger y=x/2+e_r$ yields the identity.
This reverse residual is needed for amplitude amplification. To
implement a physical inverse window, we instead apply the forward
error estimate to its preimage and use the adjoint of the saved circuit.

For integer $q\ge1$, the corresponding residual maps are
\[
 \mathcal R_{q,+}:=S_{01}p(S_{11})^qg_{6q}(S_{11})\Gamma,\qquad
 \mathcal R_{q,-}:=S_{10}^\dagger
 p(S_{11}^\dagger)^qg_{6q}(S_{11}^\dagger)\Gamma.
\]
\begin{proposition}[Residual-based implementation bound]\label{prop:residualimplementation}
Let $Y$ be an isometry from an input space into the physical data space, and suppose
$\norm{\mathcal R_{q,+}Y}\le r_+$ and
$\norm{\mathcal R_{q,-}Y}\le r_-$. With the signal embedding defined above,
\[
 \norm{\mathcal F E_{\rm sig}Y-E_{\rm sig}VY}\le r_++r_-.
\]
If a finite implementation on the same full comparison space satisfies
$\norm{\widehat{\mathcal F}-\mathcal F}\le\eta_{\rm term}$,
then
\[
 \norm{\widehat{\mathcal F}E_{\rm sig}Y-E_{\rm sig}VY}
 \le r_++r_-+\eta_{\rm term}.
\]
The estimates are uniform in an untouched reference.
\end{proposition}
\begin{proof}
On the same source--target pair, the selected-block identities give
$e_f=-E_{\rm sig}\mathcal R_{q,+}Y\psi/2$. Equation~\eqref{eq:reuseadj}
for the saved adjoint and $V^\dagger VY=Y$ give
$e_r=-E_{\rm sig}\mathcal R_{q,-}Y\psi/2$.
In Eq.~\eqref{eq:oaaresidual}, $U_o$ and
$I-2U_oP_oU_o^\dagger$ are unitary, so the complete output error is at
most $2\norm{e_f}+2\norm{e_r}$. The full-space triangle inequality
adds $\eta_{\rm term}\norm\psi$, including orthogonal signal and work
components; tensoring a reference preserves the bound.
\end{proof}
The averaging contractions $g_{6q}(S_{11})$ and its adjoint polynomial
commute with their respective $p^q$, so they may act after the filtered
catalysts in Eq.~\eqref{eq:cyclicpacket}. Contractivity of $S_{01}$ and
$S_{10}^\dagger$ gives $r_+,r_-\le2M_\Gamma2^{-u}$, hence
\begin{equation}\label{eq:outererror}
 \norm{\mathcal F E_{\rm sig}Y-E_{\rm sig}VY}\le4M_\Gamma2^{-u}.
\end{equation}
This ideal bound holds on the original physical input $Y$.
Section~\ref{sec:window-error} combines it with the finite circuit
estimate on bounded occupation and the physical occupation tail.
The amplification sign is implemented by $XZXZ=-I$ on an existing
qubit. Section~\ref{sec:implementation} realizes this circuit with finite bath
lists and numerical precision while retaining the shared environment.
\section{Circuit implementation and global error analysis}\label{sec:implementation}
We now implement the operators $S$ and $\mathcal F$ of
Section~\ref{sec:compiler} on finite registers. Their targets are
unchanged; the bath configurations are represented by sorted lists of
nonzero values, and angles are evaluated to finite precision.
Capacities, permanent labels and register widths are fixed for the
entire segment, so overlapping factors use the same bath encoding.
We first prove the circuit estimate on bounded-occupation inputs,
then combine it with the physical tail and residual estimates.

Let $v$ bound the number of system qubits, the number of associated
terms and the volume of the support box, with $r\le C_D(1+v)$.
Let $w_{\rm mesh}=\log_2M$, and let $w_{\rm reg}$ bound the number of
bits allocated to each occurrence address and bath label. Finally,
let $\Sigma$ bound the interaction weights for every family of calls
in Theorem~\ref{thm:prefix}, including inverse preimages and their
forward prefixes. We use these bounds to fix the register capacities
and their spatial layout. For $0<\eta<1$, define
\begin{equation}\label{eq:local-size}
 Z_*:=2+v+w_{\rm reg}+\lceil\Sigma\rceil+
       \left\lceil\log_2(1/\eta)\right\rceil.
\end{equation}
For the global construction, let $v_{\rm gl}$ be the maximum of the
system-qubit counts, complete term-list sizes and support-box volumes
used by the factors. Let $w_{{\rm reg},\rm gl}$ be their common allocated
address and bath-label width. With the common $\Sigma_*$ from
Theorem~\ref{thm:prefix} and the common factor budget $\eta$, define
\begin{equation}\label{eq:global-size}
 Z_{\rm gl}:=2+v_{\rm gl}+w_{{\rm reg},\rm gl}
       +\lceil\Sigma_*\rceil+\left\lceil\log_2(1/\eta)\right\rceil.
\end{equation}
Thus $Z_{\rm gl}$ is the value of Eq.~\eqref{eq:local-size} at the common
bounds. We have
$v_{\rm gl}=O_D(v_*)$ and
\[
 w_{{\rm reg},\rm gl}\le w_{\rm mesh}+O_D(1+\log(v_*+1))
                         =O_D(1+\log X).
\]
The physical reference and the hardware use different environment
representations. Within capacity, a physical bath configuration and its
canonical list are identified in the common space of
Section~\ref{sec:commonencoding}. Higher physical occupations are retained in that comparison space for
the tail estimate, but do not correspond to allocated hardware states.
The isometry $I_P$ inserts the signal and work zero states and embeds
all physical bath configurations into this comparison space. The
superscript $\mathrm c$ denotes the corresponding unitary extension.
The theorem bounds the entire initialized output, including components
outside the signal subspace.

\Needspace{9\baselineskip}
\begin{theorem}[Coherent simulation of a local factor]\label{thm:block}
Fix the classical description of a factor $W_C$ and its windows in
the prescribed execution of Section~\ref{sec:prefix}, together with
the common bounds above. For every $0<\eta<1$, a deterministic uniform
construction produces a single unitary circuit $\widehat W_C$ with
exact one-qubit rotations on its system, complete owner environment
blocks, signal and work registers. The circuit and its encoding depend
only on these classical data and bounds, not on the ideal input
isometry $Y$. For every prescribed $Y$ whose forward and inverse call
families satisfy these common bounds,
\[
 \norm{[(\widehat W_C)^{\rm c}-W_C^{\rm c}]I_PY}\le\eta.
\]
The common space retains the complete physical bath space, including
occupations above the cutoff. The estimate holds on the prescribed
$Y$, uniformly in an untouched reference. Its logical resources satisfy
\[
 G_{\rm logical,block}\le C_D Z_*^{10},\qquad
 Q_{\rm logical,block}\le C_D Z_*^6.
\]
On a connected $H_{\rm host}$-qubit host containing these registers,
there is a nearest-neighbour implementation with
\begin{equation}\label{eq:blockcost}
 G_{\rm NN,block},\ d_{\rm NN,block}\le C_D H_{\rm host}Z_*^{10}.
\end{equation}
Logical width counts the system, environment, signal and work registers;
physical width counts every host position. The constant $C_D$ depends
only on the fixed local model and gate conventions.
\end{theorem}
We first specify the shared registers and their operations, then define
the common comparison space and combine the estimates in
Table~\ref{tab:contracts}. Finally, we place the gates on the lattice.
The approximation bounds apply on the stated input domains and include
all signal and work components of the output.

\subsection{Registers and coherent interval evaluation}
Following \citet[Sections~5 and~9.2]{ChenLindblad}, we use $\mathsf P$
to distinguish the public and private components and $\mathsf T$ to
select an occurrence. The register $\mathsf M$ stores the reuse index.
The pair $(\mathsf J,\mathsf L)$ stores permanent bath labels and
their values, with zero denoting vacuum. All register widths and the
record capacity $C_{\rm cap}$ are fixed before construction for the
entire segment. Equation~\eqref{eq:cap} applies to every call.
Neighbouring factors share the environment arrays, including occupied
baths on which a particular factor acts trivially, and use separate
work registers.

\begin{table}[H]
\centering\small
\caption{Registers for one owner. Environment labels and widths persist across windows; owners have disjoint work registers.}
\label{tab:registers}
\begin{tabular}{p{.20\linewidth}p{.70\linewidth}}
\toprule
Register & Meaning\\\midrule
$\mathsf S_{x,k}$ & System wire $k$ at site $x$.\\
$\mathsf B_{C,\ell}$ & Label bits $\mathsf J_{C,\ell,0:w_{b,C}-1}$ and constant-size nonzero value bits $\mathsf L_{C,\ell}$, $\ell\le C_{\rm cap}$.\\
$\mathsf P$ & Flag for public ($0$) and private ($1$) components.\\
$\mathsf T_{0:w_{\mathsf T}-1}$ & Occurrence index minus one, zero-extended to the shared width.\\
$\mathsf Q$ & Block-encoding ancilla for the Hermitian unitary $O$.\\
$\mathsf F$ & Extracted bath value; zero before and after an ideal selected operation.\\
$\mathsf D$ & Query descriptors: owner selector, permanent term ordinal, period, and interval endpoints. Fields use the fixed label/address widths.\\
$\mathsf H_0,\mathsf H_1$ & One hit bit per child interval: one iff an occupied active jump label belongs to that child. Both are zero before and after ideal evaluation.\\
$\mathsf A,\mathsf W$ & Fixed-point angle and Boolean arithmetic work.\\
$\mathsf N,\mathsf M$ & Reuse length and index, each with $\lceil\log_2(10q+1)\rceil$ bits.\\
\bottomrule
\end{tabular}
\end{table}

For a window containing $L\ge1$ real occurrences generated by the
active term list $\mathcal D$ of size $P$,
the local tree uses $w=\lceil\log_2L\rceil$ low address bits ($w=0$ for
$L=1$). For a standalone owner, $w_{\mathsf T}\ge w$ may be its maximum over
nonempty windows. In the assembled circuit we instead use the maximum
over all owners and zero-extend every shorter address to this common
width. Permanent bath-label fields are padded in the same way.
Valid addresses have zero unused high bits, checked
by all validity controls. These widths differ from
$w_{\rm mesh}=\log_2M$, and $L\le PM$ gives
$w\le w_{\rm mesh}+\lceil\log_2\max\{1,P\}\rceil$.
Wire names include their owner, e.g. $\mathsf A_{C,j}$ and
$\mathsf J_{C,\ell,j}$; distinct roles use disjoint wires.
A reversible finite-width evaluator uncomputes its Boolean work
exactly when its control inputs are preserved. The approximate window
can nevertheless leave amplitude outside the initialized signal and
work subspace. This amplitude is retained and propagated by subsequent
unitary windows. Terminal gate synthesis contributes a further
full-space approximation error after assembly.

Each environment block is a canonical sorted list of nonzero values
followed by records with both fields zero. In the owner's fixed complete
term order, ordinals $1,\ldots,P_C$ identify terms and zero denotes an
empty record. The encoder fixes the correspondence to physical labels.
We allocate
\[
 w_{b,C}=\lceil\log_2(P_C+1)\rceil+\log_2 M
\]
label bits. We zero-pad them to a common width
$w_b=\max_C w_{b,C}$ when accounting for a shared routine's cost. List
and query validation require every occupied record to name a jump term,
a period in $\{0,\ldots,M-1\}$ and zero padding. These fields retain
their widths across all windows; a constant-size owner selector identifies the neighbouring block. Permanent labels remain valid outside the
active window. Both persistent label and occurrence widths enter
$w_{\rm reg}$ and $Z_*$.

For $e$ in the active term list, define
$\operatorname{rank}_{\mathcal D}(e)\in\{0,\ldots,P-1\}$ to be its
zero-based position in that ordered list. This temporary rank differs
from the one-based permanent ordinal used in a bath label.
For periods $m_0\le m<m_1$ and $P>0$, $L=P(m_1-m_0)$ and
\begin{equation}\label{eq:rankdecoder}
 j-1=P(m-m_0)+\operatorname{rank}_{\mathcal D}(e).
\end{equation}
The quotient and remainder on division by $P$ and a bounded term-list scan recover the term
and period. The active list changes the occurrence rank, while the
physical term's owner, ordinal and period still name the same bath;
insertion and deletion keep the list sorted. Terms or periods outside the
window are inactive. If the active term list $\mathcal D$ is empty, the local compiler
returns the identity before occurrence decoding. An empty bath-record
array is instead a valid vacuum input to a nonempty local circuit.

\paragraph{Signal encodings and reflection.}
For hardware system and environment data $x$, define
\[
 E_{\rm pub}\ket{x}
 :=\ket{\pub}_{\mathsf P}\ket{0^{w_{\mathsf T}}}_{\mathsf T}
   \ket0_{\mathsf Q}\ket{x}_{\mathsf S\mathsf B}
   \ket0_{\mathsf F\mathsf D\mathsf H_0\mathsf H_1\mathsf A\mathsf W}.
\]
Private occurrence $j$ uses $\mathsf P=\priv$, $\mathsf T=j-1$ and
either value of $\mathsf Q$, in the same owner bath coordinates.
Outer length and reuse registers are appended below, and Boolean helpers
belong to their declared work registers. Public strings with nonzero
$\mathsf T$ or $\mathsf Q$ are padding. Selected block-encoding and private-phase
guards accept only valid private occurrences. The root rotation targets
$\mathsf P$ under $\mathsf T=0^{w_{\mathsf T}},\mathsf Q=0$, coupling
the public component to the private root.

Using uncompressed physical data $x$ gives an analytical embedding of
the complete direct sum $\mathcal K$ of Eq.~\eqref{eq:semanticspace}:
\[
 \begin{aligned}
 E_{\rm sem}x_{\pub}
 &=\ket{\pub,0^{w_{\mathsf T}},0}_{\mathsf P\mathsf T\mathsf Q}
       \otimes x\otimes\ket0_{\mathsf W_{\rm sem}},\\
 E_{\rm sem}(\ket z\otimes x)_j
 &=\ket{\priv,j-1,z}_{\mathsf P\mathsf T\mathsf Q}
       \otimes x\otimes\ket0_{\mathsf W_{\rm sem}},\qquad z=0,1.
 \end{aligned}
\]
Here $\mathsf W_{\rm sem}$ includes numerical, descriptor, exchange and
Boolean work; outer registers are untouched. Section~\ref{sec:commonencoding}
identifies capacity-limited physical data with canonical hardware lists
and keeps other physical configurations and invalid hardware strings
orthogonal. In particular, both private $\mathsf Q$ sectors are included.

Appending initialized outer registers to $E_{\rm pub}$ defines
$E_{\rm sig}$. Write $\mathsf A_o=(\mathsf N,\mathsf M)$ for the
outer reuse registers and
$\mathsf W_o=(\mathsf F,\mathsf D,\mathsf H_0,\mathsf H_1,\mathsf A,\mathsf W)$
for the extracted bath, descriptors, interval-hit bits, angle and
arithmetic work. The arithmetic work also contains temporary bits used
to prepare the outer reuse distribution. On these registers the range
projector is
\begin{equation}\label{eq:signal-registers}
 P_o:=
 \ket{\pub,0^{w_{\mathsf T}},0}\bra{\pub,0^{w_{\mathsf T}},0}_{\mathsf P\mathsf T\mathsf Q}
 \otimes\ket0\bra0_{\mathsf A_o}
 \otimes\ket0\bra0_{\mathsf W_o}
 \otimes I_{\rm system,env}.
\end{equation}
It acts as identity on all hardware data strings, including noncanonical
lists. The physical projector uses the same zero-register predicate on
the full bath space. Its reflection includes high-occupation columns;
the hardware reflection's extension is identity on physical configurations above the record capacity. They agree on the capacity subspace.

The hardware reflection uses additional zero helpers, disjoint from
$\mathsf A_o,\mathsf W_o$ and the data. These helpers are included in
the width count; the projector above is specified on their zero
subspace. The circuit computes the all-zero predicate with a Toffoli
AND network, applies $Z$ to its output and reverses the network. Its
control, output and work wires are distinct. Ideal $U_o$ preserves these
helpers' zero subspace, where the circuit acts as $R_o=I-2P_o$; the saved
gates define its unitary on all helper states. Both gate and helper
counts are linear in the signal bits. Temporary predicates are reversed
before their inputs change: in particular, the reuse-index comparison
is cleared before its second permutation, retaining private coherence.

\begin{lemma}[Bath-register exchange]\label{lem:root}
For a preserved permanent bath label $b$, let $\mathscr L[b:=v]$ assign
value $v$ to $b$, representing zero by absence, and let $v_b$ be its old
value (zero if absent). Define
$F_b(\mathscr L,v):=(\mathscr L[b:=v],v_b)$ when both lists are canonical
and fit the capacity, and fix the input otherwise. This involution has
Boolean cost $O(C_{\rm cap}^2(w_b+1))$ for record processing plus
$O((C_{\rm cap}+1)P_C(w_b+1))$ for validation of queries and records against the
owner's complete term list.
\end{lemma}
\begin{proof}
We validate the query and list (strictly increasing occupied labels, nonzero
values, zero empty fields, legal jump ordinals, periods and padding),
and save $v_b$. We remove the matching record, append $(b,v)$ if $v\ne0$
or an empty record otherwise, and sort the resulting $C_{\rm cap}+1$
records with a fixed odd--even comparison network, occupied labels first.
The last record detects overflow. We output the first $C_{\rm cap}$ records
and $v_b$ if validation succeeds without overflow, and the original
input otherwise. This finite Boolean evaluator has the stated cost.
On legal pairs it exchanges the two values; two-sided legality is
symmetric, giving two-cycles and fixed points also at the capacity boundary.
\end{proof}

\begin{lemma}[Coherent Boolean evaluation]\label{lem:booleanframe}
A bounded-fan-in Boolean network for $f$, with $s$ separately stored
intermediate bits and $t$ output bits, has a compute--copy--uncompute
circuit $U_f$ with $O(s+t)$ X, CNOT and Toffoli gates satisfying
$U_f(\ket{x}\ket z\ket{0^s})=\ket{x}\ket{z\oplus f(x)}\ket{0^s}$.
Each gate uses only its declared registers; this action and its reversed
adjoint hold on superpositions with arbitrary untouched references.
\end{lemma}
\begin{proof}
We evaluate AND into fresh zero bits with Toffoli, XOR with CNOTs and
negation with X and CNOT gates, retain the intermediates, copy the outputs and
reverse the evaluator \cite{Bennett,VBE}. Wires are distinct as required;
Appendix~\ref{app:program} fixes the exact Toffoli expansion.
\end{proof}

We evaluate $F_b$ out of place, swap the input and output data and evaluate it again.
Since $F_b^2=I$, the resulting in-place circuit satisfies
\begin{equation}\label{eq:actualrootword}
 U_{F_b}\ket x\ket0_{\rm data'}\ket0_{\rm work}
 =\ket{F_b(x)}\ket0_{\rm data'}\ket0_{\rm work}.
\end{equation}
It preserves the query coherently. Thus $U_b:=U_{F_b}O_eU_{F_b}$ extracts
the bath, applies $O_e$ to it, $\mathsf Q$ and the indicated system wires,
and returns the updated value. If every output list fits and $u_{v,s}$
are the amplitudes of $O_e$ on $v_b$ and $\ket\psi_{\mathsf Q\mathsf S}$,
\[
 U_b\ket{\Enc(\mathscr L)}\ket0_{\mathsf F}\ket\psi_{\mathsf Q\mathsf S}
 =\sum_{v,s}u_{v,s}\ket{\Enc(\mathscr L[b:=v])}
                   \ket0_{\mathsf F}\ket s_{\mathsf Q\mathsf S}.
\]
One unused record suffices because only one bath changes; the capacity
chosen below provides this margin at every ideal call. Other inputs follow the
two-sided legality rule. We clear membership work before $O_e$ changes the
bath, and occurrence work before the tree changes $\mathsf T$.
Hamiltonian terms act on $\mathsf Q\mathsf S$ without exchanges.

For interval coefficients, we write
\[
 t_e^{(0)}:=\begin{cases}
 (1-h\alpha_e/2)/(1+h\alpha_e/2),&e\in\mathcal I_H,\\
 (1-h\beta_e^2/4)/(1+h\beta_e^2/4),&e\in\mathcal I_J.
 \end{cases}
\]
These coefficients give the public-to-public blocks for Hamiltonian occurrences and jump occurrences on vacuum. For an occupied jump bath, the coefficient is zero by Eq.~\eqref{eq:through}.
\begin{lemma}[Interval products and vacuum projections]\label{lem:zerohit}
For $I=[l,r]$ in a single-pass forward window,
\begin{equation}\label{eq:zerohit}
 D_I:=\prod_{j\in I}t_j
 =d_I^{(0)}\prod_{b\in\mathcal B(I)}(I-P_b),
 \qquad d_I^{(0)}:=\prod_{e\in\mathcal D}(t_e^{(0)})^{n_e(I)}.
\end{equation}
Here $\mathcal B(I)$ contains its active jump labels and $n_e(I)$ counts
the occurrences of term $e$ in $I$.
\end{lemma}
\begin{proof}
H factors are scalar; J factors are their vacuum scalar times a diagonal
vacuum projector. Multiplication proves the identity on the complete
physical bath basis, including padded nonzero values and all spectator
factors, and hence by linearity on the full space.
\end{proof}
The trees evaluate this operator on the bath configuration \emph{after}
the selected block-encoding gate: exciting a bath changes its coefficient from
$t_b^{(0)}$ to zero. Lemma~\ref{lem:merge} includes either or both zero
child coefficients.

We intersect each padded node interval with $[1,L]$; an empty intersection
has product one and membership bit zero. Otherwise, the evaluator scans all allocated records
$\mathcal R_{\rm rec}$ in the touched owner blocks. We decode
record $\ell$ to $(v_\ell,e_\ell,m_\ell)$ and set
$j_\ell=P(m_\ell-m_0)+\operatorname{rank}_{\mathcal D}(e_\ell)+1$
for an active term and period, or zero otherwise. The hit predicate is
\[
 \bigvee_{\ell\in\mathcal R_{\rm rec}}
 \bigl[(v_\ell\ne0)\land(e_\ell\in\mathcal D)
 \land(m_0\le m_\ell<m_1)\land(l\le j_\ell\le r)\bigr].
\]
For the left and right children, this predicate is stored in
$\mathsf H_0$ and $\mathsf H_1$, respectively. The endpoints and
decoded label fields use $\mathsf D$; these descriptors and hit bits
are reversed after the rotation while their inputs are preserved.
A hit gives coefficient zero; otherwise, the evaluator computes
$d_I^{(0)}$ by binary powering. For $\rho_e=1+\operatorname{rank}_{\mathcal D}(e)$,
\[
 n_e([l,r])=\left\lfloor\frac{r-\rho_e}{P}\right\rfloor-
             \left\lfloor\frac{l-1-\rho_e}{P}\right\rfloor.
\]
We use Euclidean quotients and mark spectator terms and labels outside the window
inactive. Loops range over $P$ terms, $C_{\rm cap}$ records and $w$ bits.

Figure~\ref{fig:four-leaf} gives the four-leaf example of the following
address rule. At level $j$, bits above $\mathsf T_j$ specify the parent interval;
its representatives have zero lower bits and differ in $\mathsf T_j$.
We compute their coefficients and merge angle, rotate $\mathsf T_j$ in the
private $\mathsf Q=0$ sector with lower and unused high bits zero, then
reverse the evaluator. The higher prefix and the bath configuration
after the selected block-encoding gate remain available. Input levels run $j=0,\ldots,w-1$,
then the rotation between the public component and the root; inverse output levels run
$j=w-1,\ldots,0$. Their angles are respectively
$-2\operatorname{atan2}(d,cb)$ for $I_{AB}$ and
$2\operatorname{atan2}(da,b)$ for $O_{AB}^{\mathsf T}$.
The $s=0$ identity branch returns zero angle. For $w=0$, the selected
block-encoding gate and private phase are followed directly by the root rotation.

\paragraph{From the angle register to controlled gates.}
The evaluator writes a nonnegative magnitude in a $b_\theta$-bit
register $\mathsf A$. On basis value
$A=\sum_{j=0}^{b_\theta-1}2^jA_j$, it represents
$\theta(A):=4A/2^{b_\theta}\in[0,4)$. The stage fixes a sign
$\sigma\in\{-1,1\}$. Using the convention
$R_y(\theta)=\left(\begin{smallmatrix}
\cos(\theta/2)&-\sin(\theta/2)\\
\sin(\theta/2)&\cos(\theta/2)
\end{smallmatrix}\right)$, apply one gate
$\mathrm C_{A_j}R_y(\sigma2^{j+2-b_\theta})$ for each angle bit.
These rotations have the same axis, so the resulting unitary satisfies
\begin{equation}\label{eq:angle-register-action}
 U_\theta\sum_A\ket A\ket{\psi_A}
 =\sum_A\ket A R_y(\sigma\theta(A))\ket{\psi_A}.
\end{equation}
The address prefix, bath configuration, enable and angle bits are
unchanged by these gates. Reversing the evaluator therefore restores
its angle and arithmetic work. The $O(b_\theta)$ controlled rotations
and their Boolean guards are counted in the circuit size.
Appendix~\ref{app:arithmetic} fixes $b_\theta$ and proves the numerical
error and word-length bounds.
\subsection{A common space for the physical and encoded circuits}\label{sec:commonencoding}
The comparison must represent the same input in two ways: as an
uncompressed physical bath state, or, within capacity, as a canonical
list with the same initialized signal and work. Configurations above
capacity occur only in the physical description and are controlled by
the occupation tail. For each coarse cube $C$, let $\mathcal P_C$ be the
full physical bath space, $\mathcal H_C$ the full list-register space,
and $\mathcal L_C$ their identified capacity subspace. The physical
configurations above capacity span $\mathcal P_C^{\rm high}$, and the
noncanonical register strings span $\mathcal H_C^{\rm invalid}$.
The common space for this cube is

\[
 \mathfrak K_C:=\mathcal L_C\oplus\mathcal P_C^{\rm high}
                    \oplus\mathcal H_C^{\rm invalid}.
\]
The canonical inclusions $j_{P,C}$ and $j_{H,C}$ embed $\mathcal P_C$ and
$\mathcal H_C$, respectively, into $\mathfrak K_C$. They agree on the
canonical capacity subspace. Tensoring the owner spaces with the common
system and full signal and work registers gives one comparison space for the
entire segment. It contains every work state, with zero work specifying
the initialized input subspace. Let $\mathcal B_{\rm seg}^{\rm env}$
be the complete set of owner environment blocks, and let
$\mathcal X_{\rm seg}$ contain the common system, signal and work
registers. The global inclusions are
\[
 J_H^{\rm all}:=I_{\mathcal X_{\rm seg}}\otimes
       \bigotimes_{A\in\mathcal B_{\rm seg}^{\rm env}}j_{H,A},\qquad
 J_P^{\rm all}:=I_{\mathcal X_{\rm seg}}\otimes
       \bigotimes_{A\in\mathcal B_{\rm seg}^{\rm env}}j_{P,A}.
\]

For each fixed factor $W_C$, let $\mathcal B_C^{\rm env}$ contain
its complete owner environment blocks, including inactive baths and
occupied spectators. Let $\mathcal X_C$ contain all system, signal and
work factors used anywhere in its construction. Every constituent
unitary is first tensored with identity on the remaining factors in
these fixed sets. We then use the same inclusions
\[
 J_{H,C}=I_{\mathcal X_C}\otimes
             \bigotimes_{A\in\mathcal B_C^{\rm env}}j_{H,A},
 \qquad
 J_{P,C}=I_{\mathcal X_C}\otimes
             \bigotimes_{A\in\mathcal B_C^{\rm env}}j_{P,A}
\]
for every gate, window and complete circuit of this factor. For a
hardware unitary $U$ on these registers, define
\begin{equation}\label{eq:local-extension}
 U^{\rm c}=\bigl[J_{H,C}UJ_{H,C}^\dagger+I-J_{H,C}J_{H,C}^\dagger\bigr]
                 \otimes I_{\rm outside}.
\end{equation}
Physical unitaries use the same formula with $J_{P,C}$. The extension
symbol is relative to the fixed factor containing the operation.
Each extension is unitary, and using one inclusion within the factor
gives
\[
 (UV)^{\rm c}=U^{\rm c}V^{\rm c},\qquad
 (U^\dagger)^{\rm c}=(U^{\rm c})^\dagger.
\]
Thus the product of its saved gate extensions equals the extension of
the saved circuit, including on physical high-occupation columns.
The image projections define the unitary extensions in the comparison
space used for analysis. The hardware directly executes the saved
gates on its allocated registers.

Different factors may use different complete block sets. They still
satisfy
$U^{\rm c}J_H^{\rm all}=J_H^{\rm all}(U\otimes I_{\rm outside})$
with the one global hardware embedding, and the analogous identity
holds for physical operations. Multiplying these identities in order
identifies each complete physical or hardware execution with its
common-space execution, even for overlapping factors. The embeddings
act only on environments, so final environment traces preserve the
reduced system channel.

For the ideal LCU and amplified window, use the physical inclusion
$J_{P,C}$ and denote the resulting circuit by $\mathcal F^{\rm c}$.
It is the extension of the earlier $\mathcal F$, formed by extending
all its constituent gates with this same complete factor set. In
particular,
\begin{equation}\label{eq:amplified-extension}
 \mathcal F^{\rm c}J_P^{\rm all}
 =J_P^{\rm all}(\mathcal F\otimes I_{\rm outside}).
\end{equation}
The saved $XZXZ$ sequence representing the minus sign acts as $-I$ on
the physical image and as identity on its orthogonal complement; the
hardware extension has the corresponding action on its own image.
Thus $\mathcal F^{\rm c}$ is defined by the extended gate sequence,
not by multiplying the identity on the entire common space by $-1$.

The embedding of the direct sum with initialized work registers is
$I_{\rm sem}:=J_P^{\rm all}E_{\rm sem}$, with identity on untouched
outer registers. It includes the full direct sum of public and private components and
all physical bath configurations. On configurations fitting the
capacity, it agrees with the initialized hardware list encoding.
Lemma~\ref{lem:bandwidth} ensures that a transducer call on
$Q_{\le B}$ remains within this capacity when $C_{\rm cap}=B+3$.

For physical public inputs, $E_{\rm init}$ instead inserts zero
signal, outer and work registers while preserving every bath
configuration. We put
\[
 I_P:=J_P^{\rm all}E_{\rm init},\qquad \bar Y:=I_PY.
\]
This embedding of physical inputs includes
high-occupation columns and will be used for the complete window
comparison.

\begin{table}[H]
\centering\small
\caption{Maps used in the circuit comparison. The inclusions and their
extensions are mathematical identifications, not quantum subroutines.
Only the hardware circuits and initialization of their registers are
executed.}
\label{tab:embeddings}
\begin{tabular}{>{\raggedright\arraybackslash}p{.24\linewidth}
>{\raggedright\arraybackslash}p{.34\linewidth}
>{\raggedright\arraybackslash}p{.32\linewidth}}
\toprule
Map & Domain and output & Purpose\\\midrule
$j_{P,C}$, $j_{H,C}$ & Physical bath $\mathcal P_C$ or list-register space $\mathcal H_C$ into $\mathfrak K_C$ & Identify exactly those configurations fitting the capacity\\
$J_{P,C}$, $J_{H,C}$ & All physical or hardware registers of one fixed factor into the common space & Extend each gate and its products with one fixed inclusion\\
$J_P^{\rm all}$, $J_H^{\rm all}$ & Entire physical or hardware segment into the common space & Compare overlapping factors in one global execution\\
$E_{\rm sem}$, $I_{\rm sem}=J_P^{\rm all}E_{\rm sem}$ & Public/private direct sum into address registers with zero work, then into the common space & State the transducer circuit estimate on both private ancilla values\\
$E_{\rm init}$, $I_P=J_P^{\rm all}E_{\rm init}$ & Physical public data into zero signal, outer and work registers, then into the common space & Compare a complete window on the original input $Y$\\
\bottomrule
\end{tabular}
\end{table}

\subsection{Coherent control and finite approximation errors}
\label{sec:refinement}
We use two rules to extend estimates on initialized inputs to
coherently controlled circuits and to combine their errors. The
composition rule retains the actual subsequent unitaries, including
their action on imperfect work registers.
\begin{lemma}[Controlled lifting on a preserved control register]\label{lem:controlledlift}
Let $E$ be an initialized encoding isometry, $P$ an input projection,
and $\widehat U,V$ unitaries on a common space satisfying
$\norm{(\widehat U-V)EP}\le\delta$. For an additional control qubit,
define
\[
 \mathrm C(U):=\ket0\bra0\otimes I+\ket1\bra1\otimes U.
\]
Then
\[
 \norm{[\mathrm C(\widehat U)-\mathrm C(V)](I\otimes EP)}\le\delta.
\]
\end{lemma}
\begin{proof}
The difference on the displayed domain is
$\ket1\bra1\otimes(\widehat U-V)EP$, which has norm at most $\delta$.
For a multiplexed family satisfying the corresponding bound in each
preserved control sector, the difference is block diagonal and its
norm is the maximum branch norm. Both statements allow entanglement
between the control, data and reference registers.
\end{proof}
The preserved controls are the outer length and reuse registers,
the occurrence label during a selected operation, and the higher address
bits and bath configuration during a tree rotation.
We reverse temporary guards while their inputs remain available. The
flag distinguishing the public and private components is data because $S$ changes it.

\begin{lemma}[Telescoping bound on ideal intermediate states]\label{lem:refinementhybrid}
For an initialized input $x_0$, ideal intermediate vectors $x_j$ and
actual unitaries $\widehat U_j$, suppose
$\norm{\widehat U_jx_{j-1}-x_j}\le\delta_j\norm{x_0}$. Then
\[
 \norm{\widehat U_m\cdots\widehat U_1x_0-x_m}
 \le\Bigl(\sum_{j=1}^m\delta_j\Bigr)\norm{x_0}.
\]
\end{lemma}
\begin{proof}
Telescoping on ideal predecessors leaves actual unitary suffixes on each
local difference. They preserve its norm, including signal and work error;
the triangle inequality gives the bound.
\end{proof}
Augmented occupation acts as identity on arithmetic work.
Appendix~\ref{app:controls} computes each wide predicate once per
invocation and charges its width-dependent Boolean cost. With the
resulting enable bit preserved, controlling each saved elementary gate
has constant overhead because every such gate has fixed arity.

\label{sec:finite-precision}
Quantizing leaf tangents gives a certified lower bound for every nonzero
merge normalization. Appendix~\ref{app:arithmetic} combines the resulting
uniform finite arithmetic with list exchange and phase-sensitive local
matrix synthesis. We fix a nonempty static periodic forward window with
$L\ge1$ occurrences, active-list size $P$ and complete-list size
$P_{\rm cat}\ge P$. For integer cutoff $B\ge0$, we use the canonical encoding
with $C_{\rm cap}=B+3$ in every touched block. Write
$w=\lceil\log_2L\rceil$, let $w_{\mathsf T}\ge w$ and $w_{\rm label}$
be the actual allocated occurrence and maximal padded label widths, and set
$w_*:=\max\{w,w_{\rm mesh},w_{\rm label},w_{\mathsf T}\}$.

\begin{lemma}[Circuit implementation of the local transducer]\label{lem:directfinite}
For every $0<\xi<1$, a deterministic unitary circuit with exact one-qubit rotations
$\widehat S_{B,\xi}$ satisfies
\[
 \norm{(\widehat S_{B,\xi})^{\rm c}I_{\rm sem}Q_{\le B}
                  -I_{\rm sem}S Q_{\le B}}\le\xi.
\]
The initialized domain includes the complete direct sum of public and private components,
both private $\mathsf Q$ values and zero numerical work. Gate count and
workspace are polynomial in $P_{\rm cat}$, $B$, $w_*$ and $\log(1/\xi)$,
including complete-list scans of inactive terms and their permanent labels.
\end{lemma}
\begin{proof}
Appendix~\ref{app:arithmetic} and Eq.~\eqref{eq:bitwidth} use complete arithmetic word
width $O(w_*+\log(1/\xi)+\log(P_{\rm cat}+B+1))$ and control every
rotation uniformly over preserved addresses and the bath configuration after the selected gate.
Combine the $2w+1$ levels with Lemma~\ref{lem:refinementhybrid}, the
exchange of Lemma~\ref{lem:root} and matrix gates that preserve the specified phase of
Lemma~\ref{lem:certifiedendpoint}. Only one bath changes between
exchanges; $B+3$ covers intermediate lists. The appendix's budgets sum
below $\xi$; Section~\ref{sec:local-cost} counts their occurrences.
\end{proof}
For $L=0$, return identity before address allocation or coefficient preparation.

\subsection{Occupation cutoffs and forward and inverse window errors}\label{sec:window-error}
The occupation theorem controls physical collision prefixes. During
a transducer call, the selected local block encoding can change
one physical bath value, while both address trees preserve the bath
configuration. A forward input with occupation at most $K_{\rm call}$
therefore has intermediate physical occupation at most $K_{\rm call}+1$.
The capacity $C_{\rm cap}=K_{\rm call}+3$ accommodates extraction and
replacement in both directions, giving agreement with the physical
exchange on the inputs used below.

We next derive the inverse estimate from the forward guarantee, keeping
track of the margin of one in the occupation cutoff. Let
$\widehat S_{K_{\rm call},\xi}$ be the stored transducer circuit with
error $\xi$ on all encoded inputs of augmented occupation number at most
$K_{\rm call}$. For an inverse input $y$ in the direct sum, with augmented
occupation at most $K_{\rm call}-1$, Lemma~\ref{lem:bandwidth} places
$x:=S^\dagger y$ in $Q_{\le K_{\rm call}}$. Unitarity in the common
space therefore gives
\[
 \norm{\bigl[(\widehat S_{K_{\rm call},\xi})^{\rm c}\bigr]^\dagger
           I_{\rm sem}y-I_{\rm sem}S^\dagger y}
 =\norm{I_{\rm sem}y-(\widehat S_{K_{\rm call},\xi})^{\rm c}
           I_{\rm sem}x}\le\xi\norm y.
\]
This is the inverse-input domain used below.

We choose one capacity for the entire amplified window. Each completed
call adds at most one augmented label, and every inverse call requires
the additional one-label input margin just established. For a window
error budget $\delta$, the following parameters first suppress the
physical and catalyst tails, then bound the occupation increase through
all reuse calls:

\[
 A_*=16(\nu+1)^2(1+\lceil\Sigma\rceil),\quad
 u=A_*+\left\lceil\log_2\frac{512}{\delta}\right\rceil,\quad q=7u,
\]
\begin{equation}\label{eq:cap}
 L_0=2u,\qquad K_{\rm call}=L_0+30q+1,\qquad C_{\rm cap}=K_{\rm call}+3.
\end{equation}
\begin{table}[H]
\centering\small
\caption{How occupation estimates determine one fixed register capacity.
Physical occupation counts occupied baths; augmented occupation also
counts the selected private occurrence label. They agree in the public
input branch.}
\label{tab:cutoffs}
\begin{tabular}{>{\raggedright\arraybackslash}p{.23\linewidth}
>{\raggedright\arraybackslash}p{.67\linewidth}}
\toprule
Parameter & Role\\\midrule
$L_0=2u$ & Physical input cutoff on $Y$; the catalyst analysis separately uses augmented cutoff $2u$ on $\Gamma Y$\\
$30q$ & Maximum number of completed transducer or adjoint calls in the full three-call amplification\\
$K_{\rm call}=L_0+30q+1$ & Augmented cutoff for a forward circuit, with one-label margin for inverse inputs; this is $B$ in Lemma~\ref{lem:directfinite}\\
$C_{\rm cap}=K_{\rm call}+3$ & Persistent records per owner, with slack for bath extraction and replacement; fixed for all calls. Additional reversible work is counted separately\\
\bottomrule
\end{tabular}
\end{table}
The elementary inequalities $\log x\le x$ and $\log2\ge1/2$ imply
$M_\Gamma\le2^{A_*}$ and $e^{\nu\Sigma/8}\le2^{A_*}$.
In the full outer circuit, the $30q$ completed calls can increase the
initial augmented occupation by at most $30q$. Hence every ideal call
input has occupation at most $L_0+30q\le K_{\rm call}-1$; port
permutations, outer labels and signal reflections preserve this count.
The ideal direct subroutines restore their numerical and exchange work,
so both forward and inverse estimates apply at each call at the fixed
capacity in Table~\ref{tab:cutoffs}. Actual work errors are retained and
propagated by the actual unitary suffixes of
Lemma~\ref{lem:refinementhybrid}.

Each transducer call receives budget $\xi=\delta/(240q)$, and a further
$\delta/8$ covers finite approximations in the outer preparation and evaluation circuits.
We count every rotation outside the direct calls in the full three-call
circuit, including repeated and inverse occurrences, and denote this
count by $J_{\rm out}$. Each known angle is approximated within
$\delta/[8\max\{1,J_{\rm out}\}]$, with the opposite sign used in its
adjoint. Lemma~\ref{lem:refinementhybrid} charges at most $\delta/16$
for these rotations; Boolean evaluations and permutations are exact.
Together with the errors in direct calls, the same lemma bounds the outer
circuit error by $\delta/4$ on the low-occupation inputs.

There are now two approximation estimates with different domains. The
finite compiler is uniform on low-occupation inputs, whereas the ideal
amplified circuit approximates $V$ on the original physical column $Y$
through its catalyst residual. The occupation bound connects these
estimates: we split the compiler error at a physical occupation cutoff,
bound the high-occupation contribution by unitarity, and apply the residual estimate
to the unsplit column $Y$. This order preserves the hypothesis on physical prefixes required by Lemma~\ref{lem:catalyst}. Let $Q_0$
project onto occupation at most $L_0$ in the complete bath set
of $\mathcal B_C^{\rm env}$, including inactive spectators. The
associated supports lie in the fixed volume $v_*$, so
Theorem~\ref{thm:prefix} bounds the tail. This cutoff bounds each owner
block's occupation by $L_0<C_{\rm cap}$.

The image projection $\bar Q_0$ is zero on the complement of the local
physical image and identity on untouched factors; hence
$\bar Q_0I_P=I_PQ_0$. Let $\widehat V$ be the common-space extension of
the saved finite amplified-window circuit and $V^{\rm c}$ the physical
unitary extension. We split the compiler discrepancy into its
low-occupation and tail parts, leaving the physical residual on the
original column:

\begin{align*}
 (\widehat V-V^{\rm c})\bar Y
 ={}&(\widehat V-\mathcal F^{\rm c})\bar Q_0\bar Y\;+
 (\widehat V-\mathcal F^{\rm c})(I-\bar Q_0)\bar Y\\
 &+(\mathcal F^{\rm c}-V^{\rm c})\bar Y.
\end{align*}
The first term is at most $\delta/4$ by the uniform compiler estimate
on $Q_0Y$. Unitarity of both complete circuit extensions bounds the
second term by $2\norm{(I-Q_0)Y}$. For the third term,
Eq.~\eqref{eq:outererror} applies to the original physical column $Y$,
whose prefix representation supplies the catalyst estimate. Since
$2^{A_*-u}\le\delta/512$, the tail and physical residual terms are at
most $\delta/256$ and $\delta/128$, respectively. Adding them to the
compiler contribution $\delta/4$ gives the complete window estimate
\begin{equation}\label{eq:compiledwindow}
 \norm{\widehat V I_PY-I_PVY}<\delta
\end{equation}
on the original physical prefix $Y$ with the stated occupation bounds.

We finally consider a requested physical inverse on input $Y_{\rm in}$. We apply
the forward estimate to its preimage $Y=V^\dagger Y_{\rm in}$;
Theorem~\ref{thm:prefix} provides the occupation bound for precisely
this column. Writing $\bar Y_{\rm in}=I_PY_{\rm in}$, unitarity gives
\[
 \norm{\widehat V^\dagger\bar Y_{\rm in}-(V^{\rm c})^\dagger\bar Y_{\rm in}}
 =\norm{\bar Y_{\rm in}-\widehat V\bar Y}\le\delta,
 \qquad \bar Y=(V^{\rm c})^\dagger\bar Y_{\rm in}.
\]
The adjoint of the saved forward circuit therefore implements the
required physical inverse on $Y_{\rm in}$. A local factor has at most $2r$
windows; with $\delta=\eta/(4r)$, Lemma~\ref{lem:refinementhybrid}
combines them in the order of Eq.~\eqref{eq:regroup} with error below
$\eta$.

\subsection{Local gate count, workspace and geometric placement}\label{sec:local-cost}
We use reversible addition \cite{VBE,Cuccaro}, together with schoolbook
multiplication, restoring division and integer square-root algorithms
\cite{BrentZimmermann}. Appendix~\ref{app:arithmetic}
constructs the common complete word width $b_{\rm arith}$, including
sign and integer bits. It gives $O(b_{\rm arith}^{5/2})$ Boolean gates
per sine enclosure and $O(b_{\rm arith}^{7/2})$ per angle inversion;
we use the looser $O(b_{\rm arith}^{6})$ bound.
Retaining one bit per Boolean intermediate gives the same workspace bounds.
Each size and precision parameter used in a window is $O_D(Z_*)$: $P_{\rm cat}\le C_Dv$,
$P\le P_{\rm cat}$, address widths are at most $w_{\rm reg}$, and the
budgets bound $C_{\rm cap}$, $q$ and $b_{\rm arith}$. Subdivision adds
only $\log r=O(\log(v+1))$ precision bits. This gives
Table~\ref{tab:local-costs}; terminal synthesis is charged once after assembly.

\begin{table}[H]
\centering\small
\caption{Logical gates with exact one-qubit rotations for one window, including controls, inverse occurrences and accuracy parameters. Peak workspace is separate; finite-alphabet synthesis follows assembly.}
\label{tab:local-costs}
\begin{tabular}{p{.62\linewidth}p{.25\linewidth}}\toprule
Operation & Logical gates\\\midrule
Bath-register exchange with full-list validation & $O_D(Z_*^3)$\\
Full term-list scan for the selected local operation & $O_D(Z_*^4)$\\
Interval occupation test and vacuum product & $O_D(Z_*^5)$\\
Controlled normalized-angle evaluation & $O_D(Z_*^6)$\\
$2w+1$ tree levels, selected operation and phases & $O_D(Z_*^7)$\\
Reuse preparations for $N\le10q$ and the length state & $O_D(Z_*^4)$\\
Full window with $30q$ direct calls and reflections & $O_D(Z_*^8)$\\\bottomrule
\end{tabular}
\end{table}

Each of $P$ interval terms uses $O(w+1)$ multiplications of $O(Z_*)$-bit
numbers, at $O(b_{\rm arith}^3)$ Boolean cost, giving $O_D(Z_*^5)$; the record scan
is smaller. At each tree level one reversible evaluator handles all
parent prefixes coherently, so its $O_D(Z_*^6)$ angle cost is paid once
per level. The $2w+1$ levels give $O_D(Z_*^7)$ per call; a preserved
enable adds constant overhead per elementary gate. The $30q=O_D(Z_*)$
calls, preparations and reflections fit $O_D(Z_*^8)$ per window.
Levels reuse $O_D(Z_*^6)$ work bits after ideal uncomputation, with
finite errors covered by Lemma~\ref{lem:refinementhybrid}. This width
dominates the records, labels and helpers. The $2r=O_D(v)$ windows reuse
these registers and cost $O_D(Z_*^9)$ gates. Since $Z_*\ge1$,
we relax the exponent to ten to use the same conservative bound for
all descriptor and control overhead in the block theorem:
\[
 G_{\rm logical,block}\le C_DZ_*^{10},\qquad
 Q_{\rm logical,block}\le C_DZ_*^6.
\]
The mesh enters through logarithmic address widths.

We fix geometric homes using common $Z_{\rm gl}$. For owner $A$, let $W_A$
count its permanent environment bits, own work and helpers, so
$W_{\max}:=\max_AW_A\le C_DZ_{\rm gl}^6$; neighbouring arrays are charged
to their respective owners. We choose the least positive integer $f$ with
$(f^D-d_0)R^D\ge W_{\max}$ by integer-power comparisons. We reserve the
first $d_0$ lexicographic offsets $u_k\in\{0,\ldots,f-1\}^D$ in the refinement
cell of each original site $x$, for qubits $(x,k)$ at $fx+u_k$, counting absent-site offsets
as padding. Each owner places its environment and work registers at the first remaining
cube vertices. Homes persist for the entire segment, including idle owners.

The complete host of $C$ is the connected box
\[
 \mathbb Z^D\cap\bigl(fRC+[-3fR,4fR)^D\bigr).
\]
Minimality gives $f^DR^D\le2^D(d_0R^D+W_{\max})$, hence
\[
 H_{\rm host}=(7fR)^D=O_D(R^D+W_{\max})=O_D(Z_{\rm gl}^6).
\]
The host union and idle system homes use at most
$nH_{\rm host}+d_0n=O_D(nZ_{\rm gl}^6)$ vertices, including spare routing
positions. Same-colour owners differ by a nonzero multiple of eight
in some coarse coordinate, so their complete hosts are disjoint.
Lemma~\ref{lem:route} restores homes at cost at most $6H_{\rm host}+1$
per logical gate, proving Eq.~\eqref{eq:blockcost} and globally
\begin{equation}\label{eq:globalblockcost}
 G_{\rm NN,block},\ d_{\rm NN,block}\le C_DZ_{\rm gl}^{16}.
\end{equation}

\begin{proof}[Proof of Theorem~\ref{thm:block}]
Section~\ref{sec:prefix} provides the $2r$ windows in
Eq.~\eqref{eq:regroup}, each with action at most
$s_0=(256\cdot7)^{-2}$. Their common visit sequences have
$\nu=2\cdot7^D+4$ and interaction weight at most $\Sigma$.
The catalyst estimates apply to forward inputs and to the forward
evolution from inverse preimages. With $\delta=\eta/(4r)$,
Section~\ref{sec:window-error} combines the circuit estimate at low
occupation with the tail and residual bounds on the original input
isometry. This gives Eq.~\eqref{eq:compiledwindow} and the corresponding
estimate for the adjoint of the saved circuit.
Lemma~\ref{lem:refinementhybrid} sums these errors on the prescribed
ideal inputs. Since $r\le C_D(1+v)$ and
$\log(1/\delta)=\log(1/\eta)+O_D(1+\log(v+1))$, all cutoff parameters,
reuse indices, addresses and arithmetic registers require
$O_D(Z_*)$ bits or entries as appropriate. The table and preceding
calculation give the logical gate count and width.
Lemma~\ref{lem:route} then gives Eq.~\eqref{eq:blockcost}, with every
host position included in the physical width. All circuit data are
functions of the classical description and the chosen common bounds;
$Y$ and the labelwise bounds $a_b$ enter only the error analysis.
\end{proof}

\subsection{Routing and assembly on the lattice}
\label{sec:assembly}
Fixed sequences of SWAP gates restore the logical positions even
when the work registers and spectator qubits are entangled
\cite{ACG,CSU,Beals}. The following elementary bound suffices because
each host contains only polylogarithmically many qubits.
\begin{lemma}[Routing and restoration of qubit positions]\label{lem:route}
A circuit of $g$ one- and two-qubit gates on a connected $H$-qubit host has a
nearest-neighbour implementation with at most $(6H+1)g$ gates and depth.
Every replacement restores all homes and equals the original gate
tensored with identity on the remaining wires, including entangled spectators.
\end{lemma}
\begin{proof}
One-qubit gates act at home. For ordered operands $a,b$, fix a simple
path $v_0=a,\ldots,v_d=b$. Its first $d-1$ adjacent SWAPs implement a
permutation $\Pi$ moving $a$ to $v_{d-1}$ and fixing $b$. We apply the gate
and reverse these SWAPs:
\[
 \Pi^\dagger U_{v_{d-1},v_d}\Pi=U_{a,b}.
\]
With three CNOTs per SWAP and $d\le H-1$, the cost is
$6(d-1)+1\le6H+1$. We serialize replacements and multiply these full-space
identities; the identities hold for arbitrary routing states.
\end{proof}
Disjoint same-colour hosts run their block circuits in parallel despite
serialization within each host.

\paragraph{Assembly of the saved circuits.}
The compiler follows the order in
Section~\ref{sec:overview}, paragraph ``Circuit generation and execution.''
It appends each factor's saved old-window adjoints and new forward
windows in the order of Eq.~\eqref{eq:regroup}, then assembles the
disjoint factor circuits by colour. All controls, routes and finite-gate
replacements are generated before segment initialization. Only the
execution stage applies the completed segment and refreshes its
auxiliaries; no window or factor resets the shared bath or work.

Table~\ref{tab:cost-hierarchy} records which repetitions are already
included in each cost. This prevents counting the reuse length or
the window subdivision again during assembly.
\begin{table}[H]
\centering\small
\caption{Cost hierarchy, using the common global bound $Z=Z_{\rm gl}$.
Constants depend only on the fixed local model.}
\label{tab:cost-hierarchy}
\begin{tabular}{p{.27\linewidth}p{.39\linewidth}p{.23\linewidth}}\toprule
Unit & Repetitions already included & Gate bound\\\midrule
One direct call & All address-tree levels and the selected gate
& $O_D(Z^7)$\\
One window & $30q$ direct or adjoint calls, coefficient preparation
and signal reflections & $O_D(Z^8)$\\
One local factor & At most $2r$ complete windows
& $O_D(Z^{10})$\\
One routed factor & Routing through a host of $O_D(Z^6)$ vertices
& $O_D(Z^{16})$\\
One segment & At most $n$ factors; only $O_D(1)$ colour layers for depth
& $O_D(nZ^{16})$\\
Whole finite-alphabet circuit & $N_0$ segments and one terminal
$O_D((1+\log X)^4)$ factor
& $O_D(N_0nZ^{16}(1+\log X)^4)$\\\bottomrule
\end{tabular}
\end{table}

\subsection{Proof of the main theorem}\label{sec:main-proof}
\begin{proof}[Proof of Theorem~\ref{thm:main}]
We choose $N_0,R$ by Eq.~\eqref{eq:globalparam} and factor budget
\[
 \eta=\frac{\eps}{32nN_0}.
\]
These fix $v_*,r$, then Eq.~\eqref{eq:mesh} fixes $M,h$, and the window
budget fixes $u,q,C_{\rm cap}$ and arithmetic widths. Common worst-case
$\Sigma_*$, budgets and address bounds fix all owners' widths before
construction; shorter active lists use subspaces of the shared encoding.
We generate circuits with exact one-qubit rotations and hosts, expand controlled gates and routing operations, then
assign terminal precision from the expanded length. Thus no terminal
approximation enters the choice of its own occurrence count.

First, Theorem~\ref{thm:patch} compares the exact forward segment with
the exact localized product on activation columns $U_{c-1}E_0$.
Second, Theorem~\ref{thm:prefix} supplies Theorem~\ref{thm:block} for
the localized prefixes $Y_j=W_{j-1}\cdots W_1E_0$. We use the fixed-factor
complete-block extensions of Section~\ref{sec:commonencoding}.
Here $I_P$ initializes every owner's signal and work arrays for the segment;
on a factor it is the local embedding tensored with outside zero registers.
The physical intertwining identity gives
$I_PY_j=W_{j-1}^{\rm c}\cdots W_1^{\rm c}I_PE_0$, and hence
\begin{align*}
 &[\widehat W_m^{\rm c}\cdots\widehat W_1^{\rm c}
       -W_m^{\rm c}\cdots W_1^{\rm c}]I_PE_0\\
 &\qquad=\sum_{j=1}^m\widehat W_m^{\rm c}\cdots
           \widehat W_{j+1}^{\rm c}
       [\widehat W_j^{\rm c}-W_j^{\rm c}]I_PY_j.
\end{align*}
Since $m\le n$, this column error is at most $n\eta$ per segment and
$nN_0\eta=\eps/32$ overall. Lemma~\ref{lem:samechannel} on the freshly
initialized segment channels gives local-simulation diamond error $\eps/16$.

For a complete segment $C_\tau:=g_{L_\tau}\cdots g_1$ after expansion of controlled gates and routing operations, retain exact alphabet gates and apply
Lemma~\ref{lem:terminal-synthesis} to each remaining one-qubit rotation
with $\xi_{\rm term}:=\eps/(16N_0\max\{1,L_\tau\})$.
Inverse occurrences use the literal adjoint of their saved sequence.
Full-space unitary telescoping, including all approximate work, gives
\begin{equation}\label{eq:sk-error}
 \norm{\widehat C_\tau-C_\tau}\le
 L_\tau\xi_{\rm term}\le\eps/(16N_0).
\end{equation}
Each replacement acts on the same target qubit and contributes the
polylogarithmic synthesis overhead counted below.
We discard auxiliaries at segment boundaries and repeat with fresh ones:
\begin{equation}\label{eq:actualchannel}
 \widehat\Psi_\tau(\rho):=\Tr_{\rm aux}
 \left[\widehat C_\tau(\rho\otimes\ket0\bra0_{\rm aux})
 \widehat C_\tau^\dagger\right],\qquad
 \widetilde\Phi_T:=\widehat\Psi_\tau^{\,N_0}.
\end{equation}
This channel is CPTP because it is obtained by unitary evolution with
initialized auxiliaries followed by a partial trace. Lemma~\ref{lem:samechannel} converts Eq.~\eqref{eq:sk-error}
and composes the $N_0$ segments, giving terminal diamond error $\eps/8$.
Together with Lemma~\ref{lem:collision} and Theorem~\ref{thm:patch},
\[
 \underbrace{\eps/16}_{\text{collision}}
 +\underbrace{\eps/8}_{\text{spatial localization}}
 +\underbrace{\eps/16}_{\text{local simulation}}
 +\underbrace{\eps/8}_{\text{finite-gate synthesis}}
 =3\eps/8<\eps.
\]

For resources, $N_0=O_D(1+a)$, $R=O_D(1+\log X)$,
$v_*,r,\Sigma_*=O_D((1+\log X)^D)$ and
$w_{\rm reg},\log(1/\eta)=O_D(1+\log X)$ imply
$Z_{\rm gl}=O_D((1+\log X)^D)$. Since
$L_\tau\le C_DnZ_{\rm gl}^{16}$,
\[
 \log(1/\xi_{\rm term})
 =\log\!\left(16N_0\max\{1,L_\tau\}/\eps\right)
 =O_D(1+\log X).
\]
We sum Eq.~\eqref{eq:globalblockcost} over at most $n$ factors per segment;
for depth, same-colour hosts run in parallel with $\chi=O_D(1)$ colours.
Logical gates, routing and terminal synthesis cost per block
\[
 O_D\!\left(Z_{\rm gl}^{10}Z_{\rm gl}^{6}(1+\log X)^4\right)
 =O_D\!\left((1+\log X)^{16D+4}\right).
\]
We multiply by $N_0$ and respectively $n$ or $\chi$ for $G$ and $d_{\rm NN}$.
Work is allocated once per owner: the complete-host union has at most
$C_DnZ_{\rm gl}^6$ sites, including routing spares, reused by all windows.
This proves peak width; initialization once per segment gives the fresh-qubit
bound after multiplication by $N_0$. Appendix~\ref{app:program} proves
the deterministic classical construction bound.
\end{proof}

\section{Geometric lower bounds}\label{sec:lower}
Our lower bounds use a dissipative jump that transports a particle
while preserving its internal qubit. Detecting its arrival beyond a
circuit's causal region gives a lower bound on depth. The entanglement
carried by the particle also bounds the number of gates that must act
across the intervening cuts. When gate locations follow a fixed
schedule, rare transfers give an additional gate bound that depends
on the required precision.

\subsection{Model and statements}
The system occupies $[L_{\rm side}]^D$, with $n=L_{\rm side}^D$
and two qubits per site. A simulator may use auxiliary registers at
any finite set of ambient sites $\Xi\subset\mathbb Z^D$ containing
this cube, with no bound on its size or on the finite number of
auxiliary registers per site. All gate-range, charging, initialization
and, for depth, classical-propagation conditions below apply to every
site of $\Xi$. The parameter $n$ counts only system sites.
Simulator gates have at most two quantum input
and output wires of dimension at most a fixed $q_g\ge2$, and couple
sites at $\ell_\infty$ distance at most a fixed integer $R_g\ge1$.
Every operation coupling different sites obeys these restrictions and
counts toward the gate budget. Single-site instruments and resets,
arbitrary classical outcome alphabets, and unbounded auxiliary density
at each site are allowed. The auxiliary state is initially separable
across all ambient sites and input-independent; outputs remain at their specified
sites. For depth, classical records produced during computation, including
measurement outcomes, travel at most $R_g$ sites per layer; this restriction
applies even when an outcome distribution is input-independent. Shared
input-independent randomness and a random choice of circuit may be supplied
in advance. Depth is maximized over execution branches. For gate
count, classical communication may be unrestricted. Each circuit has
a finite execution bound, to which we extend completed histories by
identity operations. A \emph{fixed schedule} prescribes one deterministic
list of gate locations for all branches, including random-seed and
measurement branches, and charges disabled slots as well.
An \emph{adaptive schedule} may change locations after outcomes; its
cost on a specified input is the expected number of executed gates,
averaged over all outcomes and input-independent random seeds.

The upper-bound architecture belongs to these classes. Take $\Xi$
to contain the coarse images $\lfloor z/f\rfloor$ of all fine-grid
host vertices. This map preserves system labels and sends a fine
nearest-neighbour edge to a same-site or neighbouring-site edge.
The initial auxiliary state remains a product across ambient sites.
The lower-bound
classes also allow arbitrary continuous elementary gates.

We choose a positive dyadic rational jump amplitude $b$ and write
$\gamma:=b^2$ for its rate and $\theta:=\gamma T$ for the normalized
physical time. Let
$0<\delta<1$ denote \emph{half}-diamond error:
$\tfrac12\|\mathcal C-e^{T\mathcal L}\|_\diamond\le\delta$.
To compare with Theorem~\ref{thm:main}, take $\delta=\eps/2$.
For $E:=\log(n/\delta)$ define
\[
 \mathcal R_1(\theta,E):=\theta+\frac{E}{\log(e+E/\theta)},\qquad
 \mathcal R_{D+1}(\theta,E):=\theta+
 \left(\frac{E}{\log(e+E/\theta^{D+1})}\right)^{1/(D+1)}.
\]
\Needspace{10\baselineskip}
\begin{theorem}[Static dissipative hard instances]\label{thm:lower}
For fixed $D,q_g,R_g$, sufficiently large $L_{\rm side}$, $\theta\ge1$,
and $0<\delta\le\delta_0(D,q_g,R_g)$, worst-case simulation of the class
of static rational generators with $H=0$ constructed below satisfies
\begin{align}
 d_{\rm wc}&=\Omega_D\!\left(\min\{L_{\rm side},\mathcal R_1(\theta,E)\}\right),
 \label{eq:lower-depth}\\
 G^{\rm adaptive}_{\rm wc}&=\Omega_D\!\left(n\min\{L_{\rm side},\theta\}\right),
 \label{eq:lower-adaptive}\\
 G^{\rm fixed}_{\rm wc}&=\Omega_D\!\left(n\min\{L_{\rm side},\mathcal R_{D+1}(\theta,E)\}\right).
 \label{eq:lower-fixed}
\end{align}
For each inequality there is a generator in the stated class such that every simulator of its evolution in the corresponding circuit model obeys that bound. The generator may be chosen separately for the three inequalities. In
Eq.~\eqref{eq:lower-adaptive}, the lower bound holds for the expected
gate count on one input common to every cut used in the proof. The hard instances have
traceless jumps, constant incidence and $\lambda=\Theta(\gamma)$ under
the normalization of Section~\ref{sec:model}.
\end{theorem}
Implicit constants may depend on $D,q_g,R_g$. Before the diameter
cap, the precision bounds add their time and accuracy contributions,
and the leading volume and time factors agree with the upper bounds.
The fixed-schedule precision term charges deterministic locations;
the adaptive bound permits outcome-dependent locations. Related
lower bounds concern time-dependent Hamiltonians
\cite[Theorem~2]{HHKL}, sparse Lindbladian access \cite{ChildsLi},
and engineered Hamiltonian state transfer \cite{ChristandlTransfer}.
Our instances use static rational dissipative matrices.

The two precision radii reflect different transfer events. For depth,
one particle travels distance $r$, and the negative logarithm of the
arrival-probability bound scales as $r\log(e+r/\theta)$ in the
rare-arrival regime.
For a fixed schedule, a source cube of side $k$ holds $k^D$ internal
qubits and is surrounded by $k$ cuts. A sufficient success event
requires $O(k^{D+1})$ hop waits to finish within $O(T/k)$ each.
Its probability bound has negative logarithm
$O(k^{D+1}\log(e+k/\theta))$ in the rare-arrival regime.
Combining independent cubes until the aggregate success probability
exceeds the allowed error gives the root of order $D+1$ in
Eq.~\eqref{eq:lower-fixed}. The proofs below give the precise probability
comparisons and gate counts.

\subsection{Coherent transport and the depth bound}
Each site has two qubits. We use the fixed computational-basis encoding
$\ket v:=\ket{00}_{\rm comp}$ for a vacancy,
$\ket0:=\ket{01}_{\rm comp}$ and $\ket1:=\ket{10}_{\rm comp}$
for the two internal states of a particle, and leave
$\ket{11}_{\rm comp}$ inert.
Here a particle's $\ket0$ is distinct from the vacuum label of a
collision bath. For an oriented bond, the jump operator and generator are
\begin{equation}\label{eq:transportjump}
 S_x:=\ket{v,0}\bra{0,v}+\ket{v,1}\bra{1,v},\qquad
 J_x:=bS_x,\qquad \mathcal L_{\rm tr}:=\sum_x\mathcal D_{J_x}.
\end{equation}
One coherent jump per bond combines the two internal-state transitions.
For example,
\[
 S_x\bigl[(a\ket0+c\ket1)\otimes\ket v\bigr]
 =\ket v\otimes(a\ket0+c\ket1).
\]
The jump time and occupied position are random, whereas the internal
amplitudes are transported together. The same identity holds when the
internal state is entangled with a reference. The matrix entries are rational, $\|S_x\|=1$ and $\Tr S_x=0$.
Since $\|S_x\|_F=\sqrt2$, the least power-of-two scale obeys
$8\gamma\le\beta_x^2<32\gamma$. Each site below belongs to at most
two bonds, so $8\gamma\le\lambda<64\gamma$ for every nontrivial
instance, establishing $\lambda=\Theta(\gamma)$.

For an occupation pattern $\eta$, let $V_\eta$ place the internal
qubits in occupied sites from left to right within each isolated chain,
with a fixed lexicographic ordering of the chains. A legal hop satisfies
$S_xV_\eta=V_{\eta'}$ and $S_x^\dagger S_xV_\eta=V_\eta$;
both are zero for a disabled hop. Thus an enabled bond satisfies
\[
 \mathcal D_{bS_x}(V_\eta\rho V_\eta^\dagger)
 =\gamma\bigl(V_{\eta'}\rho V_{\eta'}^\dagger
                 -V_\eta\rho V_\eta^\dagger\bigr),
\]
and a disabled bond contributes zero. These identities hold with the
identity on any reference. The GKSL equation therefore reduces to the
classical forward equation on occupation patterns, giving
\begin{equation}\label{eq:transportintertwine}
 e^{T\mathcal L_{\rm tr}}(V_{\eta_0}\rho V_{\eta_0}^\dagger)
 =\sum_\eta p_T(\eta\mid\eta_0)V_\eta\rho V_\eta^\dagger.
\end{equation}
Here $p_T$ is the classical directed exclusion process with rate $\gamma$.
Particles retain their order, and occupation measurements read their positions while preserving the transported internal state. The exclusion and
last-passage correspondence is standard \cite[Section~1]{Johansson};
Eq.~\eqref{eq:transportintertwine} supplies the quantum-state part needed here.

\begin{lemma}[Causal channel test]\label{lem:lower-causal}
Let $X,Y$ be site sets at distance $r$. Two input states differing only
on $X$ have identical reduced $Y$ outputs under a depth-$d$ simulator
whenever $R_gd<r$. If a target channel distinguishes such inputs on $Y$
with trace distance $p>2\delta$, every simulator with error at most $\delta$ in half the diamond norm
therefore has $d\ge r/R_g$.
\end{lemma}
\begin{proof}
We propagate a joint observable on $Y$ backwards through the layers
in the ambient lattice $\Xi$;
its support grows by at most $R_g$ per layer. The same argument includes
local classical registers, preparations and discards. Averaging over an
input-independent choice of circuit preserves independence from $X$.
The triangle inequality assigns at most $2\delta$ to the two approximation errors, so a target distinction larger than this forces the stated depth. This circuit causality principle also underlies dissipative preparation lower bounds \cite{KonigPastawski}; the joint receiver observable lets us apply it to many simultaneous transfers.
\end{proof}

On a rod with one particle and an absorbing endpoint, the $r$ hops
have independent rate-$\gamma$ exponential waits. The arrival
probability is therefore $p_r(\theta):=\Pr\{\operatorname{Pois}(\theta)
\ge r\}$. We amplify rare arrivals by testing whether any of many
independent receivers is occupied. For $1\le r\le L_{\rm side}/4$,
we place isolated $r$-bond rods at spacing $2r$ on every transverse
row: rod $j$ has source $1+2jr$ and receiver $1+(2j+1)r$, where
$0\le j<\lfloor L_{\rm side}/(2r)\rfloor$. Thus
$m=L_{\rm side}^{D-1}\lfloor L_{\rm side}/(2r)\rfloor\ge n/(4r)$,
and every source is at distance at least $r$ from every receiver.
Comparing the all-vacuum input with one particle at each source, the
joint receiver test and Lemma~\ref{lem:lower-causal} force
$d\ge r/R_g$ whenever
\begin{equation}\label{eq:poissoncriterion}
 1-(1-p_r(\theta))^m>2\delta.
\end{equation}
Each channel approximation changes this single probability by at
most $\delta$, independently of the number of rods.

To obtain Eq.~\eqref{eq:lower-depth}, we choose a fixed
$c\le1/128$ and set
$r=\lfloor c\min\{L_{\rm side},\mathcal R_1(\theta,E)\}\rfloor$.
We distinguish a typical-arrival regime, where a constant success
probability is enough, from a rare-arrival regime, where the number of
rods compensates for the small probability. If $1\le r\le\theta/2$, the second-moment
inequality for a Poisson random variable gives $p_r\ge1/8$, so the
joint test succeeds for a sufficiently small fixed $\delta_0$.
Indeed, for $Z\sim\operatorname{Pois}(\theta)$,
$\Pr\{Z\ge\theta/2\}\ge\theta^2/[4(\theta^2+\theta)]\ge1/8$
when $\theta\ge1$.
Otherwise, we write $\ell=\log(e+E/\theta)$. Since
$r\le c(\theta+E/\ell)$ and $\theta<2r$, we obtain
\[
 r\le \frac{c}{1-2c}\frac{E}{\ell}
 \le\frac{2cE}{\ell},\qquad
 \log(r/\theta)\le\ell.
\]
The single Poisson mass at $r$, together with $r!\le r^r$, now gives
\[
 -\log p_r\le\theta+r\log(r/\theta)
 \le 2r+r\ell\le6cE\le E/4.
\]
It follows that
\[
 \frac{mp_r}{\delta}\ge\frac{e^{3E/4}}{4r}
                  \ge\frac{e^{3E/4}}{8cE}.
\]
This expression increases for $E\ge4/3$. Since
$E\ge\log(1/\delta_0)$, a sufficiently small fixed
$\delta_0\le1/4$ makes it exceed four uniformly;
$(1-p_r)^m\le e^{-mp_r}$ then proves
Eq.~\eqref{eq:poissoncriterion}. If rounding gives $r=0$, the proposed
bound is constant, and a one-bond instance with
$p_1\ge1-e^{-1}$ suffices. Choosing $r$ from the minimum with
$L_{\rm side}$ already includes the diameter cap.

\subsection{Quantum transmission and executed gate count}
For adaptive circuits we use one input with an entangled pair between
each source qubit and an external inert reference assigned to the same
site. The simulator never acts on these references.
The input is separable across spatial cuts, but a successful transfer
creates entanglement across every cut it traverses. We use conditional entropy, measured in bits,
$S(A|B)_\rho:=S(AB)_\rho-S(B)_\rho$, and entanglement of formation
$E_F$, the convex roof of pure-state entanglement entropy. Its average
LOCC monotonicity \cite[Theorem~2]{VidalMonotones} combines with
conditional-entropy continuity \cite[Lemma~2]{WinterContinuity},
whose dimension dependence involves only the decoded register $A$.

\begin{lemma}[Entanglement estimates]\label{lem:enttools}
$E_F$ is convex and nonincreasing on average under local operations and
classical communication. If states $\rho,\sigma$ on $AB$ have half
trace-norm distance at most $\delta\in[0,1]$, then
\[
 |S(A|B)_\rho-S(A|B)_\sigma|
 \le2\delta\log_2\dim A+h_*(\delta),\qquad
 h_*(\delta):=(1+\delta)h_2\!\left(\frac\delta{1+\delta}\right),
\]
where $h_2$ is binary entropy. Also $E_F(\rho)\ge-S(A|B)_\rho$.
\end{lemma}
\begin{proof}
Convexity follows from the convex-roof definition; average LOCC
monotonicity and the continuity inequality are the cited results.
For any pure-state ensemble of $\rho$, concavity of conditional entropy
\cite{LiebRuskaiSSA} bounds $-S(A|B)_\rho$ by that ensemble's average
entanglement. Taking the infimum proves the last inequality.
\end{proof}
A two-wire unitary increases entanglement by at most $2\log_2q_g$
\cite[Lemma~1]{BHLSCapacity}. The same bound applies on average to a
crossing instrument in our model. For a pure state in an input ensemble,
we place its isometric dilation environment on one side of the cut
and denote the other side by $BB_0$, where $B$ is the affected wire.
The untouched marginal on $B_0$ is unchanged, while before and after
the dilation $|S(BB_0)-S(B_0)|\le\log_2q_g$. The pure-state
entanglement therefore increases by at most $2\log_2q_g$.
Taking the convex-roof infimum over input ensembles and then using
average LOCC monotonicity for environment measurement, outcome
communication and discard proves the claim for the two bounded
output wires of the charged operation.

For every existing outcome history $h$, let $E_F^{\rm before}(h)$ be
the entanglement before the next operation, and let
$\mathbf1_{\rm cross}(h)$ indicate whether that operation crosses the
cut. The preceding increment bound gives
\[
 \mathbb E[E_F^{\rm after}\mid h]
 \le E_F^{\rm before}(h)+2\log_2q_g\,\mathbf1_{\rm cross}(h).
\]
Taking expectations and summing over the executed operations starts
from zero entanglement of formation and charges only the expected
number of crossing gates. Convexity bounds the entanglement of
the unconditional output by the resulting average. Thus, for one
sitewise separable system--reference input, any decoder acting locally on each side of the cut
with ideal output $\sigma^{AB}$ gives
\begin{equation}\label{eq:entcut}
 2\log_2q_g\,\mathbb E g_{\rm cut}
 \ge -S(A|B)_\sigma-2\delta\log_2\dim A-h_*(\delta).
\end{equation}
The decoder is applied after the simulator and acts by local operations
and classical communication across the cut, so it preserves the bound
in Eq.~\eqref{eq:entcut}.

We realize a transfer of many qubits on an isolated $3k$-site rod
with $k$ particles initially at its left end; success places all of
them in the final $k$ sites. If $T_{i,j}$ is the time at which particle
$i$, indexed from right to left, completes hop $j$, then
\[
 T_{i,j}=\max\{T_{i-1,j},T_{i,j-1}\}+E_{i,j},\quad
 1\le i\le k,\quad1\le j\le2k,
\]
with zero boundary times and independent rate-$\gamma$ exponential
waits. Initially adjacent particles differ in position by one.
Particle $i$'s target for hop $j$ is the site occupied by particle
$i-1$ after its hop $j-1$; the latter must complete hop $j$ to free it.
A hop therefore needs its predecessor and right neighbour to finish;
once enabled, it remains enabled, so memorylessness gives this
recursion. Completion at $T_{k,2k}$ is the maximum over at most
$2^{3k-1}$ dependence paths of $3k-1$ waits. The exponential-moment
bound at $\gamma/2$ and a union bound give the first estimate below;
requiring all $2k^2$ waits to be at most $T/(3k-1)$ gives the second:
\begin{equation}\label{eq:packedprob}
 p_k(\theta)\ge1-64^k e^{-\theta/2},\qquad
 p_k(\theta)\ge\left(1-e^{-\theta/(3k-1)}\right)^{2k^2}.
\end{equation}

We take $k=\max\{1,\lfloor\min(\theta,L_{\rm side})/64\rfloor\}$
and place the rods every $4k$ longitudinal sites on all transverse
rows, covering a constant fraction of a sufficiently large cube.
For $k\ge2$, the first estimate in Eq.~\eqref{eq:packedprob} exceeds
$0.99$; for $k=1$, the exact two-hop probability is at least
$1-2/e>1/4$. We entangle each source qubit with an inert reference at
its original site. The same input is used at every cut, so its expected crossing counts
can later be summed. For a rod whose leftmost longitudinal coordinate
is $a$, we use the planes immediately after coordinates
$a+k-1,\ldots,a+2k-2$. Each selected plane at coordinate $s$ denotes
the full ambient bipartition
$\Xi\cap\{z:z_1\le s\}$ and its complement, including every auxiliary
register on the corresponding side. Every source reference is on their left,
and every receiver is on their right. At each plane, the receiver side
tests whether all of the last $k$ sites are occupied. This is a local success test and
leaves the transported internal qubits intact. It communicates the outcome
to the reference side, so both sides can retain a copy of the flag.
On success, each side decodes its $k$ qubits; on failure, both replace
their decoded registers by a fixed product state. Thus the analysis
decoder is an LOCC operation across each cut, as required in
Eq.~\eqref{eq:entcut}.
Writing $\Phi_{2^k}$ for $k$ normalized Bell pairs, its decoded target is
\[
 \sigma_{\rm row}
 =p_k\Phi_{2^k}\otimes\ket{11}\bra{11}
 +(1-p_k)\ket{0^{2k}}\bra{0^{2k}}\otimes\ket{00}\bra{00}.
\]
Here $A$ consists of the source references and their flag, while $B$
consists of the decoded receiver qubits and their flag. Conditioning
on the receiver flag gives $S(A|B)_{\sigma_{\rm row}}=-p_k k$, while
$\log_2\dim A=k+1$. At a selected plane we retain one rod from
its longitudinal block on each of the
$N_{\rm row}=L_{\rm side}^{D-1}$ transverse rows. The product of
their flagged target states gives
\[
 -S(A|B)=N_{\rm row}p_k k,\qquad
 \log_2\dim A=N_{\rm row}(k+1).
\]
These quantities add over the product target. Since $p_k$ is
bounded below by a constant, Eq.~\eqref{eq:entcut} requires
$\Omega_D(L_{\rm side}^{D-1}k)$ expected crossing gates per plane
for sufficiently small fixed $\delta_0$. There are
$\Theta(L_{\rm side})$ selected planes, and a range-$R_g$ gate
crosses at most $R_g$ of them. All expectations use the same input,
so
\[
 \sum_{\rm selected\ planes}\mathbb E g_{\rm cut}
 \le R_g\,\mathbb E G,
\]
proving Eq.~\eqref{eq:lower-adaptive}.

\subsection{The extra precision bound for a fixed schedule}
A fixed schedule lets us identify cuts with insufficient crossing
slots before observing any outcomes. Their Schmidt-number bound
survives conditioning on a rare success. By selecting the first
successful transfer among predetermined deficient blocks, we collect
these events into one unconditional test with error at most $\delta$.

The Schmidt number of a state is the least $r$ for which it is a mixture
of pure states of Schmidt rank at most $r$. For a nonzero positive
operator we use the Schmidt number of its normalization; the zero
operator satisfies every upper bound. We use the following properties
from \citet[Proposition~1 and Lemma~1]{TerhalSchmidt}: local selection
cannot increase Schmidt number, and bounded Schmidt number limits
acceptance of a maximally entangled state.
\begin{lemma}[Schmidt-number estimates]\label{lem:schmidttools}
Local operations and conditioning on a local outcome cannot increase
Schmidt number. Mixing states of Schmidt number at most $r$ preserves
that bound. A positive operator of Schmidt number at most $r$ has overlap with a $d$-dimensional maximally entangled state
at most $(r/d)$ times its trace.
\end{lemma}
For a spatial region $A$, let $g_A$ count prescribed slots crossing
its boundary. Each crossing Kraus operator has operator Schmidt rank
at most $q_g^2$, so the separable initial state has output Schmidt
number at most $q_g^{2g_A}$. Refining histories into individual
Kraus operators multiplies Schmidt rank by at most $q_g^2$ per
crossing slot; local steps, selection and final mixing preserve
this common bound. Here $A$ denotes a spatial region, distinct from
the decoded register in Eq.~\eqref{eq:entcut}.

We consider $m$ disjoint source blocks, each holding $s$ onsite Bell pairs
with inert references, and a family of $K$ enclosing cuts per block.
All receiver registers and success flags lie outside \emph{every}
source cut. The analysis decoder measures the occupation projection
$\ket0\bra0+\ket1\bra1$ at each receiver site. On an occupied
outcome it maps these two states to a logical qubit; the complementary
outcome, including the unused fourth basis state, is recorded as
failure and replaced by a fixed state. Hence this decoder defines a
quantum channel on every simulator output.

The following statement isolates the role of a fixed schedule before
we choose the geometric blocks.
\begin{lemma}[Gate count from independent transfer flags]\label{lem:flag-gates}
In the preceding source--receiver setting, suppose the $m$ ideal
transfer flags are independent with success probability $p$, and
each success transports exactly $s\ge1$ qubits. If each elementary
gate crosses at most $\chi_g$ of the $mK$ cuts, then every fixed-schedule
simulator with half-diamond error at most $\delta$ satisfies
\[
 1-(1-p)^{\lceil m/2\rceil}>6\delta
 \quad\Longrightarrow\quad
 G\ge\frac{mKs}{8\chi_g\log_2q_g}.
\]
\end{lemma}
\begin{proof}
We call a block deficient if one of its cuts has
$g_A<s/(4\log_2q_g)$, and fix one such cut for each deficient block.
The schedule determines this set before any flags are observed.

The output test measures the receiver flags and checks the Bell state
at the \emph{first successful deficient block}. For the predetermined
set $B$ of deficient blocks, this single two-outcome test has ideal
acceptance probability $q_B:=1-(1-p)^{|B|}$.
Let $\sigma_i$ be the unnormalized simulator output after selecting
the event that deficient block $i$ is the first to succeed and
locally discarding all registers except its $s$ references and $s$
receivers. Let $\Pi_i$ be the normalized Bell-state projector on
these retained registers. The trace of $\sigma_i$ is the probability of this selection event. All flags used for the selection lie outside the cut chosen
for block $i$, so the selected operator
retains the Schmidt-number bound for that cut. With
$\zeta:=2^{-s/2}\le1/\sqrt2$, Lemma~\ref{lem:schmidttools} gives
\[
 \Tr(\Pi_i\sigma_i)\le\zeta\Tr\sigma_i,\qquad
 \Pr_{\mathcal C}(\mathrm{accept})
 =\sum_{i\in B}\Tr(\Pi_i\sigma_i)
 \le\zeta\sum_{i\in B}\Tr\sigma_i.
\]
The selected events are disjoint and sum to the event that some
deficient block succeeds. Comparing this event and the acceptance
test with their ideal probabilities uses the global half-diamond
error once for each of two unconditional tests. Consequently,

\begin{equation}\label{eq:flagdirectsum}
 q_B-\delta\le\Pr_{\mathcal C}(\mathrm{accept})
 \le\zeta\Pr_{\mathcal C}(\text{a deficient block succeeds})
 \le\zeta(q_B+\delta).
\end{equation}
Since $\zeta\le1/\sqrt2$, Eq.~\eqref{eq:flagdirectsum} implies
$q_B\le[(1+\zeta)/(1-\zeta)]\delta\le(3+2\sqrt2)\delta<6\delta$.
If $1-(1-p)^{\lceil m/2\rceil}>6\delta$, at least half the blocks
have $s/(4\log_2q_g)$ crossing slots at every one of their $K$ cuts.
Summing with cut congestion at most $\chi_g$ gives
$G\ge mKs/(8\chi_g\log_2q_g)$.
\end{proof}

To realize this test geometrically, we tile the lattice by cells of
side $4k$. Each cell contains $k^{D-1}$ parallel rods of length $3k$,
with $k$ particles on each rod. Thus its source cube holds $s=k^D$
qubits, and the receiver cube is translated by $2k$ in the first
coordinate. Requiring all rods in the cell to complete their transfer
gives success probability at least
\[
 p\ge\left(1-e^{-\theta/(3k-1)}\right)^{2k^{D+1}}.
\]
These occupation flags are independent in the ideal output because
Eq.~\eqref{eq:transportintertwine} applies separately to the disjoint
cells and the initial state is their tensor product.
For a source cube $A_C$, we take the $K=k$ cuts bounding
\[
 A_C^{(j)}:=\{x\in\mathbb Z^D:
       \operatorname{dist}_\infty(x,A_C)\le j\},
 \qquad j=0,\ldots,k-1.
\]
Each cut separates $\Xi\cap A_C^{(j)}$ from
$\Xi\setminus A_C^{(j)}$, including every auxiliary register. The outer neighbourhoods
are disjoint because source cubes are spaced
by $4k$, and every receiver lies outside all of them. The function
$\operatorname{dist}_\infty(\cdot,A_C)$ is 1-Lipschitz, so two gate
endpoints at distance at most $R_g$ can lie on opposite sides of at
most $R_g$ cuts in one family. Since an endpoint belongs to at most
one outer neighbourhood, a gate meets at most two families. Its total
cut congestion is therefore at most $2R_g$, independently of transverse
routing or the number of auxiliary wires at a site.
There are $m=\lfloor L_{\rm side}/(4k)\rfloor^D=\Theta_D(n/k^D)$
cells for $k\le L_{\rm side}/8$. The preceding test then gives
$G=\Omega_D(nk)$ whenever its success criterion holds.

It remains to choose $k$ so that the aggregate success probability
exceeds the error. We write
$\ell=\log(e+E/\theta^{D+1})$ and
$Q=(E/\ell)^{1/(D+1)}$. If $Q\le4\theta$, the adaptive time bound
already implies Eq.~\eqref{eq:lower-fixed}, up to a factor of five.
Otherwise we choose
$k=\lfloor c_D\min\{L_{\rm side},Q\}\rfloor$, with
$c_D\le1/8$ a sufficiently small fixed constant. The case
$k\le\theta$ is also covered by the time bound: the floor definition
gives $\min\{L_{\rm side},Q\}<(k+1)/c_D\le2\theta/c_D$, and hence
\[
 \min\{L_{\rm side},\theta+Q\}
 \le(1+2/c_D)\min\{L_{\rm side},\theta\}.
\]
For $k>\theta$,
the elementary inequality $1-e^{-x}\ge x/2$ for $0\le x\le1$ yields
\[
 -\log p\le 2k^{D+1}\log(6k/\theta).
\]
Moreover, $k^{D+1}\le c_D^{D+1}E/\ell$ and
\[
 \log(6k/\theta)
 \le \log6+\frac{\ell}{D+1}
 \le\left(\log6+\frac1{D+1}\right)\ell.
\]
We may therefore fix $c_D$ so that $-\log p\le E/4$.
Since $m=\Theta_D(n/k^D)$ and $k=O_D(E^{1/(D+1)})$, it follows that
\[
 \frac{mp}{\delta}\ge
 c'_D\frac{e^{3E/4}}{E^{D/(D+1)}}>24
\]
for a sufficiently small fixed $\delta_0$. Consequently,
$1-(1-p)^{\lceil m/2\rceil}\ge1-e^{-mp/2}>6\delta$, as required by
Eq.~\eqref{eq:flagdirectsum}. The earlier $k\le\theta$ branch already
covers the bounded rounding cases. This proves the precision
contribution. We choose the time-bound instance or the precision-bound
instance according to which gives the larger contribution. The
resulting single instance proves Eq.~\eqref{eq:lower-fixed}, since
\[
 \max\{\min(L_{\rm side},\theta),\min(L_{\rm side},Q)\}
 \ge\tfrac12\min(L_{\rm side},\theta+Q).
\]

This completes the proof of Theorem~\ref{thm:lower}.
For short times outside the theorem's regime, CPTP contractivity
instead gives
$\tfrac12\|e^{T\mathcal L}-\operatorname{id}\|_\diamond
\le\tfrac12T\|\mathcal L\|_\diamond$ by integrating the semigroup
derivative; the identity suffices whenever this is at most $\delta$.

\Needspace{8\baselineskip}
\section{Discussion}
The spatial decomposition separates the geometric argument from the
local simulation method. A replacement local solver must implement the
prescribed shared-environment unitaries on the physical prefixes and
inverse preimages produced by that decomposition. Improving such a
solver can reduce the logarithmic overhead without changing the
spatial estimate or the global composition argument.

For state preparation, a mixing estimate supplies a sufficient
evolution time \cite{TemmeMixing,KastoryanoTemme}; the simulation error
then adds to the physical preparation error as described in
Section~\ref{sec:intro-result}. Stability bounds can separately account
for perturbations of the generator \cite{CubittStability}. For dynamics
with persistent motion \cite{Buca}, the guarantee instead concerns the
state at the specified time, rather than convergence to a stationary
state.

The lower bounds show how the combined probability of rare arrivals
imposes precision costs before diameter saturation. Extending the
upper bound to time-dependent generators or environments with memory
would require a common dilation with a mesh-uniform locality estimate,
occupation bounds for the resulting correlated prefixes and inverse
preimages, and efficient coherent evaluation of the required interval
operators. The width bound relies on refreshing the auxiliaries between segments
and on the allowed, explicitly counted lattice refinement.

\section*{Note added}
Our results were obtained independently of the concurrent works of
Mizuta~\cite{MizutaPatching} and Trivedi, Yu and Malz~\cite{TrivediYuMalz}.
Mizuta develops patching and merging methods that give nearly linear
gate complexity for channel simulation with sparsely separated
dissipative domains, and improved observable-estimation algorithms for
general one-dimensional local generators. Trivedi, Yu and Malz obtain the same leading
gate-count and depth scaling for lattice Lindbladian channel simulation,
using a unitary dilation and a compressed representation of environmental
excitations.

\section*{AI Usage}
The core ideas and overall approach of this work were developed by the authors. The proofs and manuscript were developed through multiple rounds of interaction between the authors and GPT-5.6 and GPT-6, which assisted with calculations, proof development, and writing. The authors take full responsibility for the correctness of the proofs and conclusions presented in this paper.

\section*{Acknowledgments}
The authors thank Xinzhao Wang for helpful discussion. This work is supported by the National Natural Science Foundation of China Grant No.~12361161602 and NSAF Grant No.~U2330201, the Quantum Science and Technology--National Science and Technology Major Project No.~2023ZD0300200, Beijing Natural Science Foundation Z250004, and Beijing Science and Technology Planning Project Grant No.~Z25110100810000. Y.~Wu acknowledges support from NSFC Grants No.~62501060, No.~62461160263, and No.~62371050.

\appendix

\section{Error of the collision representation}\label{app:collision}
We compare each rational collision channel with its local semigroup, then
the ordered semigroup product with $e^{T\mathcal L}$
\cite{Trotter,LindbladCommutator}. Positive-time CPTP contractivity adds
errors over periods. This comparison uses vacuum jump columns; spatial
localization uses the full occupied-bath unitary of Section~\ref{sec:model}.

\begin{lemma}[Ordered jump-channel identity]\label{lem:orderedjump}
For a stacked jump term $J$, put
\[
 A_J=J^\dagger J,\qquad R_h=(I+hA_J/4)^{-1},\qquad
 \mathcal S_{R_h}(Z)=R_hZR_h,\qquad a_J=\norm{A_J}.
\]
For every $h\ge0$, its reduced Cayley channel satisfies
\begin{equation}\label{eq:orderedjump}
 \Phi_{J,h}=\operatorname{id}+h\mathcal D_J\mathcal S_{R_h},
\end{equation}
with the sandwich \emph{before} the dissipator, and
\begin{align}
 \norm{\Phi_{J,h}-\operatorname{id}-h\mathcal D_J}_\diamond
       &\le h^2a_J^2,\label{eq:jumpeuler}\\
 \norm{e^{h\mathcal D_J}-\operatorname{id}-h\mathcal D_J}_\diamond
       &\le2h^2a_J^2,\label{eq:jumpexprem}\\
 \norm{\Phi_{J,h}-e^{h\mathcal D_J}}_\diamond
       &\le3h^2a_J^2.\label{eq:jumpfullerror}
\end{align}
\end{lemma}
\begin{proof}
The vacuum Kraus operators are $K_0=(I-hA_J/4)R_h$ and
$K_\mu=-i\sqrt h\,L_\mu R_h$, with Gram sum
$R_h[(I-hA_J/4)^2+hA_J]R_h=I$. For $Y=R_hZR_h$ and $c=h/4$,
\[
 Z=(I+cA_J)Y(I+cA_J),\qquad
 K_0ZK_0^\dagger=(I-cA_J)Y(I-cA_J).
\]
Their difference cancels the quadratic terms; adding jump terms gives
$\Phi_{J,h}(Z)-Z=h\mathcal D_J(Y)$ for all complex matrices and reference
extensions, fixing the superoperator order. Since $\norm{R_h}\le1$,
$\norm{I-R_h}\le ha_J/4$ and
$R_hZR_h-Z=(R_h-I)ZR_h+Z(R_h-I)$,
\[
 \norm{\mathcal S_{R_h}-\operatorname{id}}_\diamond\le ha_J/2,
 \qquad \norm{\mathcal D_J}_\diamond\le2a_J.
\]
The latter follows from $\Tr_{\rm label}(JZJ^\dagger)$ and the two
anticommutator terms \cite[Chapter~3]{Watrous}, proving
Eq.~\eqref{eq:jumpeuler}. Also
\[
 e^{h\mathcal D_J}-\operatorname{id}-h\mathcal D_J
 =\int_0^h(h-t)e^{t\mathcal D_J}\mathcal D_J^2\,dt.
\]
The positive-time GKSL semigroup has diamond norm one
\cite{GKS,Lindblad}, so this is at most $h^2\norm{\mathcal D_J}_\diamond^2/2
\le2h^2a_J^2$. The triangle inequality proves the last bound.
This includes $h=0$ and an empty jump family. The Euler expression is an
analytical coefficient; the implemented step is the exact Kraus channel.
\end{proof}

The Hamiltonian estimate is the scalar Cayley approximation used in
\citet[Lemma~3]{ChenDyadic}; we include its short proof.
\begin{lemma}[Hamiltonian Cayley approximation]\label{lem:hcayleydefect}
For $H=H^\dagger$, $h\ge0$ and $h\norm H\le1$, define
$C_H(h):=(I-ihH/2)(I+ihH/2)^{-1}$. Then
\begin{equation}\label{eq:hcayleyerror}
 \norm{C_H(h)-e^{-ihH}}\le\frac{h^3\norm H^3}{12},\qquad
 \norm{\operatorname{Ad}_{C_H(h)}-e^{h\mathcal H_H}}_\diamond
 \le\frac{h^3\norm H^3}{6}.
\end{equation}
\end{lemma}
\begin{proof}
At a real eigenvalue $x$ of $hH$, the Cayley phase is
$e^{-2i\arctan(x/2)}$ and, for $x\ge0$,
\[
 0\le x-2\arctan(x/2)
   =\int_0^x\frac{t^2}{4+t^2}\,dt\le\frac{x^3}{12}.
\]
Oddness covers $x<0$. Spectral calculus and the phase-distance bound
give the operator estimate; Lemma~\ref{lem:samechannel} doubles it
for $e^{h\mathcal H_H}=\operatorname{Ad}_{e^{-ihH}}$.
\end{proof}

We use the positive-time first-order product bound of
\citet[Theorem~6, version~2]{LindbladCommutator}.
\begin{lemma}[Positive-time product formula]\label{lem:positivesplitting}
For finite-dimensional GKSL generators $\mathcal L_1,\ldots,\mathcal L_m$
and $h\ge0$,
\begin{equation}\label{eq:positivesplitting}
 \left\|e^{h\mathcal L_m}\cdots e^{h\mathcal L_1}
          -e^{h\sum_e\mathcal L_e}\right\|_\diamond
 \le\frac{h^2}{2}\sum_{e<f}\norm{[\mathcal L_e,\mathcal L_f]}_\diamond.
\end{equation}
\end{lemma}
The source uses this product order and commutators
$[\mathcal L_e,\sum_{f>e}\mathcal L_f]$. We expand the tail sums and use
the triangle inequality; each tail is GKSL and hence positive-time contractive.

\begin{lemma}[Physical channel accuracy]\label{lem:collision}
Let $m_{\rm tot}\ge0$ be an integer, $h\ge0$ and $T=m_{\rm tot}h$.
The forward collision circuit, in the prescribed term order with an
independent fresh bath for each jump label, induces a CPTP channel.
If $h\lambda\le1$, its diamond distance from $e^{T\mathcal L}$ is at most
\begin{equation}\label{eq:physicalsharp}
 \frac12 n\lambda^2Th
 \le C_{\rm col}n\lambda^2Th,\qquad C_{\rm col}=8(\kappa+1).
\end{equation}
This estimate covers only the forward fresh-label circuit.
\end{lemma}
We retain the conservative $C_{\rm col}$ in Eq.~\eqref{eq:mesh}.
\begin{proof}
Normalization gives $\norm H\le\lambda_e/2$,
$\norm{J^\dagger J}\le\lambda_e/4$ and
$\norm{\mathcal L_e}_\diamond\le\lambda_e$.
The preceding lemmas give H/J step errors
$h^2\lambda_e^2/48$ and $3h^2\lambda_e^2/16$, using
$h\lambda_e\le1$ for H. Disjoint supports commute, so local substitutions
and Eq.~\eqref{eq:positivesplitting} bound one period by
\[
 h^2\left(\frac3{16}\sum_e\lambda_e^2+
      \sum_{e<f:\,X_e\cap X_f\ne\varnothing}\lambda_e\lambda_f\right).
\]
Every retained support is nonempty, and
\[
 \sum_{x\in\Lambda}\left(\sum_{e:\,x\in X_e}\lambda_e\right)^2
 =\sum_e|X_e|\lambda_e^2+
     2\sum_{e<f}|X_e\cap X_f|\lambda_e\lambda_f.
\]
Half this expression dominates the parenthesis and is at most
$n\lambda^2/2$. A once-used bath is untouched by all later gates, so its
partial trace commutes through them. Iterating identifies the actual
forward circuit's end-of-circuit trace with the product of the local
channels just bounded, also with an untouched reference. The localized circuit instead retains its shared baths and work
registers throughout the segment.
CPTP contractivity then telescopes over $m_{\rm tot}$ periods to
Eq.~\eqref{eq:physicalsharp}. For $h=0$ or $m_{\rm tot}=0$ it is the
identity channel with zero error.
\end{proof}

\section{Finite arithmetic and local gate synthesis}\label{app:arithmetic}
The tree-rotation angles are computed coherently from preserved
address prefixes and the current bath configuration; selected
matrix gates have fixed dimension. Quantized leaf tangents give each
nonidentity merge a known denominator bound, fixing precisions and loop
counts before generation. Local matrix synthesis retains basis order and
the full phase required for controlled use. Our rotation convention is
$R_\mu(\theta):=\exp(-i\theta\sigma_\mu/2)$, $\mu=y,z$.
Throughout this appendix we use the forward-window notation, complete
term list and allocated widths of Lemma~\ref{lem:directfinite}: $L\ge1$, $w=\lceil\log_2L\rceil$; an empty
window returns identity before these routines.

\subsection{Quantization and merge stability}
We write leaf half-angle tangents as
\[
 t_j=\frac{1-u_j^2}{1+u_j^2},\qquad b_j=\frac{2u_j}{1+u_j^2},\qquad 0\le u_j\le1.
\]
For a Hamiltonian occurrence or a jump occurrence on vacuum, $u_j$ is
the square root of a known rational number; on an occupied jump bath,
it is one. For call budget $0<\xi<1$, we set
\[
 b_{\rm q}=\left\lceil\log_2\frac{64\max\{1,L\}}{\xi}\right\rceil,
 \qquad \widetilde u_j=2^{-b_{\rm q}}\lfloor2^{b_{\rm q}}u_j\rfloor.
\]
We determine each static term's vacuum square-root floor once by integer
comparisons, scanning the term list once and reusing its period counts.
The $2$-Lipschitz function $2\arctan u$ gives, for the quantized transducer,
\[
 \norm{S-\widetilde S}\le2L2^{-b_{\rm q}}\le\xi/32.
\]
We compile $\widetilde S$ with its exact two trees and charge this perturbation
to the finite-call error; all catalyst, decay and reuse identities still
refer to the original $S$. Each quantized interval is identity or has
sine at least $2^{-b_{\rm q}}$: a leaf with $\widetilde u_j>0$ has
$\widetilde b_j:=2\widetilde u_j/(1+\widetilde u_j^2)
\ge\widetilde u_j\ge2^{-b_{\rm q}}$, and interval sine dominates
each leaf sine. An identity interval is recognized exactly from the term counts and
occupied-label membership tests.
Real leaves retain their selected block-encoding gate and private
phase even after rounding; padded leaves act as identity on these
coordinates.

\paragraph{Full-register arithmetic action.}
The fixed-width, fixed-loop evaluator validates descriptors, clips to
$[1,L]$ and preserves validation until uncomputation. Invalid descriptors
return zero angle. We put $m:=2^{-2b_{\rm q}}$. A computed denominator
below $m/2$ is replaced by one before division, with a retained guard
selecting zero angle; the precision below excludes this branch at every
valid nonidentity node. The retained inputs and guard bits make the same finite computation
reversible on the remaining strings.

For a nonidentity merge with quantized child cosines $a,c$ and approximations
$\widehat a,\widehat c\in[0,1]$ within $\tau_{\rm q}$,
\[
 d=1-a^2c^2\ge m,\qquad
 x_{\rm in}=\frac{1-c^2}{d},\qquad
 x_{\rm out}=\frac{a^2(1-c^2)}{d}.
\]
We choose the largest power of two satisfying
\begin{equation}\label{eq:explicitnodeprecision}
 \tau_{\rm q}\le\frac{m^2\xi^2}{2^{20}(w+1)^2},
\end{equation}
by integer comparison with the rational right side. Then
$\tau_{\rm q}\le m/8$, and $\widehat d=1-\widehat a^2\widehat c^2$ obeys
\[
 |d-\widehat d|\le4\tau_{\rm q},\qquad \widehat d\ge m/2.
\]
Either numerator satisfies $|n-\widehat n|\le4\tau_{\rm q}$ and
$0\le n,\widehat n\le1$, giving
\[
 \left|\frac nd-\frac{\widehat n}{\widehat d}\right|
 \le\frac{4\tau_{\rm q}}{m}+\frac{8\tau_{\rm q}}{m^2}
 \le\frac{12\tau_{\rm q}}{m^2}.
\]
We evaluate the polynomial numerators and denominators at the hatted
values exactly, using at most four times the child fractional width. Only the final restoring division is
rounded: allocate $b_{\rm div}=\lceil\log_2(1/\tau_{\rm q})\rceil$
fractional and $2b_{\rm q}+2$ integer bits for a quotient at most $2/m$.
Downward rounding and clipping to $[0,1]$ give probability error at most
$13\tau_{\rm q}/m^2\le32\tau_{\rm q}/m^2$.
The modulus
\[
 |2\arcsin\sqrt{x}-2\arcsin\sqrt{y}|
 \le\pi\sqrt{|x-y|},\qquad x,y\in[0,1],
\]
therefore gives per-merge angular error
$\pi\sqrt{32}\,\xi/[2^{10}(w+1)]$. Each level is block diagonal in the preserved address prefix and bath
configuration, so its operator error is the maximum node error.
There are $w$ levels per tree and root probability error at most
$2\tau_{\rm q}$; hence total angular error is
\[
 2w\,\frac{\pi\sqrt{32}\,\xi}{2^{10}(w+1)}
 +\pi\sqrt{2\tau_{\rm q}}
 \le\frac{\pi\xi}{2^{10}}
       \max\{2\sqrt{32},\sqrt2\}<\xi/16.
\]
An $R_y$ operator error is at most half its angular error. We reserve another
$\xi/16$ for finite angle evaluation across $2w+1$ levels and $\xi/8$
for $U_K$ and its saved adjoint in $O=U_K(Z\otimes I)U_K^\dagger$
(Appendix~\ref{app:endpoint}). With the $\xi/32$ leaf perturbation,
this leaves rounding slack below total call error $\xi$.
All estimates are uniform on valid bath configurations; the direct-sum
norm gives the same bound on coherent superpositions with one shared
quantization and precision choice.

\subsection{A fixed-round angle algorithm}
For each tree rotation, fix
\[
 \delta_\theta=\frac{\xi}{64(2w+1)},\qquad
 B_\theta=\left\lceil\log_2\frac{256}{\delta_\theta^2}\right\rceil,\qquad
 b_\theta=\left\lceil\log_2\frac{32}{\delta_\theta}\right\rceil.
\]
Choose the largest power of two
$\rho\le\tau_{\rm q}/[8(P+1)2^{w+1}]$ and put
$b_c=\log_2(1/\rho)$; rounded child products use this grid.
The rounded merge probability has $b_{\rm div}$ fractional bits,
whereas the root probability $1-\widehat a^2$ has at most $2b_c$.
Both also need one bit to represent the endpoint one. Define
\[
 b=2+\max\left\{w_*,b_c,b_{\rm div}+2b_{\rm q}+2,
 B_\theta,b_\theta,
 \left\lceil\log_2(P_{\rm cat}+C_{\rm cap}+1)\right\rceil\right\}.
\]
Here $w_*$ is the allocated-width bound of Lemma~\ref{lem:directfinite}.
All numerical operand and scratch words below have $O(b)$ bits,
including sign, integer, carry and remainder fields. We denote by
$b_{\rm arith}\ge b$ a common complete word width obtained from the
specified widths. The estimates below give $b_{\rm arith}=O(b)$;
the total number of retained Boolean work bits is counted separately.

The merge evaluator returns a squared rotation coefficient, which must
be converted coherently to an angle. Its inverse derivative diverges
near probabilities zero and one, so we use interval comparisons with
the monotone function
$F(x):=\sin^2(\pi x/2)$ on $[0,1]$. The exact probabilities
zero and one return $x=0$ and $x=1$, respectively. For an interior
probability, the evaluator returns at each bisection midpoint a
certified interval of width
$\delta_\theta^2/256$, where $\delta_\theta$ is the angular
error budget. A probability strictly below or above that interval
determines the corresponding half of the bracket. If it lies inside,
the algorithm returns the midpoint: the preceding modulus bounds this
case by $\pi\delta_\theta/16$, so resolving the comparison further
would serve no purpose. Otherwise the algorithm stops once the bracket
width is at most $\delta_\theta/8$ and returns its midpoint.
There are at most
$1+\lceil\log_2(8/\delta_\theta)\rceil$ rounds, and a preserved
halt bit makes any remaining rounds controlled identities. This also
specifies the action at exact midpoint equality.

Efficient high-precision evaluation of elementary functions is standard
\cite[Theorems~7.1--7.2]{BrentElementary}. To obtain fixed finite circuits, we use power series with argument
reduction
\cite[Section~4.4.2]{BrentZimmermann}; Lemma~\ref{lem:boundedsine}
specifies fixed integer loop counts, directed rounding and a certified
interval output, together with bounds on Boolean gates and retained
intermediate bits.

\begin{lemma}[A bounded sine enclosure]\label{lem:boundedsine}
For an integer $B\ge1$ and a dyadic $x\in[0,1]$ with $O(B)$ bits, a fixed
Boolean network returns a dyadic interval containing
$F(x)=\sin^2(\pi x/2)$ and of width at most $2^{-B}$. Its size and
number of intermediate bits are $O(B^{5/2})$.
\end{lemma}
\begin{proof}
The enclosure is constructed in three stages: power-of-two argument
reduction and truncated sine/cosine series, a fixed number of doublings
to recover the original angle, and outward rounding of the squared
sine interval. The parameters below balance the series remainder with
the rounding error accumulated during those doublings.
We use the following integer parameters:
\begin{align*}
 s&=\max\{2,\lceil\sqrt{B+1}\rceil\},\qquad
 m=\left\lceil\frac{B+2s+16}{2(s-1)}\right\rceil,\\
 W&=B+2s+\lceil\log_2(4m+8)\rceil+16,\qquad d=2^{-W}.
\end{align*}
Here $s$ is the number of angle doublings and $m$ is the Taylor
degree parameter. The square root in the definition of $s$ is an
integer ceiling, computed by integer comparisons. Rounded sine and cosine approximations use two's-complement registers with $W$ fractional bits and
four integer bits, including the sign, so their range is $[-8,8)$.
The fixed classical constant $\widehat\pi$ may use $W+1$ fractional
bits. The rounding map $\mathcal R_d$ rounds downward to the
$d$-grid, with error strictly below $d$; additions of these fixed-point values
are exact. Exact products use separate scratch registers before
the prescribed rounding, with widths specified below.

The classical compiler first computes a dyadic $\widehat\pi$ with
$|\widehat\pi-\pi|\le d$. A finite rational construction uses
Machin's identity
$\pi=16\arctan(1/5)-4\arctan(1/239)$ and the first $N$ terms of
each alternating arctangent series, with
$N=\lceil(W+10)/4\rceil$. The combined remainder is at most
$20\,5^{-(2N+1)}<d/4$. Rounding the rational result to the
$d/2$-grid gives the required approximation. These constants are
classical preprocessing data; each use loads their fixed dyadic bits.
The identity follows from the tangent addition formula and the
principal ranges of the two positive arctangents.

With clipping to the indicated intervals, the program sets
\[
 \widehat z=\operatorname{clip}_{[0,2]}
       \mathcal R_d(\widehat\pi x/2),\qquad
 \widehat t=\operatorname{clip}_{[0,2^{1-s}]}
       \mathcal R_d(\widehat z/2^s),\qquad
 q=\mathcal R_d(\widehat t^2).
\]
The first value differs from $\pi x/2$ by at most $2d$, and the
second differs from $\pi x/2^{s+1}$ by at most $2d$. In particular,
$0\le\widehat t\le1/2$ and $0\le q\le1/4$.

Starting with $C_0=1$ and $S_0=\widehat t$, the program performs
exactly $m$ steps, indexed by $j=0,\ldots,m-1$:
\[
 C_{j+1}=\mathcal R_d\!\left(
       \frac{-\mathcal R_d(C_jq)}{(2j+1)(2j+2)}\right),\qquad
 S_{j+1}=\mathcal R_d\!\left(
       \frac{-\mathcal R_d(S_jq)}{(2j+2)(2j+3)}\right).
\]
It sums these terms exactly and clips both sums to $[-1,1]$.
For comparison with the exact Taylor terms at $\widehat t$, each
recurrence has error at most one eighth of the preceding term error
plus $3d$. Thus every term error is at most $4d$, and either sum has
rounding error at most $4(m+1)d$.
The alternating-series remainder for either function is at most
\[
 \widehat t^{2m+2}\le2^{-2(s-1)(m+1)}.
\]
Together with the initial argument error, this bounds the two-coordinate
error relative to the sine and cosine of $\pi x/2^{s+1}$ by
$(4m+6)d+2^{-2(s-1)(m+1)}$. Clipping preserves this bound.
All partial sums and rounded fixed-point values in this Taylor
calculation lie in $(-4,4)$, so the allocated integer bits suffice.
Scaled integer numerators and the explicit divisors use the separate
scratch registers below.

Writing $(c_0,v_0)$ for the two clipped sums, the program now makes
exactly $s$ angle-doubling steps:
\[
 c_{j+1}=\operatorname{clip}_{[-1,1]}
       \bigl(\mathcal R_d(c_j^2)-\mathcal R_d(v_j^2)\bigr),\qquad
 v_{j+1}=\operatorname{clip}_{[-1,1]}
       \bigl(2\mathcal R_d(c_jv_j)\bigr).
\]
On $[-1,1]^2$, the map $(c,v)\mapsto(c^2-v^2,2cv)$ is
4-Lipschitz in the maximum norm. Each rounded step adds at most $2d$.
Consequently the final coordinate error is at most
\[
 4^s\left((4m+7)d+2^{-2(s-1)(m+1)}\right)
 <2^{-(B+15)}.
\]
Indeed, the choice of $W$ gives
$4^s(4m+7)2^{-W}<2^{-(B+16)}$. The choice of $m$ also gives
$2(s-1)m\ge B+2s+16$, and hence
\[
 4^s2^{-2(s-1)(m+1)}
 \le2^{-(B+2s+14)}\le2^{-(B+18)}.
\]
Their sum is smaller than the displayed error bound.
The interval of radius $2^{-(B+12)}$ around $v_s$, intersected with
$[-1,1]$, therefore contains $\sin(\pi x/2)$. Its squared range has
width at most $2^{-(B+10)}$. Computing that range and rounding its
endpoints outward to the $d$-grid adds at most $2d$, leaving width
less than $2^{-B}$ as required.

For an input with $K$ fractional bits, the compiler stores its
unsigned numerator in $K+1$ bits and the constant numerator
$H=2^{W+1}\widehat\pi$ in $W+3$ bits. Their full product fits in
$W+K+4$ bits. Products and endpoint squares of two signed fixed-point values use $2W+8$ bits. The largest Taylor divisor is
$(2m)(2m+1)$, so all divisors fit in
$\lceil\log_2((2m)(2m+1)+1)\rceil$ bits. Quotient, remainder,
carry and rounding logic are included in the retained Boolean
intermediate outputs. The squared-range minimum is zero when its
input interval crosses zero, and otherwise is the smaller endpoint
square; these cases use fixed comparisons and selections.

There are $O(m+s+1)=O(\sqrt B)$ fixed-width arithmetic steps, and
$W=O(B)$. Since $K=O(B)$, all operand and scratch widths are
$O(B)$. Each multiplication or restoring division by a specified
integer uses $O(B^2)$ Boolean gates by schoolbook arithmetic
\cite[Chapter~1]{BrentZimmermann}; shifts, clipping and additions are
smaller. Division denominators are positive explicit integers with
$O(\log(m+1))$ bits. All loop counts, wire widths and rounding
operations have been fixed above. The resulting Boolean network and
its stored intermediate outputs therefore have size $O(B^{5/2})$.
Lemma~\ref{lem:booleanframe} gives its coherent
compute--copy--uncompute implementation.
\end{proof}

At each bisection midpoint we apply Lemma~\ref{lem:boundedsine}
with $B=B_\theta$. There are $O(B_\theta)$ rounds, so each midpoint
has $O(B_\theta)$ bits. The target probability can have more bits:
we compare it exactly with the returned dyadic endpoints after
padding to a common denominator, using $O(b)$ bits and gates.
One round therefore costs $O(b^{5/2})$ Boolean gates and retained
intermediate bits, and the complete inversion costs $O(b^{7/2})$.
This is within the $O(b_{\rm arith}^{6})$ allowance used below.
A preserved halt bit pads the remaining rounds after an early return.
The inputs remain unchanged while the output controls the rotation,
after which the stored evaluator is reversed.

For a rational target probability, the word width additionally covers
its numerator and denominator; exact comparisons use $O(b^2)$-gate
cross multiplication, leaving the same inversion bound. The rational
subtree probabilities in Appendix~\ref{app:reuse-registers} have
$O(q+\log(q+1))$-bit numerators and denominators; their corresponding
common width is $O(q+\log(1/\delta_\theta)+1)$.

Approximating $\pi x$ and the final
angle each to $\delta_\theta/8$ adds at most $\delta_\theta/4$.
Both the midpoint-enclosure return and the final-bracket return have
angular error at most $\pi\delta_\theta/16$, so the complete bound is
\[
 (\pi/16+1/4)\delta_\theta<\delta_\theta.
\]
The induced one-qubit operator error is consequently below
$\delta_\theta/2$.

It remains to compute the child products to the accuracy assumed
above. We keep all intermediate values in $[0,1]$ by clipping and
round each initial term cosine and multiplication to the chosen
$\rho$-grid. Square-and-multiply at
exponent at most $2^w$ has error at most $3\,2^w\rho$; multiplying
at most $P$ powers of term cosines is therefore accurate within $\tau_{\rm q}$.
The one-leaf case is evaluated directly. For the fixed term list and
allocated widths, the precision choices give
$b_{\rm q},b_{\rm div},b_c,B_\theta,b_\theta
=O(w_*+\log(1/\xi)+\log(P_{\rm cat}+C_{\rm cap}+1))$.
The fixed-degree products, division shifts and signed sine words
increase these widths by only constant factors. Thus the common
complete arithmetic word width satisfies
\begin{equation}\label{eq:bitwidth}
 b_{\rm arith}=O\bigl(w_*+\log(1/\xi)+\log(P_{\rm cat}+C_{\rm cap}+1)\bigr).
\end{equation}
The selected angle width satisfies
$4\,2^{-b_\theta}\le\delta_\theta/8$. The angle-bit circuit in
Eq.~\eqref{eq:angle-register-action} therefore implements the rounded
angle with the prescribed accuracy. Its preserved descriptor and
occupation inputs allow the evaluator to be reversed after the
controlled rotations, as specified in Section~\ref{sec:implementation}.

\subsection{Constant-dimensional local unitaries}\label{app:endpoint}
The selected block encodings require local matrix gates in a fixed
computational-basis order. We preserve the order supplied by the support
list through the numerical evaluator and wire placement. Their overall
phase is also retained, because controlling a local unitary by its term
label turns that phase into a relative phase between branches.

The ordered support list fixes tensor-factor order, with increasing qubit
index within each site. A local matrix on $k$ qubits uses the basis ordinal
$\operatorname{ind}(x)=\sum_{j=0}^{k-1}2^{k-1-j}x_j$ for
$\ket{x_0}\otimes\cdots\otimes\ket{x_{k-1}}$.
For a placement $f$ into the system wires, the embedded matrix has entries
\begin{equation}\label{eq:matrix-placement}
 A[f]_{x,y}=A_{\operatorname{ind}(x\circ f),\operatorname{ind}(y\circ f)}
 \prod_{t\notin\operatorname{im}f}\delta_{x_t,y_t}.
\end{equation}
Arithmetic registers use little-endian binary integers. Accordingly,
the local gate compiler lists the $k$ system wires in reverse support
order when its numerical matrix evaluator uses little-endian ordinals;
this leaves Eq.~\eqref{eq:matrix-placement} unchanged. Bath labels keep
their separate prescribed coordinates. This convention applies to both
$H_e$ and each $L_{e\mu}$, and hence to their Cayley transforms.

The local matrix includes the system wires, the extracted bath value,
and the block-encoding ancilla. To specify its placement, let
$f(0),\ldots,f(k-1)$ be the ordered system wires in
Eq.~\eqref{eq:matrix-placement}. For a jump term, we write
$r_E:=\log_2d_E$ for the width of the extracted bath-value register
$\mathsf F$. Let $\mu_t$ be the bit stored in
$\mathsf F_t$, and write $\mu=\sum_{t=0}^{r_E-1}2^t\mu_t$ for
the extracted bath value. If $c_Q\in\{0,1\}$ is the value of
$\mathsf Q$, the matrix uses the ordinal
\begin{equation}\label{eq:endpoint-order}
 \operatorname{ind}_{\rm end}(c_Q,\mu,x)
 =c_Q 2^{k+r_E}+\mu 2^k+\operatorname{ind}(x\circ f).
\end{equation}
A numerical gate evaluator with little-endian indexing therefore uses the operand list
\[
 [f(k-1),\ldots,f(0),\mathsf F_0,\ldots,\mathsf F_{r_E-1},\mathsf Q].
\]
For a Hamiltonian term the bath register is absent ($r_E=0$, $\mu=0$). This realizes the original term list matrices with the bath and ancilla blocks in their stated order.

For $K=K^\dagger$ with $\norm K\le1/2$, define
\[
 A_K:=\sqrt{(I+K)/2},\quad B_K:=\sqrt{(I-K)/2},\quad
 U_K:=\begin{pmatrix}A_K&-B_K\\B_K&A_K\end{pmatrix}.
\]
The commuting functions $A_K,B_K$ satisfy $A_K^2+B_K^2=I$, so
$U_K\in SU(d)$ in its fixed dimension and
$O=U_K(Z\otimes I)U_K^\dagger$. Both square-root arguments have spectrum
in $[1/4,3/4]$. We approximate these matrix functions by their convergent power series \cite{Higham}: the binomial series and $\norm K\le1/2$ give
\[
 \sqrt{I\pm K}=\sum_{m=0}^{M}\binom{1/2}{m}(\pm K)^m+E_M,
 \qquad \norm{E_M}\le2^{-M}.
\]
With certified scalar $1/\sqrt2$, $M=O(p)$ yields $p$ entry bits in
polynomial classical time \cite{BrentZimmermann}.

\begin{lemma}[Finite-precision local decomposition]\label{lem:certifiedendpoint}
For a fixed $d\ge2$, let $U\in SU(d)$ have a deterministic polynomial-time
entry-interval algorithm. For rational $0<\zeta\le1$, one can construct
$O_d(1)$ two-level rotations with angles specified to
$O_d(1+\log(1/\zeta))$ bits. Their product $\widetilde U\in SU(d)$
satisfies $\norm{U-\widetilde U}<\zeta/4$, including its phase.
The classical time is polynomial in the input length and accuracy bits.
\end{lemma}
\begin{proof}
Our starting point is the two-level decomposition reviewed by
\citet[Section~8]{Barenco}, together with the local gate constructions
of \citet{SBM}. An entry-interval algorithm gives an approximate pivot, which we
compare with a dyadic threshold: with $N=d(d-1)/2$, we choose
$\zeta/(128d^3)<\tau\le\zeta/(64d^3)$ and set $\eta=\tau/8$.
A pivot below this scale can be skipped within the error budget, whereas
a retained pivot has a known lower bound for division.

We process the below-diagonal entries columnwise, from bottom to top.
For rows $(j,k)$ of the current matrix $B$, the entry algorithm supplies
dyadic centres $a,b$ of $B_{jj},B_{kj}$ within $\eta$. If
$|a|^2+|b|^2\le\tau^2$, the algorithm chooses identity; otherwise
it chooses the exact rotation
\[
 G(a,b):=\frac1{\sqrt{|a|^2+|b|^2}}
 \begin{pmatrix}\bar a&\bar b\\-b&a\end{pmatrix}.
\]
In the second case the divisor exceeds $\tau$, and subsequent entry
enclosures track the exact unitary chosen by this rotation. Either
choice leaves a residual of at most $2\tau$. Later row operations
preserve the earlier column-tail norms, so the final lower triangle has
squared Frobenius norm at most $4N\tau^2$. To control the upper
triangle as well, we use unitarity: across each diagonal cut, the squared
masses in the two directions are equal. Summing over these cuts bounds
the total off-diagonal Frobenius norm by
$e\le2\sqrt{dN}\tau\le2d^2\tau$.

This bound gives $|B_{ii}|>1/2$. For analysis, we set
$z_i:=B_{ii}/|B_{ii}|$ and $D_0:=\operatorname{diag}(z_1,\ldots,z_d)$.
The column norms give
$1-|B_{ii}|\le\sum_{j\ne i}|B_{ji}|^2$.
Since $e<1$, the diagonal correction and determinant multilinearity yield
\begin{align*}
 \norm{B-D_0}&\le e+e^2\le2e,\\
 |\det D_0-1|&\le d\norm{D_0-B}\le2de.
\end{align*}
We choose $D_*$ with first $d-1$ entries $z_1,\ldots,z_{d-1}$ and
last entry $(z_1\cdots z_{d-1})^{-1}$. This changes the last entry
of $D_0$ by $|\det D_0-1|$, so
$\|B-D_*\|\le2(d+1)e$. For the finite construction, we evaluate the
final diagonal entries of $B$ to obtain dyadic centres $w_i$ with
$|w_i-B_{ii}|\le\eta$, and set $\widehat z_i:=w_i/|w_i|$ for $i<d$.
The lower bound $|B_{ii}|>1/2$ ensures that these centres are nonzero and
gives $|\widehat z_i-z_i|\le4\eta$. We form $\widehat D$ with first
$d-1$ entries $\widehat z_i$ and last entry
$(\widehat z_1\cdots\widehat z_{d-1})^{-1}$. Thus
$\widehat D\in SU(d)$ exactly, and telescoping the product of its phases
gives $\|\widehat D-D_*\|\le4(d-1)\eta$. These estimates give
\[
 \|U-G_1^\dagger\cdots G_N^\dagger\widehat D\|
 \le2(d+1)e+4(d-1)\eta\le7d^3\tau<\zeta/8.
\]
For a two-level unitary with first column $(\alpha,\beta)^{\mathsf T}$,
the phase-exact factorization is
$R_z(v-u)R_y(\theta)R_z(-u-v)$, where
$\theta:=2\operatorname{atan2}(|\beta|,|\alpha|)$,
$u:=\arg\alpha$ and $v:=\arg\beta$ for nonzero components.
For a zero component, the algorithm sets its phase parameter to zero.
When $\beta=0$, the rotation has $\theta=0$ and is independent of
$v$; when $\alpha=0$, it has $\theta=\pi$ and is independent of
$u$. Thus this convention also gives the stated unitary in both
degenerate cases. For these rotations, the only exact zero tests concern the selected
dyadic numerators; every nonzero one has a known dyadic lower bound.
We evaluate arguments in $(-\pi,\pi]$ and retain the resulting Euler
angles with their full phases. There are at most $Q=3N+d-1$ one-axis
rotations, so computing each angle within $\zeta/(4Q)$ adds at most
$\zeta/8$ to the operator error \cite{BrentZimmermann}. Combined
with the elimination estimate, this proves the claimed approximation.
Finally, the fixed number of matrix operations and the lower bounds on
all nontrivial divisors make the guard precision and classical cost
polynomial in the stated input and accuracy parameters.
\end{proof}

The two-level rotations are lowered on the fixed local qubit register.
For distinct basis strings $x,y$, X gates send $x$ to $0$; CNOTs from
one set bit of $x\oplus y$ send the other string to that single-bit
basis vector. Conjugating a rotation controlled on all other bits being
zero by this affine permutation gives the desired two-level gate.
Standard controlled-rotation and parity decompositions
\cite{Barenco,SBM} then use a constant number of one-/two-qubit gates.
Here the fixed local dimension bounds the number of controls. We use the complete controlled-unitary
circuits of \citet[Section~7]{Barenco}, after the stated basis permutation;
their action is defined also on the unused local basis states.

The angles for these local unitaries are classical constants. The
quantum circuit therefore scans the
local term list, computes its equality guard, applies the selected
stored gate sequence, and reverses the guard. These guards partition the preserved
term label, so the selected local operator error is the maximum error
among terms,
whereas the emitted gate count includes every scanned term. In
$U_K(Z\otimes I)U_K^\dagger$, both appearances use the same synthesized
$U_K$ and its adjoint circuit. Their error is at most twice the single-circuit
error. Choosing $\zeta=\xi/4$ makes this contribution less than $\xi/8$.
All remaining $i$ and sign phases are exact gates. Finite-alphabet
synthesis is applied once to the completed circuit in
phase-sensitive operator norm.

\subsection{Controls and full-register unitary extensions}\label{app:controls}
Before terminal synthesis, lower controlled rotations into target-only
rotations and exact Boolean/CNOT gates \cite{Barenco,SBM}. For $\mu=y,z$,
the chronological sequence $R_\mu(\theta/2)$, CNOT,
$R_\mu(-\theta/2)$, CNOT is phase-exact
\cite[Lemma~5.1]{Barenco}; we use their Corollary~5.3 for remaining fixed
one-qubit controls. Controlled-CNOT uses the exact Toffoli sequence below.
Controlled named phases compute the control--target conjunction in a
separate zero helper, apply that phase and uncompute. We compute each wide
enable once and preserve guard bits throughout its stored circuit.
This gives the intended action on zero helpers and a full-register saved
unitary as in Section~\ref{sec:model}. The completed circuit contains
exact $H,X,Z,S,S^\dagger,T,T^\dagger,\mathrm{CX}$ and target-only $SU(2)$
rotations; every remaining continuous gate has determinant one.

\section{Uniform finite-gate implementation}\label{app:program}
The classical compiler fixes common descriptors, widths and gate order,
expands Boolean operations and controls, and routes local gates. It then
assigns synthesis accuracies to the remaining one-qubit rotations using
their expanded occurrence count, including repeated and inverse calls.

\subsection{Reuse registers and coefficient preparations}\label{app:reuse-registers}
We implement the selected operators of Section~\ref{sec:finite-reuse}
on the same system and bath registers in every length branch.

\paragraph{Reuse register.}
The flag $b=0/1$ denotes public/private, and $\mathsf M$ stores the public
slot index, zero in the private space. For slot $k$ and data $z$, use
\[
 R_k\ket{b,i,z}:=\ket{b,i\mathbin{\mathrm{xor}}bk,z}.
\]
We apply $R_k$, compute $i=k$, apply the saved $S$ under this predicate,
clear it, then apply $R_k$ again. Public $(0,k)$ and private $(1,0)$
both have index $k$ during $S$, which preserves it; the second permutation
returns outputs to their respective slot/private encoding, while other
public slots bypass. A preserved length adds guards $1\le N\le10q$,
$k\le N$. We use
\[
 F_N^{\pub}:=\ket{\pub}\bra{\pub}\otimes F_N+
 \ket{\priv}\bra{\priv}\otimes I_{\mathsf M},\qquad
 F_N\ket0=\frac1{\sqrt N}\sum_{k=1}^{N}\ket k.
\]
For $1\le N\le10q$, Lemma~\ref{lem:imported-preparation} prepares this
explicit vector directly on $w_M=\lceil\log_2(10q+1)\rceil$ qubits;
for other lengths we set $F_N=I$. The finite controlled family costs
$O(q^2(1+\log(q+1)))$: $O(q)$ branches each use $O(q)$ preparation
gates, with $O(1+\log(q+1))$ equality/public guards per branch. Then
\[
 \mathscr U_N:=(F_N^{\pub})^\dagger U_{10q}\cdots U_1F_N^{\pub},
\]
where $U_k$ is the slot call with the above enable. This fixed product
handles superposed lengths, with the private component in every enabled slot.

\paragraph{Coefficients, signs and helpers.}
We reverse the recursive disentangling of Shende, Bullock and
Markov~\cite[Theorem~9]{SBM} and lower its
multiplexed rotations by their Theorem~8. Subtree masses fix angles as in
\citet[Eq.~(8)]{Mottonen}, zero on a zero-mass subtree. A classical binomial
recurrence evaluates the $O(q)$ rational probabilities in
Eq.~\eqref{eq:coefficient-prepare}, with $O(q+\log q)$-bit
numerators and denominators. We use exact
rational subtree ratios and Appendix~\ref{app:arithmetic}'s fixed-round
sine inversion, including endpoints, to set angles before emission.
The workspace of each preparation circuit is included in the allocated
work registers. We lower controls to target-only
rotations, budgeting every angle occurrence in full-unitary norm,
including unused input columns. We compute $2q\le N\le4q$, apply $Z$ to
its helper and uncompute for the selection sign. Each branch preserves
$N$; Eq.~\eqref{eq:prepareselect-full} holds on initialized helpers and
the saved gates define the full unitary on all their states.

\subsection{Classical descriptors and circuit composition}
\paragraph{Descriptors and arithmetic.}
A descriptor $(\mathcal D,h,m_0,m_1,\pi)$ gives the active list, time
interval and map to permanent term labels, owners and block ordinals.
We validate $\pi$, set $P=|\mathcal D|$, $L=P(m_1-m_0)$ and use
Eq.~\eqref{eq:rankdecoder}. The shared layout gives capacities, actual
widths and disjoint wires, with fixed input-term and lexicographic owner
order. We use little-endian integers, downward dyadic rounding and Euclidean
division. The validity test reads every allocated address bit and requires unused
high bits to be zero; other strings follow the specified padding action.

We use ripple-carry addition \cite{VBE,Cuccaro}, schoolbook multiplication,
restoring division and bitwise integer roots \cite{BrentZimmermann},
with conservative $O(b^3)$ Boolean cost on $b$-bit arguments. Rational
root tests use integer cross-multiplication. For negative $x$ and positive
$d$, write $|x|=qd+r$: Euclidean division returns $(-q,0)$ for $r=0$
and $(-q-1,d-r)$ otherwise. Lemma~\ref{lem:booleanframe} supplies coherent
compute--copy--uncompute \cite{Bennett}. We simplify equal-input ANDs so
Toffoli controls/targets are distinct.

We use the six-CNOT Toffoli of Shende and
Markov~\cite[Figure~1]{ShendeMarkov}. For distinct
controls $a,b$ and target $t$, its chronological fifteen-gate word is
\begin{align*}
 &H_t,\ \mathrm{CX}_{b,t},\ T_t^\dagger,\
       \mathrm{CX}_{a,t},\ T_t,\ \mathrm{CX}_{b,t},\ T_t^\dagger,\\
 &\mathrm{CX}_{a,t},\ T_b,\ T_t,\ H_t,\
       \mathrm{CX}_{a,b},\ T_a,\ T_b^\dagger,\ \mathrm{CX}_{a,b}.
\end{align*}
With Eq.~\eqref{eq:sk-alphabet}'s exact phases it maps
$\ket{a,b,t}$ to $\ket{a,b,t\oplus ab}$ with phase one on every basis
state. A SWAP consists of three CNOTs. Thus every Boolean gate has a specified
elementary decomposition before routing.

\paragraph{Composition.}
Equations~\eqref{eq:direct} and \eqref{eq:regroup} fix transducer and
window order. We choose shared capacities, widths and $W_{\max}$ before
lattice refinement, and expand controls/routing before assigning
Eq.~\eqref{eq:sk-error}'s accuracy. For an empty active term list, the compiler returns the identity before
dividing by $P$. Equation~\eqref{eq:actualchannel} includes segment
auxiliary preparation and discard.

\subsection{Repeated subcircuits and error accumulation}
Saved subcircuits form a finite acyclic call graph; inverse calls reverse
and adjoint the same children. Gate counts and error bounds include every occurrence after the calls
are expanded.
\begin{lemma}[Multiplicity accounting]\label{lem:multiplicity}
Suppose rotation type $\ell$ appears $n_\ell$ times in the expanded
circuit, including inverse occurrences, and its one saved replacement
has error $\eta_\ell$ and length $s_\ell$. The expanded circuit satisfies
\[
 \norm{\widehat C-C}\le\sum_\ell n_\ell\eta_\ell,
 \qquad G=\sum_\ell n_\ell s_\ell+G_{\rm exact}.
\]
Here $G_{\rm exact}$ includes all expanded exact control and routing gates.
\end{lemma}
\begin{proof}
We expand the acyclic description in chronological order. A saved
replacement and its adjoint have the same norm error and gate length,
so every occurrence of type $\ell$ contributes at most
$\eta_\ell$ to the unitary hybrid and exactly $s_\ell$ gates.
Lemma~\ref{lem:refinementhybrid} therefore gives the first sum, while
direct counting gives the second. Each reference to a saved sequence contributes its expanded occurrences
to both sums.
\end{proof}
\subsection{Finite gate synthesis}\label{app:terminal}\label{sec:explicit-terminal-sk}
After control decomposition and routing, the remaining gates to be
synthesized are uncontrolled one-qubit rotations with specified angles. We approximate them in operator
norm, including phase. We fix the inverse-closed alphabet
\begin{equation}\label{eq:sk-alphabet}
 \mathcal G:=\{H,X,Z,S,S^\dagger,T,T^\dagger,\mathrm{CX},
 A_x,A_x^\dagger,A_z,A_z^\dagger\},\qquad
 A_p:=\frac{3I-4i\sigma_p}{5}\quad(p=x,z).
\end{equation}
Here $X=\sigma_x$, $Z=\sigma_z$, and the remaining one-qubit gates are
\[
 H=\frac1{\sqrt2}\begin{pmatrix}1&1\\1&-1\end{pmatrix},\qquad
 S=\operatorname{diag}(1,i),\qquad
 T=\operatorname{diag}(1,e^{i\pi/4}).
\]
The ordered CNOT has action
$\mathrm{CX}\ket{a,b}=\ket{a,a\oplus b}$; either orientation on
an adjacent pair is allowed. Toffoli and SWAP
refer to their already specified finite decompositions.
The four $A$ gates densely generate $SU(2)$. The unit complex number
$(3+4i)/5$ has infinite multiplicative order: finite order would make
its sum with its conjugate, $6/5$, a rational algebraic integer,
contradicting its nonintegral value.
Powers around each of the two axes are dense in their rotation circles,
and Euler decomposition generates $SU(2)$. All their entries are algebraic.

We use the finite-precision $SU(2)$ synthesis theorem of
van Melkebeek and Watson~\cite[Theorem~4.5]{VMW}, with operator-norm error rather than distance
up to a global phase.

\begin{lemma}[One-qubit synthesis]\label{lem:terminal-synthesis}
For a polynomial-time, certified entry algorithm for $U\in SU(2)$
and rational $0<\xi<1$, a deterministic classical algorithm produces
an $A$-gate sequence $\widehat U$ with $\|\widehat U-U\|\le\xi$ and
length $O((1+\log(1/\xi))^4)$. It requests $O(1+\log(1/\xi))$
input precision and polynomial classical time in this precision and
the entry-evaluation cost.
\end{lemma}
\begin{proof}
We apply van Melkebeek and Watson~\cite[Theorem~4.5]{VMW} at $d=2$ to the inverse-closed algebraic
$SU(2)$ subset $\{A_x,A_x^\dagger,A_z,A_z^\dagger\}$. To construct its
fixed coarse net, we choose dyadic $0<\eps_0<1/16$ below the theorem's coarse
accuracy and mesh $\eta_0\le\eps_0/32$ in $[-1,1]^4$.
Unit quaternions satisfy
$U(q)=q_0I+i(q_x\sigma_x+q_y\sigma_y+q_z\sigma_z)$ and
$\norm{U(q)-U(q')}=\norm{q-q'}_2$. Every unit quaternion is within
$\eta_0$ of a grid point and $2\eta_0$ of its normalization; normalized
nonzero grid points therefore cover $SU(2)$ at radius $\eps_0/16$.

For each grid matrix, the preprocessing algorithm enumerates words over
the four $A$ gates. At stage $t$, test all
words of length at most $t$ with $t$-bit entry intervals, certifying
Frobenius distance below $\eps_0/4$; increase $t$ if any grid point is
uncovered. Each stage is finite, so undecided comparisons cannot block
enumeration. Density supplies a word strictly within this radius at
each point, eventually certified by entry intervals. The finite grid
therefore yields a net of radius below $\eps_0$, with fixed preprocessing cost and word lengths, independent of the input
and requested accuracy.
The synthesis exponent $\log5/\log(3/2)<4$
\cite[Section~4.3.3]{VMW}, \cite{DN} and the theorem's finite-precision
arithmetic give the stated length and classical time.
\end{proof}
Equation~\eqref{eq:sk-error} and Lemma~\ref{lem:multiplicity} charge all
expanded occurrences, including adjoints of the saved replacements,
with one $\log^4X$ synthesis factor in gate count and depth.

\paragraph{Classical running time.}
Validation and dyadic normalization use exact rational operations with
cost polynomial in the raw input length. Integer arithmetic and power comparisons set
segment, mesh, capacity and precision parameters with polynomial-length
binary descriptions. Each descriptor determines bounded scans and loops;
Appendix~\ref{app:arithmetic} bounds their Boolean size in local parameters.
We expand the controls and routing operations, synthesize the resulting
one-qubit occurrences using
Lemma~\ref{lem:terminal-synthesis}, and enumerate the owners, windows and segments,
charging saved calls by Lemma~\ref{lem:multiplicity}. Total generation time
is polynomial in raw input and expanded output-description length,
including the gate enumeration required at large physical times.

\label{end:technical}

\phantomsection\label{start:references}
\begingroup
\small
\setlength{\bibitemsep}{2pt plus 1pt minus 1pt}
\printbibliography
\endgroup
\end{document}